\documentclass[journal,9pt]{IEEEtran}

\usepackage{mathtools}
\usepackage{amssymb, amsthm, amscd, MnSymbol}
\usepackage{upgreek}
\usepackage{bm}
\usepackage{nicefrac}
\usepackage{url}

\usepackage{verbatim}

\ifCLASSOPTIONcompsoc
  \usepackage[caption=false,font=normalsize,labelfont=sf,textfont=sf]{subfig}
\else
  \usepackage[caption=false,font=footnotesize]{subfig}
\fi

\usepackage[inline]{enumitem}

\usepackage{float}
\usepackage{stfloats}
\usepackage[section]{placeins}
\usepackage{balance}

\usepackage{cite}

\usepackage[switch]{lineno}
\usepackage[named]{algo}
\usepackage[noend]{algpseudocode}
\usepackage{algorithm}

\usepackage{threeparttable}
\usepackage{array}
\newcolumntype{P}[1]{>{\centering\arraybackslash}p{#1}}

\usepackage{xcolor}

\usepackage{url}

\usepackage[pdftex]{graphicx}	
\usepackage[none]{hyphenat}

\newtheorem{theorem}{Theorem}
\newtheorem{definition}[theorem]{Definition}

\newtheorem{proposition}[theorem]{Proposition}
\newtheorem{corollary}[theorem]{Corollary}

\theoremstyle{remark}

\begin{document}
%
\title{Co-Pulsing FDA Radar}

\author{Wanghan Lv, Kumar Vijay Mishra, \IEEEmembership{Senior Member, IEEE} and Shichao Chen
\thanks{W. L. and S. C. are with the School of Computer Science and Technology, Nanjing Tech University, Nanjing 211816, China, e-mail: \{lwanghan, 202110006919\}@njtech.edu.cn. }
\thanks{K. V. M. is with The University of Iowa, Iowa City, IA 52246 USA, e-mail: kumarvijay-mishra@uiowa.edu.}
}

\maketitle

\IEEEpeerreviewmaketitle

\begin{abstract}
Target localization based on frequency diverse array (FDA) radar has lately garnered significant research interest. A linear frequency offset (FO) across FDA antennas yields a range-angle dependent beampattern that allows for joint estimation of range and direction-of-arrival (DoA). Prior works on FDA largely focus on the one-dimensional linear array to estimate only azimuth angle and range while ignoring the elevation and Doppler velocity. However, in many applications, the latter two parameters are also essential for target localization. Further, there is also an interest in radar systems that employ fewer measurements in temporal, Doppler, or spatial signal domains. We address these multiple challenges by proposing a \textit{c}o-prime L-shaped FDA, wherein \textit{c}o-prime FOs are applied across the elements of L-shaped co-prime array and each element transmits at a non-uniform \textit{c}o-prime pulse repetition interval (C$^3$ or C-Cube). This \textit{co-pulsing} FDA yields significantly large degrees-of-freedom (DoFs) for target localization in the range-azimuth-elevation-Doppler domain while also reducing the time-on-target and transmit spectral usage. By exploiting these DoFs, we develop \textit{C-C}ube auto-pair\textit{ing} (CCing) algorithm, in which all the parameters are \textit{ipso facto} paired during a joint estimation. We show that C-Cube FDA requires at least $2\sqrt{Q+1}-1$ antenna elements and $2\sqrt{Q+1}-1$ pulses to guarantee perfect recovery of $Q$ targets as against $Q+1$ elements and $Q+1$ pulses required by both L-shaped uniform linear array and L-shaped linear FO FDA with uniform pulsing. We derive Cram\'er-Rao bounds (CRBs) for joint angle-range-Doppler estimation errors for C-Cube FDA and provide the conditions under which the CRBs exist. Numerical experiments with our CCing algorithm show great performance improvements in parameter recovery, wherein C-Cube radar achieves at least $15\%$ higher target hit-rate with shorter dwell time than its uniform counterparts. 
\end{abstract}

\begin{IEEEkeywords}
Co-prime pulsing, frequency diverse array, L-shaped array, parameter estimation, sparse recovery.
\end{IEEEkeywords}

\IEEEpeerreviewmaketitle

\section{Introduction}
During the past several decades, phased array antenna technology has progressed significantly \cite{RJMailloux2018,galati1994advanced,frank2008advanced} and found applications in diverse fields such as radar, sonar, ultrasound, and acoustics \cite{NIGiannoccaro2012,SYang2012}.
The ability of phased arrays to electronically steer a coherent beam toward boresight is useful for tracking weak targets and suppressing strong sidelobe interferences from other directions. Phased array antenna has only angle-dependent beampattern and, as a result, it is used to estimate only direction-of-arrival (DoA) \cite{MElmer2012}. To localize targets in both angle and range, beam-steering should be achieved across the signal bandwidth leading to a complicated waveform design. As an alternative, recently, a new framework of frequency diverse array (FDA) was proposed, wherein a small frequency offset (FO) to the carrier frequency is applied across the array elements \cite{PAntonik2006_1,PAntonik2006_2}, 
resulting in range and angle dependent beampattern. This has been shown to yield a joint estimation of target angle and range parameters \cite{PFSammartino2013,WQWang2014,WQWang2015}. In FDA radars, spatial (DoA) and range resolutions are fundamentally limited by array aperture and maximum frequency increment.

The classical FDA literature has largely focused on a one-dimensional (1-D) uniform linear array (ULA) with linearly increasing uniform FO across the array elements. The properties of 1-D FDA such as the periodicity of the beampattern in range, angle and time domains were introduced in \cite{MSecmen2007,JJHuang2009,THiggins2009}, where the coupling relationship of the beampattern and beam-steering were also derived. Later, for this uniform linear FDA, \cite{WQWang2014_2,JXu2015,RGui2018,FLiu2021} investigated joint DoA and range estimation algorithms. In \cite{WQWang2014_2}, a double-pulse method for range-angle localization was proposed by alternating the antenna between a phased array (zero offset) and FDA (non-zero offset) in subsequent pulses. This approach first estimates the target DoAs using the traditional phased array configuration and then localizes the targets in range domain using FDA. To estimate range and angle at the same time, an unambiguous approach for joint estimation was devised by combining multiple-input multiple-output (MIMO) configuration with FDA \cite{JXu2015}. This FDA-MIMO radar exploits degrees-of-freedom (DoFs) in the range-angle domains. Its estimation accuracy and computational complexity has been shown to improve in a bistatic configuration \cite{FLiu2021}. 
Later works have addressed the degraded beam-focusing ability of FDA-MIMO through approaches such as transmit subaperture FDA radar \cite{RGui2018}. 

Often an exceedingly large number of antennas are required to synthesize a given array aperture to unambiguously distinguish closely-spaced targets. The resulting unacceptably huge size, cost, weight, and area have led to the development of \textit{sparse} arrays, which leverage the presence of a limited number of targets in the scanned region. A uniform \cite{agrawal1972mutual} or random \cite{lo1964mathematical} removal of elements from a \textit{filled} array leads to grating lobes or increased sidelobe levels, thereby reducing the spatial resolution and directivity. However, these issues are mitigated through the use of more structured sparse designs such as co-prime arrays \cite{PPVaidyanathan2011,SQin2015,WZheng2021} which provide a closed form of sensor positions and offer enhanced DoFs for parameter estimation. 

In the context of FDA, introduction of the FOs requires additional bandwidth. Therefore, sparse FDA solutions focus on reducing both spectrum utilization and aperture without any serious degradation in localization performance. Some early FDA works suggested using logarithmic \cite{khan2014frequency}, non-uniform \cite{basit2017beam}, and random \cite{liu2016random} offsets to optimize the available bandwidth for filled FDAs. 
In \cite{SQin2017}, both co-prime arrays and co-prime FOs were introduced for an FDA radar and Bayesian compressive sensing (BCS) was used to jointly estimate angles and ranges. This approach required pre-defined spatial grids leading to a trade-off between gridding error and computational complexity. This \textit{co-prime FDA} was improved in \cite{RCao2021} through a doubly-Toeplitz-based estimation, which incorporated coarray interpolation and off-grid estimation techniques. To mitigate the effect of missing elements or \textit{holes} in the space-frequency coarray, \cite{LLiu2022} introduced a moving time-modulated co-prime FDA, wherein the majority of holes in the coarray positions and frequency offsets could be filled. Recently, co-prime FDA-MIMO has also been investigated to further enhance the accuracy and resolution performance through additional DoFs in polarization \cite{JWang2019} or unfolded structures \cite{CWang2020}. 
In \cite{sedighi2019optimum}, a sparse variant of FDA-MIMO with linear offsets was optimized for an optimal antenna placement. 
Nearly all of the aforementioned sparse geometries and FO designs have been investigated for only 1-D arrays. Some 2-D planar FDA arrays were considered in \cite{XLi2018,CWang2021} for retrieving both azimuth and elevation angles but these configurations are overly complex. 

    \begin{table*}
    \caption{Comparison with Prior Art}
    \label{tbl:summary}       
    \centering
    \begin{threeparttable}
    \begin{tabular}{p{4.1cm}P{1.5cm}p{2.5cm}p{1.5cm}P{1.2cm}P{2.1cm}P{2.1cm}}
    \hline\noalign{\smallskip}
    Array\tnote{a} & FO & Spectrum\tnote{b} & Beampattern\tnote{c} & PRI & Antennas\tnote{d} & Pulses\tnote{e}
    \\
    \noalign{\smallskip}
    \hline
    \noalign{\smallskip}
    L-shaped ULA (U-U) \cite{TKishigami2016} & None & $B$ & Angle & Uniform & $Q+1$ & $Q+1$\\
    L-shaped Co-prime array (C-U) \cite{AMElbir2020} & None & $B$ & Angle & Uniform & $2\sqrt{Q+1}-1$ & $Q+1$\\
    \hline
    \multicolumn{7}{l}{\textbf{This paper}:}\\
    L-shaped ULA (U-C)& None & $B$ & Angle & Co-prime & $Q+1$ & $2\sqrt{Q+1}-1$\\
    L-shaped Co-prime array (C-C)& None & $B$ & Angle & Co-prime & $2\sqrt{Q+1}-1$ & $2\sqrt{Q+1}-1$\\
    L-shaped FDA (U-Cube) & Linear & $B+(P_s-1)\Delta f$ & Angle-range & Uniform  & $Q+1$ &$Q+1$ \\
    L-shaped FDA (UUC)& Linear & $B+(P_s-1)\Delta f$ & Angle-range & Co-prime  & $Q+1$  &$2\sqrt{Q+1}-1$ \\
    L-shaped FDA (UCU)& Co-prime & $B+\xi_{P_s-1}\Delta f$ & Angle-range & Uniform & $Q+1$ & $Q+1$ \\
    L-shaped FDA (UCC)& Co-prime & $B+\xi_{P_s-1}\Delta f$ & Angle-range & Co-prime & $Q+1$ & $2\sqrt{Q+1}-1$\\
    L-shaped Co-prime FDA (CUU)& Linear & $B+(P_s-1)\Delta f$ & Angle-range & Uniform  & $Q+1$ & $Q+1$\\
    L-shaped Co-prime FDA (CUC) & Linear & $B+(P_s-1)\Delta f$ & Angle-range & Co-prime & $Q+1$ &$2\sqrt{Q+1}-1$ \\
    L-shaped Co-prime FDA (CCU) & Co-prime & $B+\xi_{P_s-1}\Delta f$ & Angle-range & Uniform & $2\sqrt{Q+1}-1$ & $Q+1$ \\

    L-shaped Co-prime FDA (C-Cube) & Co-prime & $B+\xi_{P_s-1}\Delta f$ & Angle-range & Co-prime & $2\sqrt{Q+1}-1$ & $2\sqrt{Q+1}-1$ \\
    \noalign{\smallskip}\hline\noalign{\smallskip}
    \end{tabular}
     \begin{tablenotes}[para]
     \item[a] The first, second, and third letters in the abbreviations denote the array structure (\textit{U}LA or \textit{c}o-prime), FO (\textit{u}niformly linear or \textit{c}o-prime), and PRI (\textit{u}niform or \textit{c}o-prime), respectively; `-' is used to indicate the absence of FOs.
     \item[b] Cumulative transmit bandwidth of all array elements, where $B$ is the bandwidth of transmit time-limited pulse, $P_s$ is the number of sensors along each axis, $\Delta f$ is the fundamental frequency increment, and $\xi_{P_s-1}$ for co-prime arrays denotes the maximum value in the co-prime set $\mathcal{S}_{s1} \cup \mathcal{S}_{s2}$.
     \item[c] Parameters that the antenna beampattern depends upon. 
     \item[d] Number of antennas in the same physical aperture, where $Q$ is the number of targets. 
     \item[e] Number of pulses in the same CPI. 
    \end{tablenotes}
    \end{threeparttable}
    \end{table*}
    
Sparsity may also be applied to the \textit{slow-time} domain. The Doppler resolution is determined by the number of transmit pulses. Hence, a large number of transmit pulses or a longer slow-time duration yields high Doppler precision but adversely affects the ability of the radar to track targets in other directions. Some recent studies address this by sparsely transmitting pulses randomly along slow-time at non-uniform pulse repetition frequency (PRF) \cite{na2018tendsur,Vijay_Random_Sparse_Step_Frequency_2020}. Since non-uniform random pulsing leads to high sidelobes in the Doppler spectrum \cite{Maier_Non_Uniform_PRI_1993}, very recently \cite{xu2021difference} suggested \textit{co-pulsing}, i.e. transmitting pulses at co-prime pulse repetition interval (PRI) for a single-antenna radar. As for FDA, in general, there is very limited literature on Doppler estimation. A joint range-angle-Doppler localization was proposed in \cite{xu2016adaptive} but the array was a non-sparse FDA-MIMO system. Sparse pulse transmission for FDA has remained unexamined so far.

Contrary to prior works, we address a multitude of these shortcomings via a simpler 2-D FDA structure in the form of an L-shaped array, which performs co-prime sampling in spatial, spectral, and Doppler domains. While L-shaped configuration has been previously investigated for ULA \cite{TKishigami2016} and co-prime array \cite{AMElbir2020} (both with uniform PRI), its application to FDA and co-prime pulsing remains unexplored. In this paper, we develop the theory and localization method for L-shaped \textit{c}o-prime array with \textit{c}o-prime FOs and \textit{c}o-prime pulsing (C$^3$ or C-Cube) that allows for joint estimation of 2-D DoAs, range, and Doppler velocity. We exploit the difference coarray, frequency difference equivalence and PRI difference equivalence to achieve a significantly increased number of DoFs in space-time-Doppler domain thereby saving resources in aperture, spectrum, and dwell time. Our \textit{C-C}ube auto-pair\textit{ing} parameter retrieval (CCing) algorithm is based on singular value decomposition (SVD) of the concatenated covariance matrix, wherein the elevation and azimuth (range and Doppler) are embodied in the left (right) eigenvectors. The auto-pairing property between left and right eigenvectors allows for automatic pairing of 2-D DoA with range and Doppler velocity. 

In the process, we also obtain results for several other co-pulsing L-shaped arrays: ULA, co-prime array, FDA with linear and co-prime FOs, and co-prime FDA with linear FOs. These are compared with their uniform pulsing counterparts in Table~\ref{tbl:summary}. Our theoretical performance guarantees show that, to perfectly recover parameters of $Q$ targets, C-Cube requires a minimum of $2\sqrt{Q+1}-1$ antenna elements and $2\sqrt{Q+1}-1$ pulses. In comparison, the uniform-pulsing L-shaped ULA and L-shaped linear FO FDA need at least $Q+1$ elements and $Q+1$ pulses. Note that L-shaped co-prime array has a similar requirement of pulses and antennas as C-Cube but the former yields only an angle-dependent beampattern. Note that arrays - FDAs and otherwise - with similar element spacing and pulsing pattern have identical guarantees regarding the number of elements and pulses. But these arrays may still differ in FOs and thereby have differing spectrum consumption. Overall, fewer pulses, antennas, and transmit frequencies in C-Cube compared to its ULA counterparts imply that the power and interceptibility of the former radar are also lower. 

The rest of the paper is organized as follows. In the next section, we formulate the signal model of L-shaped C-Cube radar. We develop our CCing algorithm in Section~\ref{Section_ESPRIT} and derive the recovery guarantees for various co-pulsing L-shaped arrays in Section~\ref{Section_PA}, where we also obtain Cram\'er-Rao error bounds (CRBs) for joint angle-range-Doppler estimation using C-Cube FDA. We discuss other related sparse configurations and alternatives in Section~\ref{sec:adv}. We validate our methods through numerical experiments in Section~\ref{Section_simulation} and conclude in Section~\ref{Section_conclusion}.

Throughout this paper, we reserve boldface lowercase, boldface uppercase uppercase, and calligraphic letters for vectors, matrices, and index sets, respectively. We denote the transpose, conjugate, and Hermitian by $(\cdot)^T$, $(\cdot)^*$ and $(\cdot)^H$, respectively; $\otimes$, $\odot$ and $\diamond$ denote the Kronecker, Khatri-Rao and Hadamard products, respectively; and $\mathrm{vec}(\cdot)$ is the vectorization operator that turns a matrix into a vector by stacking all columns on top of the another. The notation $\lfloor \cdot \rfloor$ indicates the greatest integer smaller than or equal to the argument; $\angle(.)$ denotes the phase of its argument; $\mathrm{diag}(\cdot)$ denotes a diagonal matrix uses the elements of input vector as its diagonal elements; $\mathrm{E}[\cdot]$ is the statistical expectation function; $\log(.)$ represents natural logarithm of a number; $\Vert \cdot \Vert_{\mathcal{F}}$ denotes the Frobenius norm; $\mathbf{I}_L$ is the $L\times L$ identity matrix; and $\mathbf{\Pi}^{\bot}_{\bm{W}}= \mathbf{I}_L- \mathbf{W}(\mathbf{W}^H \mathbf{W})^{-1}\mathbf{W}^H$ denotes the orthogonal projection into the null space of $\mathbf{W}^H$. 

\section{Signal Model}  \label{Section_SM}
Consider an L-shaped co-prime array with $2P_s-1$ sensors, consisting of two orthogonally-placed linear co-prime arrays such that there are $P_s$ sensors each along the $x$- and $z$-axes, including a common reference sensor at the origin (Fig.~\ref{L-shaped}). Denote $\bm{\xi}=[\xi_0d,\xi_1d\ldots,\xi_{P_s-1}d]^T$ as the positions of the array sensors along $x$-axis, 
where $d$ is regarded as the fundamental spatial (inter-element) spacing  set to $d=\lambda_b/2=c/(2f_b)$, $f_b$ is the base carrier frequency, $c$ is the speed of light, $ \xi_m \in \mathcal{S}_{s1}=\{M_s i,0\leq i \leq N_s-1 , i \in \mathbb{N} \}\cup  \mathcal{S}_{s2}=\{N_s j,1\leq j \leq 2M_s-1,j \in \mathbb{N}  \} $, $M_s$ and $N_s$ are co-prime integers and $M_s < N_s$. Without loss of generality, $0=\xi_0 < \xi_1< \ldots < \xi_{P_s-1}$. It follows that $P_s= N_s+2M_s-1$. 

\subsection{Transmit signal}
Each element transmits a signal with an incrementally offset carrier frequency. That is, a continuous-wave signal transmitted from the $m$-th, $0\leq m \leq P_s-1$, element of $x$-axis is 
\begin{eqnarray}
s_{\xi_m}(t)=A_{\xi_m} e^{\mathrm{j}2\pi f_{\xi_m}t},~~ 0\leq t\leq T_p,~ m=0,1,\ldots,P_s-1,
\end{eqnarray}
where $A_{\xi_m}$ is amplitude, $f_{\xi_m}=f_b+ \xi_m \Delta f$ is carrier frequency of the $m$-th element, $\Delta f$ is fundamental frequency increment, $T_p=1/B$ is radar pulse duration and $B$ is its bandwidth. 

Similarly, along the $z$-axis, the positions of the array sensor are also denoted as $\bm{\xi}=[\xi_0d,\xi_1d\ldots,\xi_{P_s-1}d]^T$. The transmit signal from the $m$-th, $0\leq m \leq P_s-1$, element of $z$-axis is 
\begin{eqnarray}
s_{\xi_{-m}}(t)=A_{\xi_{-m}} e^{\mathrm{j}2\pi f_{\xi_{-m}}t},~~ 0\leq t\leq T_p,~ m=0,1,\ldots,P_s-1,
\end{eqnarray}
where $A_{\xi_{-m}}$ is the amplitude. $f_{\xi_{-m}} = f_b  + \xi_{-m} \Delta f= f_b  - \xi_{m} \Delta f$ and $\xi_{-m} \triangleq -\xi_{m}$. Without loss of generality, set $A_{\xi_m}=A_{\xi_{-m}}=1$. For the sake of simplicity, we denote the element index of transmit L-shaped array as $m, -(P_s-1)\leq m \leq (P_s-1)$ where the non-negative values of $m$ in the range $0\leq m \leq (P_s-1)$ refer to the elements along the $x$-axis and the negative values $-(P_s-1)\leq m < 0$ refer to the elements along the $z$-axis.  

A total of $K =N_t+2M_t-1$ pulses are transmitted during the coherent processing interval (CPI). The $k$-th, $0\leq k \leq K-1$, pulse starts at time instant $\eta_k T$ where $\eta_k \in \mathcal{S}_{t1}=\{M_t i,0\leq i \leq N_t-1 , i \in \mathbb{N} \}\cup  \mathcal{S}_{t2}=\{N_t j,1\leq j \leq 2M_t-1,j \in \mathbb{N}  \} $ and $T$ is the fundamental pulse repetition interval (PRI) for the case of uniform pulsing. Assume $0=\eta_0 < \eta_1 < \ldots <\eta_{K-1}$. The $k$-th, $0\leq k \leq K-1$, echo is sampled at the rate $1/T_p$ in fast-time $t_s=\eta_k T+l_r T_p$, where $l_r=0,1,\ldots,L_r-1=\lfloor T/T_p \rfloor$. 

\begin{figure}[t]
\includegraphics[width=1.0\columnwidth]{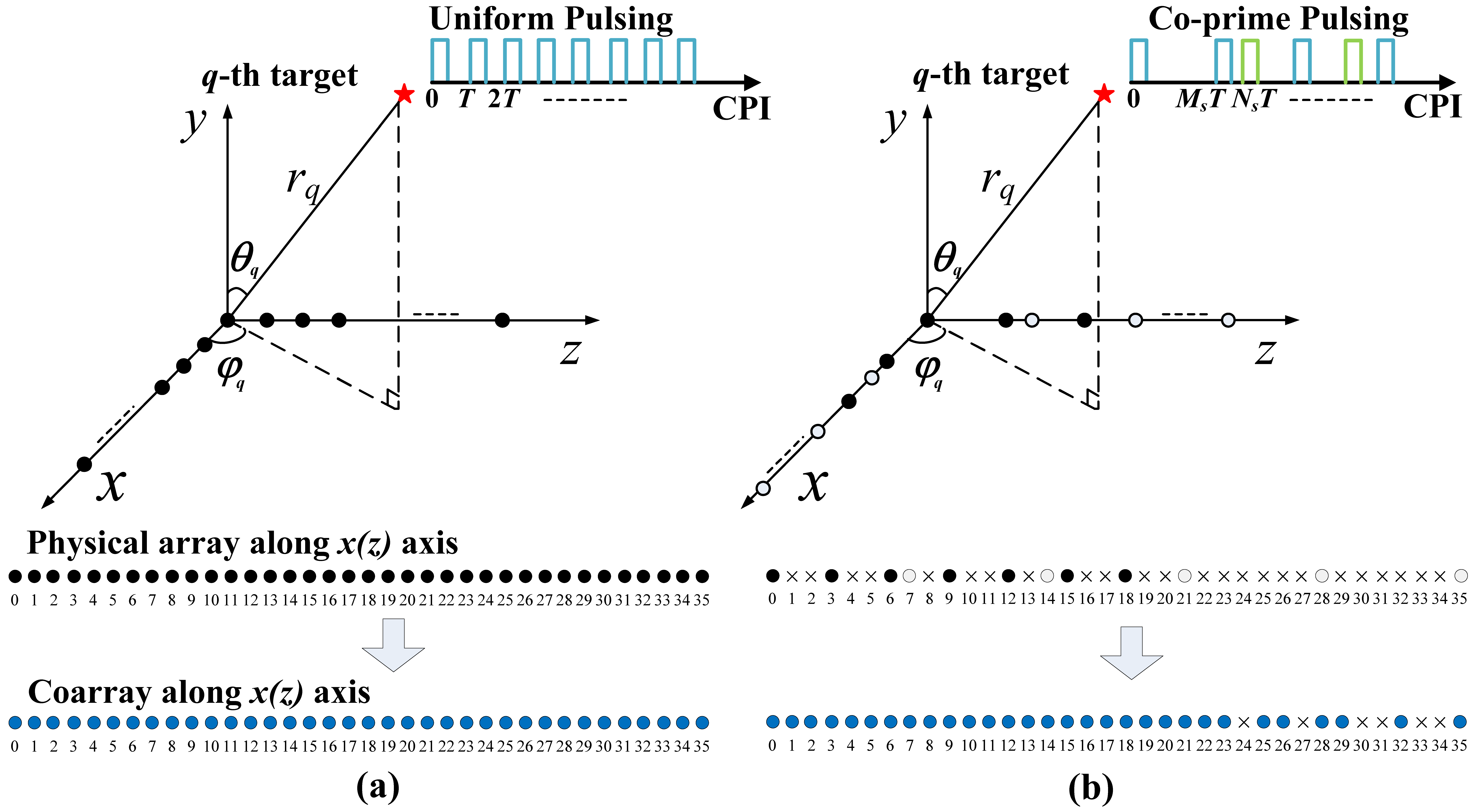}
\caption{Illustration of L-shaped FDAs with identical physical aperture and CPI duration. (a) Uniform L-shaped array with uniform FO and uniform PRI. (b) Co-prime L-shaped array with co-prime FO and co-prime PRI. In (b), along each axis, the sensor positions are given by the set $ \mathcal{S}_{s1} \cup \mathcal{S}_{s2}$ where $\mathcal{S}_{s1}$ corresponds to the black circles and $\mathcal{S}_{s2}$ corresponds to the grey circles. The holes in the coarray are denoted by '$\times$'. The transmit pulse sequence has been illustrated for a CPI corresponding to $8$ uniform PRIs.}
\label{L-shaped}
\end{figure}

\subsection{Operating conditions}
Assume that the radar scenario consists of $Q$ far-field point-targets. The transmit pulses are reflected back by the $Q$ targets and these echoes are captured by the radar receiver. Our goal is to recover the unknown parameter set $\bm{\gamma}_q=\{\theta_q, \varphi_q,r_q,\nu_q\}_{q=1}^{Q}$ where $\theta_q$ is the elevation angle, $\varphi_q$ is the azimuth angle, $r_q$ is the range and $\nu_q$ is the velocity of the $q$-th target. We make following assumptions about the target and radar parameters:
\begin{description}
    \item[A1] ``Narrowband platform'': Assume that the L-shaped FDA radar works in a narrowband environment, so the maximum increment across the L-shaped array satisfies $2 \xi_{P_s-1} \Delta f\ll f_b$. 
       
    \item[A2] ``Constant reflectivities'': The radar cross-section of each target is described by a Swerling-I model \cite{skolnik2008radar}, resulting in reflection coefficients being constant across pulses.
  
    \item[A3] ``Unambiguous DoA'': To ensure the array structure deprived of ambiguity, we have
\begin{align}                  \label{eq:SM3}
    \sin\theta_q \sin\varphi_q \neq \sin\theta_p \sin\varphi_p, &~~ \sin\theta_q \cos\varphi_q \neq \sin\theta_p \cos\varphi_p        \notag \\
   \text{for}~~~ 1\leq p & \neq  q \leq Q.
\end{align}
  
    \item[A4]  ``No range or Doppler ambiguities'': The range-Doppler pairs $\{(r_q,\nu_q)\}_{q=1}^Q$ lie in the radar’s unambiguous region of delay-Doppler plane $[0,R_{\max}]\times [0,\nu_{\max}]$. $R_{\mathrm{max}}=\frac{cT}{2}$ is the maximum unambiguous range and $\nu_{\mathrm{max}}=\frac{c}{f_bT}$ is the maximum unambiguous velocity, i.e., the time delays are no longer than the PRI and Doppler frequencies are up to the PRF. 

    \item[A5] ``Constant delays'': The modulation in frequency arising from a moving target appears as a frequency shift in the received signal. We consider this shift to be small over a CPI under the condition $\nu_q \ll c/(2B \eta_{K-1} T)$ so that the delay is approximated to be constant. In this case, the Doppler shifts induced are small, allowing for the piecewise-constant approximation: $\nu_q t \approx \nu_q \eta_k T, \text{for} ~t \in [\eta_k T, \eta_{k+1}T]$. 
  
    \item[A6] ``Constant Doppler shifts'': The velocity change of a target over a CPI is small compared with the velocity resolution where $\frac{d\nu_q}{dt} \ll \frac{c}{2f_b (\eta_{K-1} T)^2}$ is satisfied.
  
    \item[A7] ``Slow tangential velocities'': The angle change of a target over a CPI is small compared with the angle resolution, that is, $\frac{d \theta_q}{dt}\eta_{K-1} T \ll \frac{C_0 \lambda_b}{\xi_{P_s-1} d} \Rightarrow \frac{d \theta_q}{dt} \ll \frac{2C_0}{\xi_{P_s-1} \eta_{K-1}T}$ and $\frac{d \varphi_q}{dt}\eta_{K-1} T \ll \frac{C_0 \lambda_b}{\xi_{P_s-1} d} \Rightarrow \frac{d \varphi_q}{dt} \ll \frac{2C_0}{\xi_{P_s-1} \eta_{K-1}T}$  where $C_0$ is a positive constant. This allows for constant DoAs during the CPI.

\end{description}

\subsection{Received signal}
Consider the linear array along $x$-axis. The two-way time delay of the $q$-th signal received by the $n$-th, $0\leq n\leq P_s-1$ element at the $k$-th, $0 \leq k \leq K-1$, pulse is
$\tau_{n,k}^x(\bm{\gamma}_q)  =
 \frac{ 2(r_q+ \nu_q \eta_k T) +\xi_n d \sin\theta_q \sin\varphi_q   }{c} 
 =\frac{ 2r_q+ 2\nu_q \eta_k T +\xi_n d \sin\theta_q \sin\varphi_q   }{c}$. 
Then, the received signal of the $q$-th target corresponding to the $m$-th ($-(P_s-1) \leq m \leq P_s-1$) transmit element, $n$-th ($0\leq n\leq P_s-1$ ) receive element, and $k$-th ($0\leq k \leq K-1$) pulse is
\begin{eqnarray}                                   \label{eq:UF3}
x_{k,n,m}(t,\bm{\gamma}_q)= \rho_q(t) e^{\mathrm{j}2\pi f_{\xi_m} (t- \tau_{n,k}^x(\bm{\gamma}_q))} ,
\end{eqnarray}
where $\rho_q(t),q=1,\ldots,Q$, are complex scattering coefficients of the targets, modeled as uncorrelated random variables with $\mathrm{E}[\rho_q^*\rho_p]=\sigma_q^2\delta_{q,p}$. Our focus is on the DoF enhancement at the receive array. Hence,  for simplicity, assume that the information about transmit array is embodied in $\rho_q(t),q=1,\ldots,Q$ \cite{SQin2017}.

After demodulation and applying band-pass filtering, 
the baseband signal corresponding to (\ref{eq:UF3}) is\par\noindent\small
\begin{flalign}                     \label{eq:SM6}
&x_{k,n,m}(t_s,\bm{\gamma}_q)= \rho_q(t_s) e^{-\mathrm{j}2\pi f_{\xi_m}  \tau_{n,k}^x(\bm{\gamma}_q) }  \nonumber\\
&= \rho_q(t_s) e^{ \frac{-j4\pi f_b r_q}{c} }  e^{-\mathrm{j}2\pi \frac{(f_b+\xi_m \Delta f) d\sin\theta_q \sin\varphi_q}{c}\xi_n} e^{-\mathrm{j}2\pi \frac{2\Delta f r_q}{c}   \xi_m }     \nonumber \\
&~~~ \times e^{-\mathrm{j}2\pi \frac{2v_q T(f_b+\xi_m \Delta f)}{c} \eta_k } ,
\end{flalign}\normalsize

Note that since the frequency increment is negligible compared to the base carrier frequency, the signal spectrum spreading effect arising from the frequency increment in spatial and Doppler frequency domains is not significant \cite{JXu2015_STAP}. We approximate the phases in the second and fourth terms after the second equality in (\ref{eq:SM6}) as
    \begin{align}
        &-\mathrm{j}2\pi \frac{(f_b+\xi_m \Delta f) d\sin\theta_q \sin\varphi_q}{c}\xi_n    \notag \\
        &\hspace{5mm}=  -\mathrm{j}2\pi \frac{f_b d\sin\theta_q \sin\varphi_q}{c}\xi_n   -\mathrm{j}2\pi \frac{\xi_m \Delta f d\sin\theta_q \sin\varphi_q}{c}\xi_n    \notag \\
        &\hspace{5mm}\approx  -\mathrm{j}2\pi \frac{ d\sin\theta_q \sin\varphi_q}{\lambda_b}\xi_n,     
    \end{align}    
        and
    \begin{align}    
     -j2\pi \frac{2v_q T(f_b+\xi_m \Delta f)}{c} \eta_k                      
     &= -j2\pi \frac{2v_q T f_b}{c} \eta_k -  j2\pi \frac{2v_q T \xi_m \Delta f}{c} \eta_k                            \notag \\
     & \approx -j2\pi \frac{2v_q T }{\lambda_b} \eta_k
    \end{align}
    Thus, (\ref{eq:SM6}) becomes
    \begin{flalign}
    &x_{k,n,m}(t_s,\bm{\gamma}_q)      \nonumber \\
    &  \approx \rho_q(t_s) e^{ \frac{-j4\pi f_b r_q}{c} }  e^{-\mathrm{j}2\pi \frac{ d\sin\theta_q \sin\varphi_q}{\lambda_b}\xi_n} e^{-\mathrm{j}2\pi \frac{2\Delta f r_q}{c}   \xi_m } e^{-\mathrm{j}2\pi \frac{2v_q T}{\lambda_b} \eta_k }                \nonumber \\
    &=\widetilde{\rho}_q(t_s)   e^{-\mathrm{j}2\pi \frac{\sin\theta_q \sin\varphi_q}{2} \xi_n } e^{-\mathrm{j}2\pi \frac{2\Delta f r_q}{c}  \xi_m } e^{ -j2\pi \frac{2v_q T}{\lambda_b}\eta_k  }   
    \end{flalign}
where the term $\mathrm{exp}\left\{ \frac{-j4\pi f_b r_q}{c} \right\}$ is a constant that can be absorbed into the scattering coefficients such that $\widetilde{\rho}_q(t)=\rho_q(t)  \mathrm{exp}\left\{ \frac{-j4\pi f_b r_q}{c} \right\}$.

For $Q$ far-field targets, the received signal is the superposition of echoes from all targets, i.e,
\begin{align}                        \label{eq:SM7}
x_{k,n,m}(t_s)&= \sum_{q=1}^Q x_{m,n,k}(t_s,\bm{\gamma}_q)   + n^x_{k,n,m}(t_s)                              \notag \\
&= \sum_{q=1}^Q   \widetilde{\rho}_q(t_s)   e^{-\mathrm{j}2\pi \frac{\sin\theta_q \sin\varphi_q}{2} \xi_n } e^{-\mathrm{j}2\pi \frac{2\Delta f r_q}{c}  \xi_m }  e^{-\mathrm{j}2\pi \frac{2v_q T}{\lambda_b} \eta_k  }                 \notag \\
&~~  + n^x_{k,n,m}(t_s)  .
\end{align}
where $ n^x_{k,n,m}(t_s)$ is the additive uncorrelated zero mean Gaussian white noise sequence with variance $\sigma_n^2$. Stacking $x_{k,n,m}(t_s)$ for all $-(P_s-1)\leq m \leq P_s-1, 0\leq n \leq P_s-1$, and $0\leq k \leq K-1$ yields an $KP_s(2P_s-1) \times 1$ vector,
\begin{eqnarray}                           \label{eq:SM8}
\mathbf{x}(t_s) &=& \sum_{q=1}^Q \widetilde{\rho}_q (t_s) \mathbf{c}(r_q) \otimes\mathbf{b}(\nu_q) \otimes \mathbf{a}_x(\theta_q,\varphi_q) + \mathbf{n}^x(t_s) \notag  \\
&=&  (\mathbf{C} \odot \mathbf{B} \odot \mathbf{A}_x) \bm{\rho} (t_s) + \mathbf{n}^x(t_s),
\end{eqnarray}
where
$\mathbf{a}_x(\theta_q,\varphi_q)=\left[1,e^{-\mathrm{j}\pi\xi_1 \sin\theta_q \sin\varphi_q},\dots, e^{-\mathrm{j}\pi\xi_{P_s-1}\sin\theta_q \sin\varphi_q} \right]^T 
\in \mathbb{C}^{P_s\times 1}$, 
$\mathbf{c}(r_q) = \left[e^{\mathrm{j}4\pi \Delta f r_q \xi_{P_s-1}/c} ,
\ldots, e^{\mathrm{j}4\pi \Delta f r_q \xi_{1}/c},1, e^{-\mathrm{j}4\pi \Delta f r_q \xi_{1}/c},
 \notag \right.  \\ 
\left. \ldots, e^{-\mathrm{j}4\pi \Delta f r_q \xi_{P_s-1}/c}  \right]^T
\in \mathbb{C}^{(2P_s-1)\times 1}$, 
and
$\mathbf{b}(\nu_q)=\left[1,e^{-\mathrm{j}4\pi \nu_q T \eta_1 /\lambda_b},\ldots, e^{-\mathrm{j}4\pi \nu_q T \eta_{K-1}/\lambda_b} \right]^T \in \mathbb{C}^{K \times 1}$, 
are the steering vectors corresponding to $(\theta_q, \varphi_q)$, $r_q$, and $\nu_q$, respectively.
In addition, $\mathbf{A}_x = [\mathbf{a}_x(\theta_1,\varphi_1),\ldots,\mathbf{a}_x(\theta_Q,\varphi_Q)]$, $\mathbf{C} = [\mathbf{c}(r_1),\ldots,\mathbf{c}(r_Q)]$, $\mathbf{B} = [\mathbf{b}(\nu_1),\ldots,\mathbf{b}(\nu_Q)]$ and $\bm{\rho}(t_s) = [\rho_1(t_s),\ldots,\rho_Q(t_s)]^T$. It follows from (\ref{eq:SM8}) that the total number of lags of physical space-time-frequency array equals the product of lags in each dimension. For example, the structure in Fig. \ref{L-shaped}b yields 864 lags when only non-negative region is considered.

\subsection{Contiguous space-time-frequency coarray}
The $KP_s(2P_s-1)\times KP_s(2P_s-1)$ covariance matrix of data matrix $\mathbf{x}(t_s)$ is 
\begin{align}                               \label{eq:SM15}
\mathbf{R_x} &= \mathrm{E}[\mathbf{x}(t_s)\mathbf{x}^H(t_s)]                              \notag \\
&= (\mathbf{C} \odot \mathbf{B} \odot \mathbf{A}_x) \mathbf{R}_{\bm{\rho}} (\mathbf{C} \odot \mathbf{B} \odot \mathbf{A}_x)^H+ \sigma_n^2 \mathbf{I}_{KP_s(2P_s-1)},
\end{align}
where $\mathbf{R}_{\bm{\rho}} = \mathrm{E}[\bm{\rho}(t_s)\bm{\rho}^H(t_s)]= \mathrm{diag}([\sigma_1^2,\ldots,\sigma_Q^2])$. Vectorizing the matrix $\mathbf{R_{x}}$ produces $(KP_s(2P_s-1))^2 \times 1$ vector
\begin{align}                             \label{eq:SM20}                             
\mathbf{r_x} &= \mathrm{vec}(\mathbf{R_{x}})                          \notag \\
&=  [(\mathbf{C} \odot \mathbf{B} \odot \mathbf{A}_x)^* \odot (\mathbf{C} \odot \mathbf{B} \odot \mathbf{A}_x)] \mathbf{r}_{\bm{\rho}} + \sigma_n^2 \mathrm{vec}(\mathbf{I}_{KP_s(2P_s-1)})                                                           \notag \\
&=  \mathbf{K}[(\mathbf{C}^*\odot \mathbf{C})\odot (\mathbf{B}^*\odot \mathbf{B}) \odot (\mathbf{A}_x^* \odot \mathbf{A}_x)] \mathbf{r}_{\bm{\rho}}                                                \notag \\
&~~ +\sigma_n^2 \mathrm{vec}(\mathbf{I}_{KP_s(2P_s-1)}),
\end{align}
where $\mathbf{r}_{\bm{\rho}}=[\sigma_1^2,\ldots,\sigma_Q^2]^T$ is a vector containing the diagonal elements of the matrix $\mathbf{R_{\bm{\rho}}}$; $\mathbf{K}\in \mathbb{C}^{K^2P_s^2(2P_s-1)^2 \times K^2P_s^2(2P_s-1)^2}$ is a permutation matrix that has exactly one entry of value $1$ in each row and each column and $0$s elsewhere; the $(i,j)$-th entry of $\mathbf{A}^*_x \odot \mathbf{A}_x$ is $\{e^{-\mathrm{j}\pi (\xi_i-\xi_j)\sin\theta_q \sin\varphi_q},0 \leq i,j\leq P_s-1, 1 \leq q \leq Q\}$; the $(i,j)$-th entry of $\mathbf{C}^* \odot \mathbf{C}$ is $\{e^{-\mathrm{j}4\pi (\xi_i-\xi_j)\Delta f r_q/c}, 0 \leq i,j\leq P_s-1, 1 \leq q \leq Q\}$; and the $(i,j)$-th entry of $\mathbf{B}^* \odot \mathbf{B}$ is $\{e^{-\mathrm{j}4\pi \nu_q T (\eta_i-\eta_j) /\lambda_b}, 0 \leq i,j\leq K-1, 1 \leq q \leq Q\}$. 

\begin{figure}[t]
\includegraphics[width=1.0\columnwidth]{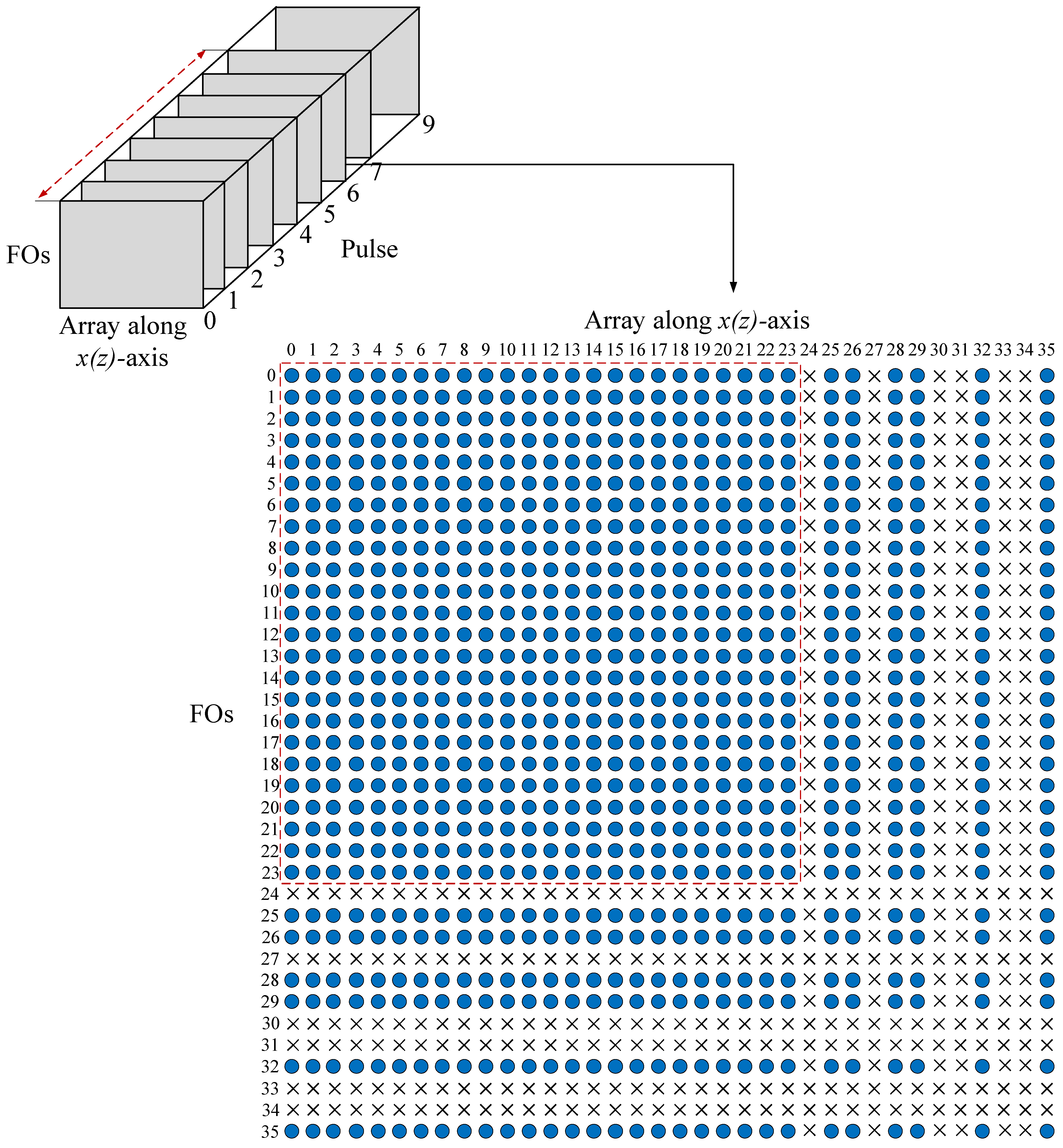}
\caption{Space-time-frequency coarray along the non-negative axes with $M_s=3, N_s=7$ and $M_t=2, N_t=3$. The red dashed line denotes the contiguous area of coarray.}
\label{Space_time_frequency_coarray}
\end{figure}

The C-Cube structure (Fig. \ref{L-shaped}b) yields space-time-frequency difference coarray along the positive $x(z)$-axis as shown in Fig. \ref{Space_time_frequency_coarray}; the negative side is obtained by mirroring along each dimension. The number of non-negative lags increases from 864 to 8100. As is the case with 1-D co-prime array \cite{EBouDaher2015}, holes also exist in the space-time-frequency difference coarray from (\ref{eq:SM20}). Our retrieval algorithm is based on the framework of subspace-based method where the rotational invariance property is exploited. Hence, we are only concerned with the contiguous space-time-frequency coarray, namely the sub-coarray between the first positive hole and the first negative hole. So, the holes in difference coarray mainly affects the aperture of the contiguous space-time-frequency coarray. The larger aperture the contiguous space-time-frequency coarray has, more DOFs and better performance it yields. The literature suggests appropriate hole filling methods \cite{EBouDaher2015,ZMao2022} that may be used for structures mentioned in this paper. This may yield more efficient coarray structures. 

Following \cite{EBouDaher2015}, a contiguous space-time-frequency coarray is produced where $\mathbf{A}^*_x \odot \mathbf{A}_x$ has the contiguous elements from $e^{-\mathrm{j}\pi (M_sN_s+M_s-1)\sin\theta_q \sin\varphi_q}$ to $e^{\mathrm{j}\pi (M_sN_s+M_s-1)\sin\theta_q \sin\varphi_q}$, $\mathbf{C}^* \odot \mathbf{C}$ has the contiguous elements from $e^{-\mathrm{j}4\pi (M_sN_s+M_s-1)\Delta f r_q/c}$ to $e^{\mathrm{j}4\pi (M_sN_s+M_s-1)\Delta f r_q/c}$, and $\mathbf{B}^* \odot \mathbf{B}$ has the contiguous elements from $e^{-\mathrm{j}4\pi \nu_q T (M_tN_t +M_t-1) /\lambda_b}$ to $e^{\mathrm{j}4\pi \nu_q T (M_tN_t +M_t-1) /\lambda_b}$. For the sake of simplicity, assume $L_s= M_sN_s+M_s-1, L_t=M_tN_t+M_t-1$. Picking up these continuous entries from $\mathbf{r_x}$ yields a new vector
\begin{eqnarray}            \label{eq:SM21}
\widetilde{\mathbf{r}}_x &=& (\widetilde{\mathbf{C}} \odot \widetilde{\mathbf{B}} \odot \widetilde{\mathbf{A}}_x) \mathbf{r}_{\bm{\rho}} + \sigma_n^2 \acute{\mathbf{e}},     
\end{eqnarray}
where
\begin{align}           \label{eq:SM22}
\widetilde{\mathbf{A}}_x &=
\left(
\begin{smallmatrix}
e^{-\mathrm{j}\pi L_s\sin\theta_1 \sin\varphi_1} & \cdots & e^{-\mathrm{j}\pi L_s\sin\theta_Q \sin\varphi_Q}  \\
\vdots                         & \ddots &   \vdots                        \\
e^{\mathrm{j}\pi L_s\sin\theta_1 \sin\varphi_1} & \cdots & e^{\mathrm{j}\pi L_s\sin\theta_Q \sin\varphi_Q}
\end{smallmatrix} \right) \in \mathbb{C}^{(2L_s+1)\times Q},
\end{align}
\begin{eqnarray}
\widetilde{\mathbf{C}}=
\left(
\begin{smallmatrix}
e^{-\mathrm{j}4\pi L_s\Delta f r_1/c} & \cdots & e^{-\mathrm{j}4\pi L_s \Delta f r_Q/c}  \\
\vdots                         & \ddots &   \vdots                        \\
e^{\mathrm{j}4\pi L_s \Delta f r_1/c} & \cdots & e^{\mathrm{j}4\pi L_s \Delta f r_Q/c}
\end{smallmatrix} \right) \in \mathbb{C}^{(2L_s+1)\times Q},
\end{eqnarray}
\begin{eqnarray}
\widetilde{\mathbf{B}}=
\left(
\begin{smallmatrix}
e^{-\mathrm{j}4\pi \nu_1 T L_t /\lambda_b} & \cdots & e^{-\mathrm{j}4\pi \nu_Q T L_t /\lambda_b}  \\
\vdots                         & \ddots &   \vdots                        \\
e^{\mathrm{j}4\pi \nu_1 T L_t /\lambda_b} & \cdots &  e^{\mathrm{j}4\pi \nu_Q T L_t /\lambda_b}
\end{smallmatrix} \right) \in \mathbb{C}^{(2L_t+1)\times Q}.
\end{eqnarray}
and $\acute{\mathbf{e}}$ is an all-zero vector except a 1 in the middle position.

Similarly, for the $z$-axis, the received signal model in matrix form follows from (\ref{eq:SM8}) as
\begin{eqnarray}             \label{eq:SM24}
\mathbf{z}(t_s) &=& \sum_{q=1}^Q \widetilde{\rho}_q (t_s) \mathbf{c}(r_q) \otimes \mathbf{b}(\nu_q) \otimes \mathbf{a}_z(\theta_q,\varphi_q) + \mathbf{n}^z(t_s) \notag  \\
&=&  (\mathbf{C} \odot \mathbf{B} \odot\mathbf{A}_z) \bm{\rho} (t_s) + \mathbf{n}^z(t_s),
\end{eqnarray}
where 
$\mathbf{a}_z(\theta_q,\varphi_q)=\left[1,e^{-\mathrm{j}\pi \xi_1\sin\theta_q \cos\varphi_q},\dots, e^{-\mathrm{j}\pi\xi_{P_s-1}\sin\theta_q \cos\varphi_q} \right]^T  \in \mathbb{C}^{P_s \times 1}$ 
and $\mathbf{A}_z=[\mathbf{a}_z(\theta_1,\varphi_1),\dots,\mathbf{a}_z(\theta_Q,\varphi_Q)]$. Repeating the computations from (\ref{eq:SM15}) to (\ref{eq:SM21}) gives the contiguous space-time-frequency coarray along the $z$-axis as
\begin{eqnarray}            \label{eq:SM25}
\widetilde{\mathbf{r}}_z &=& (\widetilde{\mathbf{C}} \odot \widetilde{\mathbf{B}} \odot \widetilde{\mathbf{A}}_z) \mathbf{r}_{\bm{\rho}} + \sigma_n^2 \acute{\mathbf{e}},    
\end{eqnarray}
where
\begin{eqnarray}
\widetilde{\mathbf{A}}_z &=
\left(
\begin{smallmatrix}
e^{-\mathrm{j}\pi L_s\sin\theta_1 \cos\varphi_1} & \cdots & e^{-\mathrm{j}\pi L_s\sin\theta_Q \cos\varphi_Q}  \\
\vdots                         & \ddots &   \vdots                        \\
e^{\mathrm{j}\pi L_s\sin\theta_1 \cos\varphi_1} & \cdots & e^{\mathrm{j}\pi L_s\sin\theta_Q \cos\varphi_Q}
\end{smallmatrix} \right) \in \mathbb{C}^{(2L_s+1)\times Q}.
\end{eqnarray}

It is clear from (\ref{eq:SM21}) and (\ref{eq:SM25}) that our model exploits the virtual manifold matrices, namely $\widetilde{\mathbf{A}}_x$, $\widetilde{\mathbf{A}}_z$, $\widetilde{\mathbf{C}}$, and $\widetilde{\mathbf{B}}$. From the DoF analysis of contiguous coarray in 1-D \cite{EBouDaher2015,AMElbir2020}, we conclude the following about DoFs in our case: If all targets have unique values for all parameters, i.e. no two targets have same azimuth angles, elevation angles, ranges, and Doppler velocities, then the DoFs reach up to $\mathcal{O}(\min\{2L_s,2L_t\})$. If there exist pairs of targets for which one or more parameters are identical, then the overall DoFs are the multiplication of the DoFs from each domain, reaching up to $\mathcal{O}(16L_s^3L_t)$. 

For single array received signals, as in (\ref{eq:SM21}) and (\ref{eq:SM25}), subspace and tensor-based multi-dimensional parameter recovery techniques have been reported in previous works. For example, MUltiple SIgnal Classification (MUSIC) was adopted in \cite{LLiu2018} for joint frequency and DoA estimation for undersampled narrowband signals. This was extended to wideband scenarios in \cite{FWang2018,ZZhang2021}, where a \textit{can}onical \textit{decomp}osition \textit{par}allel \textit{fac}tor (CANDECOMP/PARAFAC or CP) tensor decomposition method jointly recovered DoA, carrier frequency, and power spectrum. In \cite{na2018tendsur}, both tensor-completion and tensor orthogonal matched pursuit (OMP) were used to recover range, Doppler velocity, and DoA with undersampled signals. However, the aforementioned algorithms cannot be directly generalized to an L-shaped array, which requires an additional pairing procedure. 

The literature suggests a few recovery methods for L-shaped arrays, such as subspace algorithm with 2-D spatial smoothing \cite{ZPeng2020}, connection-matrices method \cite{ZLi2021} and sparse reconstruction for 1-bit DoA estimation \cite{CLi2021}. However, nearly all of these works focus on passive sensing and estimate only DoA while leaving out the retrieval of range and Doppler velocity. To this end, we now introduce an auto-pairing multi-parameter retrieval which addresses these issues when no two targets have any identical parameter values. The proposed algorithm also works when one or more of parameters of targets are the same but this is at the cost of some DoFs in the angle-range-Doppler domain. This auto-pairing is possible because, unlike prior works, our L-shaped array is also an FDA, wherein the beampattern is dependent on both range and DoA. Hence, despite separate estimation of these parameters, the FDA properties allow us to pair the parameters uniquely.

\section{Auto-Pairing Parameter Retrieval}  \label{Section_ESPRIT}
We focus on the auto-paring procedure of unknown parameters $\bm{\gamma}_q=\{\theta_q, \varphi_q, r_q,\nu_q \}_{q=1}^Q$. So, the variance of noise $\sigma_n^2$ is assumed to be known \textit{a priori} that is absorbed into $\widetilde{\mathbf{r}}_x$ in (\ref{eq:SM21}) and $\widetilde{\mathbf{r}}_z$ in (\ref{eq:SM25}), respectively. Then, by rearranging the elements of $\widetilde{\mathbf{r}}_x$ and $\widetilde{\mathbf{r}}_z$ into matrices $\mathbf{X}\in \mathbb{C}^{(2L_s+1)\times (2L_s+1)(2L_t+1)}$ and $\mathbf{Z}\in \mathbb{C}^{(2L_s+1)\times (2L_s+1)(2L_t+1)}$, we have
$\mathbf{X} = \widetilde{\mathbf{A}}_x \mathbf{R}_{\bm{\rho}} (\widetilde{\mathbf{C}}\odot \widetilde{\mathbf{B}})^T$  
and
$\mathbf{Z} = \widetilde{\mathbf{A}}_z \mathbf{R}_{\bm{\rho}} (\widetilde{\mathbf{C}}\odot \widetilde{\mathbf{B}})^T$,
where $\widetilde{\mathbf{r}}_x =\mathrm{vec}(\mathbf{X})$ and $\widetilde{\mathbf{r}}_z =\mathrm{vec}(\mathbf{Z})$. Define the concatenated covariance matrix
\begin{eqnarray}                      \label{eq:RA2}
\mathbf{R}_{XZ} = \left[
\begin{matrix}
\mathbf{X}  \\
\mathbf{Z}
\end{matrix} \right]
= \left[\begin{matrix}
 \widetilde{\mathbf{A}}_x  \\
 \widetilde{\mathbf{A}}_z
\end{matrix} \right] \mathbf{R}_{\bm{\rho}} (\widetilde{\mathbf{C}}\odot \widetilde{\mathbf{B}})^T  \in \mathbb{C}^{(4L_s+2)\times (2L_s+1)(2L_t+1)}.
\end{eqnarray}

The singular value decomposition (SVD) of $\mathbf{R}_{XZ}$ yields
\begin{eqnarray}                    \label{eq:RA3}
\mathbf{R}_{XZ} = [\mathbf{U}_1~ \mathbf{U}_2]
\left[\begin{matrix}
 \mathbf{\Lambda} & \mathbf{0}  \\
 \mathbf{0}  &  \mathbf{0}
\end{matrix} \right]    [\mathbf{V}_1~  \mathbf{V}_2]^H = \mathbf{U}_1 \mathbf{\Lambda} \mathbf{V}_1^H,
\end{eqnarray}
where the columns of the matrix $[\mathbf{U}_1~ \mathbf{U}_2]$ are the left singular vectors of $\mathbf{R}_{XZ}$ with $\mathbf{U}_1$ containing the vectors corresponding to the first $Q$ singular values, $\mathbf{\Lambda}$ is a $Q\times Q$ diagonal matrix with the $Q$ non-zero singular values of $\mathbf{R}_{XZ}$, and $[\mathbf{V}_1~  \mathbf{V}_2]$ contains the right singular vectors of $\mathbf{R}_{XZ}$ with $\mathbf{V}_1$ containing the vectors corresponding to the first $Q$ singular values. Assume that there exist two invertible $Q \times Q$ matrices $\mathbf{T}_L$ and $\mathbf{T}_R$ such that
\begin{align}                    \label{eq:RA5}
\mathbf{U}_1 =
\left[\begin{matrix}
 \mathbf{U}_{11}  \\
 \mathbf{U}_{12}
\end{matrix} \right]
&=\left[\begin{matrix}
 \widetilde{\mathbf{A}}_x  \\
 \widetilde{\mathbf{A}}_z
\end{matrix} \right] \mathbf{T}_L \in \mathbb{C}^{(4L_s+2)\times Q} 
\end{align}
\begin{align}                    \label{eq:RA5_1}
\mathbf{V}_1^H&= \mathbf{T}_R (\widetilde{\mathbf{C}}\odot \widetilde{\mathbf{B}})^{T}\in \mathbb{C}^{Q \times(2L_s+1)(2L_t+1)},
\end{align}
where $\mathbf{U}_{11}$ and $\mathbf{U}_{12}$ denote the first and last $2L_s+1$ rows of $\mathbf{U}_1$, respectively. Proposition \ref{Pro_TL} provides sufficient conditions for the existence of $\mathbf{T}_L$ and $\mathbf{T}_R$. 

\begin{proposition}             \label{Pro_TL}
Consider an L-shaped co-prime array with $P_s=N_s+2M_s-1$ sensors along each axis and $d$ be the fundamental spatial spacing. Each sensor transmits a total of $K=N_t+2M_t-1$ pulses at co-prime intervals with $T$ as the fundamental PRI. There are $Q$ far-field targets. If 
\begin{description}
\item[C1] $d\leq \frac{c}{2f_b}$,\vspace{4pt}
\item[C2] $T\leq  \frac{c}{2f_b \nu_{\mathrm{max}}}$,\vspace{4pt}
\item[C3] $\Delta f \leq \frac{c}{2R_{\mathrm{max}}}$,
\item[C4] $L_s \triangleq M_sN_s+M_s-1 > \frac{Q-1}{2}$,\vspace{4pt}
\item[C5] $L_t \triangleq M_tN_t+M_t-1 > \frac{Q-1}{2}$,
\end{description}
then there exist invertible $Q\times Q$ matrices $\mathbf{T}_L$ and $\mathbf{T}_R$ such that (\ref{eq:RA5}) and (\ref{eq:RA5_1}) hold.
\end{proposition}
\begin{IEEEproof}
See Appendix~\ref{app:proof_prop2}.
\end{IEEEproof}

In order to recover all four parameters, namely azimuth, elevation, range and Doppler from $\widetilde{\mathbf{r}}_x$ and $\widetilde{\mathbf{r}}_z$, sparse reconstruction methods may be employed. Among prior works, the joint estimation of signal parameters via rotational invariant techniques (ESPRIT) in \cite{SSIoushua2017} examines only the left singular vectors of correlation matrix to extract 2-D information in a passive array. An analogous algorithm for 4-D parameter estimation entails excessive computational load. This is also true for various 4-D tensor-based sparse recovery methods proposed for a single-array sensing in \cite{na2018tendsur}. In contrast, our proposed CCing algorithm exploits the fact that the 2-D DoA and range-Doppler information are embodied in the left and right singular vectors of correlation matrix, respectively, thus avoiding a 4-D search with high computational complexity.

\subsection{2-D DoA estimation}
To obtain the rotational invariant structure of the array to estimate the unknown parameters, denote the first (last) $2L_s$ rows of $\widetilde{\mathbf{A}}_x$ and $\widetilde{\mathbf{A}}_z$ as $\widetilde{\mathbf{A}}_{x1}$ ($\widetilde{\mathbf{A}}_{x2}$) and $\widetilde{\mathbf{A}}_{z1}$ ($\widetilde{\mathbf{A}}_{z2}$), respectively. Note that
\begin{eqnarray}                     \label{eq:RA6}
\widetilde{\mathbf{A}}_{x2}= \widetilde{\mathbf{A}}_{x1} \mathbf{\Phi},~ ~ \widetilde{\mathbf{A}}_{z2}= \widetilde{\mathbf{A}}_{z1} \mathbf{\Psi},
\end{eqnarray}
where 
$\mathbf{\Phi} \triangleq \mathrm{diag}\left[e^{\mathrm{j}\pi\sin\theta_1 \sin\varphi_1},\dots,e^{\mathrm{j}\pi\sin\theta_Q \sin\varphi_Q}\right]$ and
$\mathbf{\Psi} \triangleq \mathrm{diag}\left[e^{\mathrm{j}\pi\sin\theta_1 \cos\varphi_1},\dots,e^{\mathrm{j}\pi\sin\theta_Q \cos\varphi_Q}\right]$.

Denote $\mathbf{U}_{111}$ ($\mathbf{U}_{112}$) as the first (last) $2L_s$ rows of $\mathbf{U}_{11}$ and $\mathbf{U}_{121}$ ($\mathbf{U}_{122}$) as the first (last) $2L_s$ rows of $\mathbf{U}_{12}$. Based on (\ref{eq:RA5}) and (\ref{eq:RA6}), we have
$\mathbf{U}_{111} =\widetilde{\mathbf{A}}_{x1} \mathbf{T}_L$,  
$\mathbf{U}_{112} = \widetilde{\mathbf{A}}_{x2} \mathbf{T}_L  = \widetilde{\mathbf{A}}_{x1} \mathbf{\Phi} \mathbf{T}_L  = \mathbf{U}_{111} \mathbf{T}_L^{-1} \mathbf{\Phi} \mathbf{T}_L$, 
$\mathbf{U}_{121} =\widetilde{\mathbf{A}}_{z1} \mathbf{T}_L$, and
$\mathbf{U}_{122} = \widetilde{\mathbf{A}}_{z2} \mathbf{T}_L  = \widetilde{\mathbf{A}}_{z1} \mathbf{\Psi} \mathbf{T}_L  = \mathbf{U}_{121} \mathbf{T}_L^{-1} \mathbf{\Psi} \mathbf{T}_L$.
Then, 
\begin{eqnarray}            \label{eq:RA7_1}
\mathbf{U}_{111}^{\dag} \mathbf{U}_{112} = \mathbf{T}_L^{-1} \mathbf{\Phi} \mathbf{T}_L
\end{eqnarray}
\begin{eqnarray}            \label{eq:RA7_2}
\mathbf{U}_{121}^{\dag} \mathbf{U}_{122} = \mathbf{T}_L^{-1} \mathbf{\Psi} \mathbf{T}_L
\end{eqnarray}
Therefore, one could obtain $\mathbf{\Phi}$ and $\mathbf{T}_L$ using the eigenvalue decomposition of $\mathbf{U}_{111}^{\dag} \mathbf{U}_{112}$, up to a permutation. Denote by $\widehat{\mathbf{\Phi}}$ and $\widehat{\mathbf{T}}_L$ the resulting matrices. We then compute $\widehat{\mathbf{\Psi}}$ as
\begin{eqnarray}                   \label{eq:RA9}
\widehat{\mathbf{\Psi}} = \widehat{\mathbf{T}}_L (\mathbf{U}_{121}^{\dag} \mathbf{U}_{122}) \mathbf{T}_L^{-1}.
\end{eqnarray}
Once $\mathbf{\Psi}$ and $\mathbf{\Phi}$ are found, the elevation and azimuth angle are
\begin{eqnarray}                   \label{eq:RA10}
\widehat{\varphi}_q = \arctan \left(\frac{\angle \Phi_{qq}}{\angle \Psi_{qq}}    \right),~~~ \widehat{\theta}_q= \arcsin \left(\frac{\angle \Phi_{qq}}{\pi \sin \varphi_q} \right).
\end{eqnarray}
Clearly, the parameters $\widehat{\varphi}_q $ and $\widehat{\theta}_q$ are paired automatically.

\subsection{Joint range-Doppler estimation}
Next, we estimate the range $r_q$ and velocity $\nu_q$. Using (\ref{eq:RA5}), denote the first $2L_s$ rows of $\widetilde{\mathbf{C}}$ and the first $2L_t$ rows of $\widetilde{\mathbf{B}}$ by $\widetilde{\mathbf{C}}_1$ and $\widetilde{\mathbf{B}}_1$, respectively. The last $2L_s$ rows of $\widetilde{\mathbf{C}}$ and the last $2L_t$ rows of $\widetilde{\mathbf{B}}$ are similarly denoted by $\widetilde{\mathbf{C}}_2$ and $\widetilde{\mathbf{B}}_2$, respectively. Note the relationships 
$\widetilde{\mathbf{B}}_2 = \widetilde{\mathbf{B}}_1  \mathbf{\Gamma}$ and  $\widetilde{\mathbf{C}}_2 = \widetilde{\mathbf{C}}_1  \mathbf{\Omega}$, 
where 
$\mathbf{\Gamma} = \mathrm{diag}[e^{\mathrm{j}4\pi \nu_1 T/\lambda_b},\ldots,e^{\mathrm{j}4\pi \nu_Q T/\lambda_b}]$ 
and 
$\mathbf{\Omega} = \mathrm{diag}[e^{\mathrm{j}4\pi \Delta f r_1/c},\ldots,e^{\mathrm{j}4\pi \Delta f r_Q/c}]$.

Define
\begin{align}             \label{eq:RA11}
\mathbf{V}_{B_1}^H \triangleq \mathbf{T}_R (\widetilde{\mathbf{C}}\odot \widetilde{\mathbf{B}}_1)^T,
\end{align}
\begin{align}             \label{eq:RA12}
\mathbf{V}_{B_2}^H \triangleq \mathbf{T}_R (\widetilde{\mathbf{C}}\odot \widetilde{\mathbf{B}}_2)^T &= \mathbf{T}_R[\widetilde{\mathbf{C}}\odot (\widetilde{\mathbf{B}}_1 \mathbf{\Gamma})]^T = \mathbf{T}_R \mathbf{\Gamma}^T(\widetilde{\mathbf{C}}\odot \widetilde{\mathbf{B}}_1)^T \nonumber\\
&= \mathbf{T}_R \mathbf{\Gamma} \mathbf{T}_R^{-1}\mathbf{V}_{B_1}^H, 
\end{align}
\begin{align}              \label{eq:RA13}
\mathbf{V}_{C_1}^H \triangleq \mathbf{T}_R (\widetilde{\mathbf{C}}_1 \odot \widetilde{\mathbf{B}})^T,
\end{align}
and
\begin{align}                             \label{eq:RA14}
\mathbf{V}_{C_2}^H \triangleq \mathbf{T}_R (\widetilde{\mathbf{C}}_2 \odot \widetilde{\mathbf{B}})^T&= \mathbf{T}_R[(\widetilde{\mathbf{C}}_1 \mathbf{\Omega}) \odot \widetilde{\mathbf{B}}]^T = \mathbf{T}_R \mathbf{\Omega}^T(\widetilde{\mathbf{C}}_1\odot \widetilde{\mathbf{B}})^T \nonumber\\
&= \mathbf{T}_R \mathbf{\Omega} \mathbf{T}_R^{-1}\mathbf{V}_{C_1}^H.
\end{align}
Then, based on (\ref{eq:RA12}) and (\ref{eq:RA14}), we have
\begin{eqnarray}                \label{eq:RA15}
\mathbf{V}_{B_2}^H (\mathbf{V}_{B_1}^H)^{\dag}= \mathbf{T}_R\mathbf{\Gamma} \mathbf{T}_R^{-1},
\end{eqnarray}
\begin{eqnarray}                \label{eq:RA16}
\mathbf{V}_{C_2}^H (\mathbf{V}_{C_1}^H)^{\dag}= \mathbf{T}_R\mathbf{\Omega} \mathbf{T}_R^{-1}.
\end{eqnarray}
Performing an eigenvalue decomposition of $\mathbf{V}_{B_2}^H (\mathbf{V}_{B_1}^H)^{\dag}$ produces the resulting matrices $\widehat{\mathbf{T}}_R$ and $\widehat{\mathbf{\Gamma}}$. Compute $\widehat{\mathbf{\Omega}}$ as 
$\widehat{\mathbf{\Omega}} = \widehat{\mathbf{T}}_R^{-1} (\mathbf{V}_{C_2}^H (\mathbf{V}_{C_1}^H)^{\dag}) \widehat{\mathbf{T}}_R$. 
Once $\mathbf{\Gamma}$ and $\mathbf{\Omega}$ are found, the velocity and range are obtained as
\begin{eqnarray}                         \label{eq:RA23}
\widehat{\nu}_q = \frac{\angle \Gamma_{qq}}{4\pi T/\lambda_b}   ,~~~ \widehat{r}_q=  \frac{\angle \Omega_{qq}}{4\pi \Delta f/c}.
\end{eqnarray}
Clearly, the parameters $\widehat{\nu}_q$ and $\widehat{r}_q$ are paired automatically. 

\subsection{Pairing of 2-D DoA with range-Doppler}
The estimates $(\widehat{\theta}_q, \widehat{\varphi}_q), q=1,2,\ldots,Q$ in (\ref{eq:RA10}) are not yet paired with $(\widehat{\nu}_i,\widehat{r}_i),i=1,2,\ldots,Q$ in (\ref{eq:RA23}). This is because $\mathbf{T}_R$ in the above computations does not show dependence on $\mathbf{T}_L$. We now prove that it is possible to obtain $\mathbf{T}_R$ from $\mathbf{T}_L$ thereby leading to an \textit{ipso facto} pairing between $(\widehat{\theta}_q, \widehat{\varphi}_q), q=1,2,\ldots,Q$ and $(\widehat{\nu}_i,\widehat{r}_i),i=1,2,\ldots,Q$.

By invoking (\ref{eq:RA2}), (\ref{eq:RA3}) and (\ref{eq:RA5}), we have
\begin{align}
\mathbf{R}_{XZ}=
 \left[\begin{matrix}
 \widetilde{\mathbf{A}}_x  \\
 \widetilde{\mathbf{A}}_z
\end{matrix} \right] \mathbf{R}_{\bm{\rho}} (\widetilde{\mathbf{C}}\odot \widetilde{\mathbf{B}})^T  
= \mathbf{U}_1 \mathbf{\Lambda} \mathbf{V}_1^H = \left[\begin{matrix}
 \widetilde{\mathbf{A}}_x  \\
 \widetilde{\mathbf{A}}_z
\end{matrix} \right] \mathbf{T}_L \mathbf{\Lambda} \mathbf{T}_R (\widetilde{\mathbf{C}}\odot \widetilde{\mathbf{B}})^T.
\end{align}
It follows that
\begin{eqnarray}                        \label{eq:RA26}
\mathbf{T}_L \mathbf{\Lambda} \mathbf{T}_R = \mathbf{R}_{\bm{\rho}},
\end{eqnarray}
which implies that
\begin{eqnarray}
\mathbf{T}_R =\mathbf{\Lambda}^{-1}\mathbf{T}_L^{-1}\mathbf{R}_{\bm{\rho}}.
\end{eqnarray}
Substituting (\ref{eq:RA26}) into (\ref{eq:RA15}) and (\ref{eq:RA16}) yields
\begin{align}                        \label{eq:RA27}
\mathbf{V}_{B_2}^H (\mathbf{V}_{B_1}^H)^{\dag}= \mathbf{T}_R\mathbf{\Gamma} \mathbf{T}_R^{-1}&= \mathbf{\Lambda}^{-1}\mathbf{T}_L^{-1}\mathbf{R}_{\bm{\rho}} \mathbf{\Gamma} \mathbf{R}_{\bm{\rho}}^{-1} \mathbf{T}_L\mathbf{\Lambda} \nonumber\\
&= \mathbf{\Lambda}^{-1}\mathbf{T}_L^{-1}\widetilde{\mathbf{\Gamma}}  \mathbf{T}_L\mathbf{\Lambda},
\end{align}
and
\begin{align}                        \label{eq:RA28}
\mathbf{V}_{C_2}^H (\mathbf{V}_{C_1}^H)^{\dag}= \mathbf{T}_R\mathbf{\Omega} \mathbf{T}_R^{-1} &= \mathbf{\Lambda}^{-1}\mathbf{T}_L^{-1}\mathbf{R}_{\bm{\rho}} \mathbf{\Omega} \mathbf{R}_{\bm{\rho}}^{-1} \mathbf{T}_L\mathbf{\Lambda}\nonumber\\
&=\mathbf{\Lambda}^{-1}\mathbf{T}_L^{-1} \widetilde{\mathbf{\Omega}}\mathbf{T}_L\mathbf{\Lambda},
\end{align}
where $\widetilde{\mathbf{\Gamma}} \triangleq \mathbf{R}_{\bm{\rho}} \mathbf{\Gamma} \mathbf{R}_{\bm{\rho}}^{-1}$ and $\widetilde{\mathbf{\Omega}} \triangleq \mathbf{R}_{\bm{\rho}} \mathbf{\Omega} \mathbf{R}_{\bm{\rho}}^{-1}$. Because $\mathbf{R}_{\bm{\rho}}$ is a real-valued diagonal matrix, we have $\angle{\widetilde{\Gamma}_{qq}}=\angle{\Gamma}_{qq}$ and $\angle{\widetilde{\Omega}_{ii}}=\angle{\Omega}_{ii}$. Rewriting (\ref{eq:RA27}) and (\ref{eq:RA28}) yields
\begin{eqnarray}                 \label{eq:RA31}
\widetilde{\mathbf{\Gamma}}=\mathbf{T}_L\mathbf{\Lambda} \mathbf{V}_{B_2}^H (\mathbf{V}_{B_1}^H)^{\dag}\mathbf{\Lambda}^{-1}\mathbf{T}_L^{-1},
\end{eqnarray}
\begin{eqnarray}                \label{eq:RA32}
\widetilde{\mathbf{\Omega}} =\mathbf{T}_L\mathbf{\Lambda}\mathbf{V}_{C_2}^H (\mathbf{V}_{C_1}^H)^{\dag}\mathbf{\Lambda}^{-1}\mathbf{T}_L^{-1}.
\end{eqnarray}
The velocity and range which corresponds to the estimated elevation angle and azimuth angle $(\widehat{\theta}_q,\widehat{\varphi}_q)$ in (\ref{eq:RA10}) are now computed as
\begin{eqnarray}
\widehat{\nu}_q = \frac{\angle \widehat{\widetilde{{\Gamma}}}_{qq}}{4\pi T/\lambda_b}   ,~~~ \widehat{r}_q=  \frac{\angle \widehat{\widetilde{{\Omega}}}_{qq}}{4\pi \Delta f/c}.
\end{eqnarray}

Algorithm~\ref{alg:c3_recovery} summarizes these steps of our CCing algorithm.

\begin{algorithm}[H]
	\caption{\textit{C-C}ube auto-pair\textit{ing} parameter retrieval (CCing) 
	}
	\label{alg:c3_recovery}
	\begin{algorithmic}[1]
		\Statex \textbf{Input:} Received signal $\widetilde{\mathbf{r}}_x$ ($\widetilde{\mathbf{r}}_z$) in the space-time-frequency coarray domain
		\Statex \textbf{Output:} $\widehat{r}_q$, $\widehat{\nu}_q$, $\widehat{\theta}_q$, $\widehat{\varphi}_q$, $\forall q$
		\State Construct $\mathbf{X}, \mathbf{Z} \in \mathbb{C}^{(2L_s+1)\times (2L_s+1)(2L_t+1)}$ such that  $\widetilde{\mathbf{r}}_x=\mathrm{vec}(\mathbf{X})$, $\widetilde{\mathbf{r}}_z=\mathrm{vec}(\mathbf{Z})$
		
		\State $\mathbf{R}_{XZ} \leftarrow [\mathbf{X}^T ~ \mathbf{Z}^T]^T$. 
		
		\State $\mathbf{U}_1$ ($\mathbf{V}_1$) $\leftarrow$ $Q$ columns from $\mathbf{U}$ ($\mathbf{V}$) corresponding to the largest $Q$ singular values of $\mathbf{R}_{XZ}= \mathbf{U \Lambda} \mathbf{V}^H$ (via SVD). 
		\Statex $\mathbf{U}_{11}$ ($\mathbf{U}_{12}$) $\leftarrow$ the first (last) $2L_s+1$ rows of $\mathbf{U}_1$.
		
		\Statex
		\Statex $\triangleright$ \textbf{2-D DoA Estimation:}
		\State $\mathbf{U}_{111}$ ($\mathbf{U}_{112}$) $\leftarrow$ the first (last) $2L_s$ rows of $\mathbf{U}_{11}$. 
		\Statex $\mathbf{U}_{121}$ ($\mathbf{U}_{122}$) $\leftarrow$ the first (last) $2L_s$ rows of $\mathbf{U}_{12}$.
		
		\State $\mathbf{\Psi} \leftarrow \mathbf{T}_L (\mathbf{U}_{121}^{\dag} \mathbf{U}_{122}) \mathbf{T}_L^{-1}$, where $\mathbf{U}_{111}^{\dag} \mathbf{U}_{112} = \mathbf{T}_L^{-1} \mathbf{\Phi} \mathbf{T}_L$ (via eigenvalue decomposition). 
		
		\State $\widehat{\varphi}_q = \arctan \left(\frac{\angle \Phi_{qq}}{\angle \Psi_{qq}}    \right),~~~ \widehat{\theta}_q= \arcsin \left(\frac{\angle \Phi_{qq}}{\pi \sin \varphi_q} \right),\; q=1,2,\ldots,Q$.
		
		\Statex
		\Statex  $\triangleright$ \textbf{Range-Doppler Estimation and Pairing with 2-D DoA:}
		
		\State Construct $\mathbf{V}_{B_1}^H$, $\mathbf{V}_{B_2}^H$, $\mathbf{V}_{C_1}^H$ and $\mathbf{V}_{C_2}^H$ from $\mathbf{V}_1^H$ via (\ref{eq:RA11})- (\ref{eq:RA14})
		
		\State Compute $\widetilde{\mathbf{\Gamma}}$, $\widetilde{\mathbf{\Omega}}$ via (\ref{eq:RA31}) and (\ref{eq:RA32}). $\triangleright$ Estimation based on $\mathbf{T}_L$, which embodies the permutation order of 2-D DoA.
		
		\State $\widehat{\nu}_q = \frac{\angle \widetilde{{\Gamma}}_{qq}}{4\pi T/\lambda_b},~ \widehat{r}_q=  \frac{\angle \widetilde{{\Omega}}_{qq}}{4\pi \Delta f/c},\; q=1,2,\ldots,Q$. 
	\end{algorithmic}
\end{algorithm}


\section{Performance Analyses}                     
\label{Section_PA}
The error in joint estimation of angle-range-Doppler of C-Cube FDA is characterized by lower error bounds such as CRB. 
From (\ref{eq:SM8}) and (\ref{eq:SM24}), we rearrange the entries of vector $\left[ \mathbf{x}^T(t_s)~ \mathbf{z}^T(t_s) \right]^T$, leading to $2KP_s(2P_s-1) \times 1$ vector
\begin{align}
\mathbf{y}(t_s) &\triangleq \left(\mathbf{C} \odot \mathbf{B} \odot 
    \left[  \begin{matrix}
             \mathbf{A}_x \\
             \mathbf{A}_z  
             \end{matrix} 
    \right]\right) \bm{\rho} (t_s) + \mathbf{n}(t_s)   \notag \\
&=(\mathbf{C} \odot \mathbf{B} \odot \mathbf{A}_{xz})  \bm{\rho} (t_s) +\mathbf{n}(t_s),
\end{align}
where $\mathbf{A}_{xz}=[ \mathbf{A}_x^T ~~  \mathbf{A}_z^T]^T$ and $t_s=1, \ldots, L_r$. Denote the unknown parameters as $\bm{\gamma}=[\bm{\theta}^T~ \bm{\varphi}^T ~ \bm{r}^T~ \bm{\nu}^T]^T$, where $\bm{\theta}=[\theta_1,\ldots,\theta_Q]^T$, $\bm{\varphi}=[\varphi_1,\ldots,\varphi_Q]^T$, $\bm{r}=[r_1,\ldots,r_Q]^T$, and $\bm{\nu}=[\nu_1,\ldots,\nu_Q]^T$. The data vectors $\{\mathbf{y}(t_s)\}_{t_s=1}^{L_r}$ are i. i. d Gaussian random variables, distributed as
\begin{align}               \label{eq:PA_CRB2}
\mathbf{y}(t_s) \sim \mathcal{N}\left(0, \mathbf{R}(\bm{\gamma})  \right),
\end{align}
where $ \mathbf{R}(\bm{\gamma}) \triangleq (\mathbf{C} \odot \mathbf{B} \odot \mathbf{A}_{xz})\mathbf{R}_{\bm{\rho}}(\mathbf{C} \odot \mathbf{B} \odot \mathbf{A}_{xz})^H + \sigma_n^2 \mathbf{I}_{2KP_s(2P_s-1)}$. In the sequel, we assume signal variances $\{\sigma_q^2 \}_{q=1}^Q$ and noise variance $\sigma_n^2$ which are defined in (\ref{eq:SM15}) and (\ref{eq:SM7}), respectively, are known \textit{a priori}.

Then, the entries of the Fisher information matrix (FIM) $\mathbf{J}(\bm{\gamma})$ \cite{MWang2017} are 
\begin{align}              \label{eq:PA_CRB4}
\frac{1}{L_r}[\mathbf{J}(\bm{\gamma})]_{m,n} &= \mathrm{vec}^H\left( \frac{\partial \mathbf{R}(\bm{\gamma})}{\partial \gamma_m} \right) \mathbf{W} \mathrm{vec} \left( \frac{\partial \mathbf{R}(\bm{\gamma})}{\partial \gamma_n} \right)       \notag \\
&= \left( \frac{\partial \mathrm{vec}(\mathbf{R}(\bm{\gamma}))}{\partial \gamma_m}   \right)^H  \mathbf{W} \left(  \frac{\partial \mathrm{vec}(\mathbf{R}(\bm{\gamma}))}{\partial \gamma_n}     \right),
\end{align}
where $\mathbf{W}=  \mathbf{R}^{-T}(\bm{\gamma}) \otimes \mathbf{R}^{-1}(\bm{\gamma})$ and $\gamma_n$ is the $n$-th entry of $\bm{\gamma}$, $1\leq n \leq 4Q$. 
The CRB of the $n$-th unknown parameter $\gamma_n$ is 
\begin{align}                     \label{eq:PA_CRB5}
\mathrm{CRB}(\gamma_n) = \left[ \mathbf{J}^{-1}(\bm{\gamma}) \right]_{n,n}.
\end{align}

Vectorizing $\mathbf{R}(\bm{\gamma})$ leads to
\begin{align}
\mathbf{r}_{xz} &= \mathrm{vec}(\mathbf{R}(\bm{\gamma}))                             \notag \\
&= \mathbf{K}_{xz} \left[ (\mathbf{C}^*\odot \mathbf{C})\odot (\mathbf{B}^*\odot \mathbf{B}) \odot (\mathbf{A}_{xz}^* \odot \mathbf{A}_{xz})    \right] \mathbf{r}_{\bm{\rho}}                                   \notag \\
&~~ + \sigma_n^2 \mathrm{vec}(\mathbf{I}_{2KP_s(2P_s-1)}),
\end{align}
where $\mathbf{K}_{xz} \in \mathbb{C}^{4K^2 P_s^2(2P_s-1)^2 \times 4K^2 P_s^2(2P_s-1)^2}$ is a known permutation matrix. 

Define 
\begin{align}
\mathbf{V}_{\bm{\theta}} & \triangleq \frac{\partial \mathbf{r}_{xz}}{ \partial \bm{\theta}^T}         \notag \\
&= \mathbf{K}_{xz} \left[ (\mathbf{C}^*\odot \mathbf{C}) \odot (\mathbf{B}^*\odot \mathbf{B}) \odot \frac{\partial (\mathbf{A}_{xz}^* \odot \mathbf{A}_{xz})}{\partial \bm{\theta}^T}    \right] \mathbf{R}_{\bm{\rho}}                            \notag \\
&= \mathbf{K}_{xz} \left\{  (\mathbf{C}^*\odot \mathbf{C}) \odot (\mathbf{B}^*\odot \mathbf{B}) \odot \left[ (\mathbf{A}_{xz}^* \odot \mathbf{A}_{xz})            \right. \right.                          \notag    \\
&~~\left.\left.   \diamond \frac{\ln(\mathbf{A}_{xz}^* \odot \mathbf{A}_{xz})}{\mathrm{diag}(\tan(\bm{\theta}^T))}          \right] \right\}  \mathbf{R}_{\bm{\rho}},
\end{align}
\begin{align}                      \label{eq:PA_CRB8}
\mathbf{V}_{\bm{\varphi}} & \triangleq \frac{\partial \mathbf{r}_{xz}}{ \partial \bm{\varphi}^T}         \notag \\
&= \mathbf{K}_{xz} \left[ (\mathbf{C}^*\odot \mathbf{C}) \odot (\mathbf{B}^*\odot \mathbf{B}) \odot \frac{\partial (\mathbf{A}_{xz}^* \odot \mathbf{A}_{xz})}{\partial \bm{\varphi}^T}    \right] \mathbf{R}_{\bm{\rho}}                            \notag \\
&= \mathbf{K}_{xz} \left[ (\mathbf{C}^*\odot \mathbf{C}) \odot (\mathbf{B}^*\odot \mathbf{B}) \odot   \left(\frac{\partial \mathbf{A}^*_{xz}(\bm{\varphi})}{\partial \bm{\varphi}^T}  \odot \mathbf{A}_{xz}  \right. \right.                                \notag \\
&~~ \left.\left. +   \mathbf{A}_{xz}^* \odot    \frac{\partial \mathbf{A}_{xz}(\bm{\varphi})}{\partial \bm{\varphi}^T} \right)        \right],
\end{align}
\begin{align}
\mathbf{V}_{\bm{r}} & \triangleq  \frac{\partial \mathbf{r}_{xz}}{ \partial \bm{r}^T}         \notag \\
&= \mathbf{K}_{xz}  \left[ \frac{\partial (\mathbf{C}^*\odot \mathbf{C})}{\partial \bm{r}^T} \odot (\mathbf{B}^*\odot \mathbf{B}) \odot (\mathbf{A}_{xz}^* \odot \mathbf{A}_{xz})  \right] \mathbf{R}_{\bm{\rho}}                                 \notag \\
&=  \mathbf{K}_{xz} \left[ (\mathbf{C}^*\odot \mathbf{C}) \diamond \frac{\ln(\mathbf{C}^*\odot \mathbf{C})}{ \mathrm{diag}(\bm{r})}   \odot (\mathbf{B}^*\odot \mathbf{B})  \right. \notag \\
& ~~ \left.\odot (\mathbf{A}_{xz}^* \odot \mathbf{A}_{xz}) \right]  \mathbf{R}_{\bm{\rho}} ,  
\end{align}
and
\begin{align}
\mathbf{V}_{\bm{\nu}} & \triangleq \frac{\partial \mathbf{r}_{xz}}{ \partial \bm{\nu}^T}         \notag \\
&=\mathbf{K}_{xz}  \left[ (\mathbf{C}^*\odot \mathbf{C}) \odot  \frac{\partial (\mathbf{B}^*\odot \mathbf{B})}{\partial \bm{\nu}^T}  \odot (\mathbf{A}_{xz}^* \odot \mathbf{A}_{xz})  \right] \mathbf{R}_{\bm{\rho}}                                 \notag \\
&=\mathbf{K}_{xz}  \left\{ (\mathbf{C}^*\odot \mathbf{C}) \odot  \left[(\mathbf{B}^*\odot \mathbf{B}) \diamond \frac{\ln(\mathbf{B}^*\odot \mathbf{B})}{\mathrm{diag}(\bm{\nu})} \right]  \right.         \notag \\
&~~ \left. \odot (\mathbf{A}_{xz}^* \odot \mathbf{A}_{xz}) \right\}       \mathbf{R}_{\bm{\rho}},                        \notag \\ 
\end{align}
where $ \frac{\partial \mathbf{A}_{xz}(\bm{\varphi})}{\partial \bm{\varphi}^T}$ in (\ref{eq:PA_CRB8}) is
\begin{align}
 \frac{\partial \mathbf{A}_{xz}(\bm{\varphi})}{\partial \bm{\varphi}^T}= 
 \left[  \begin{matrix}
             \frac{\partial \mathbf{A}_{x}(\bm{\varphi})}{\partial \bm{\varphi}^T} \\
             \frac{\partial \mathbf{A}_{z}(\bm{\varphi})}{\partial \bm{\varphi}^T}
             \end{matrix} 
    \right]
=
\left[  \begin{matrix}
              \frac{\mathbf{A}_x \diamond \ln (\mathbf{A}_x)}{\mathrm{diag}(\tan(\bm{\varphi}^T))} \\
              \mathbf{A}_z \diamond \ln (\mathbf{A}_z) \mathrm{diag}(\tan(\bm{\varphi}^T))
             \end{matrix} 
    \right].
\end{align}

\begin{theorem}                                            \label{Theo_CRB}
Given a set of $Q$ targets with unknown parameters $\bm{\gamma}=[\bm{\theta}^T~ \bm{\varphi}^T ~ \bm{r}^T~ \bm{\nu}^T]^T \in \mathbb{C}^{4Q\times 1}$ and the received signal model of co-pulsing FDA radar as in (\ref{eq:PA_CRB2}), define the derivatives of $\mathbf{r}_{xz}$ with respect to 2D-DoA and range-doppler as $\mathbf{D}_L \triangleq [\mathbf{V}_{\bm{\theta}} ~ \mathbf{V}_{\bm{\varphi}}]$ and $\mathbf{D}_R \triangleq  [\mathbf{V}_{\bm{r}} ~ \mathbf{V}_{\bm{\nu}}]$, respectively. Then, the CRBs of $\bm{\theta}$, $\bm{\varphi}$, $\bm{r}$ and $ \bm{\nu}$ exist and have the forms
\begin{align}                      \label{eq:PA_CRB13}
\mathrm{CRB}(\bm{\theta}) &=\frac{1}{L_r} \left\{ \mathbf{V}_{\bm{\theta}}^H \mathbf{W}^{1/2} \mathbf{\Pi}_{\mathbf{W}^{1/2}\mathbf{D}_R}^{\bot}  \left(\mathbf{\Pi}_{\mathbf{\Pi}_{\mathbf{W}^{1/2}\mathbf{D}_R}^{\bot}  \mathbf{W}^{1/2}\mathbf{V}_{\bm{\varphi}}}^{\bot}\right) \right. \notag  \\
&~~ \times \left. \mathbf{\Pi}_{\mathbf{W}^{1/2}\mathbf{D}_R}^{\bot}  \mathbf{W}^{1/2}\mathbf{V}_{\bm{\theta}} \right\}^{-1},
\end{align}
\begin{align}
\mathrm{CRB}(\bm{\varphi})  &= \frac{1}{L_r} \left\{ \mathbf{V}_{\bm{\varphi}}^H \mathbf{W}^{1/2} \mathbf{\Pi}_{\mathbf{W}^{1/2}\mathbf{D}_R}^{\bot}  \left(\mathbf{\Pi}_{\mathbf{\Pi}_{\mathbf{W}^{1/2}\mathbf{D}_R}^{\bot}  \mathbf{W}^{1/2}\mathbf{V}_{\bm{\theta}}}^{\bot}\right) \right. \notag  \\
&~~ \times \left. \mathbf{\Pi}_{\mathbf{W}^{1/2}\mathbf{D}_R}^{\bot}  \mathbf{W}^{1/2}\mathbf{V}_{\bm{\varphi}} \right\}^{-1},
\end{align}
\begin{align}
\mathrm{CRB}(\bm{r})  &=\frac{1}{L_r} \left\{ \mathbf{V}_{\bm{r}}^H \mathbf{W}^{1/2} \mathbf{\Pi}_{\mathbf{W}^{1/2}\mathbf{D}_L}^{\bot}  \left(\mathbf{\Pi}_{\mathbf{\Pi}_{\mathbf{W}^{1/2}\mathbf{D}_L}^{\bot}  \mathbf{W}^{1/2}\mathbf{V}_{\bm{\nu}}}^{\bot}\right) \right. \notag  \\
&~~ \times \left. \mathbf{\Pi}_{\mathbf{W}^{1/2}\mathbf{D}_L}^{\bot}  \mathbf{W}^{1/2}\mathbf{V}_{\bm{r}} \right\}^{-1},
\end{align}
and
\begin{align}                            \label{eq:PA_CRB16}
\mathrm{CRB}(\bm{\nu})&=\frac{1}{L_r} \left\{ \mathbf{V}_{\bm{\nu}}^H \mathbf{W}^{1/2} \mathbf{\Pi}_{\mathbf{W}^{1/2}\mathbf{D}_L}^{\bot}  \left(\mathbf{\Pi}_{\mathbf{\Pi}_{\mathbf{W}^{1/2}\mathbf{D}_L}^{\bot}  \mathbf{W}^{1/2}\mathbf{V}_{\bm{r}}}^{\bot}\right) \right. \notag  \\
&~~ \times \left. \mathbf{\Pi}_{\mathbf{W}^{1/2}\mathbf{D}_L}^{\bot}  \mathbf{W}^{1/2}\mathbf{V}_{\bm{\nu}} \right\}^{-1},
\end{align}
if and only if $[\mathbf{D}_L~ \mathbf{D}_R]$ has the full column rank, namely $\mathrm{rank}([\mathbf{D}_L~ \mathbf{D}_R]) =4Q$.
\end{theorem}
\begin{proof}
See Appendix~\ref{app:proof_theo_CRB}.
\end{proof}

It follows from the closed form of CRBs in (\ref{eq:PA_CRB13})-(\ref{eq:PA_CRB16}) that the azimuth angle, elevation angle, range, and Doppler velocity interact with each other. This implies that the coupling among these parameters has an effect on the estimation performance.

In addition, we also derive the guarantees for recovering targets. In particular, we discuss conditions on the number of antennas and pulses required to retrieve unknown parameters of $Q$ far-field targets based on the properties of the matrices $\mathbf{T}_L$ and $\mathbf{T}_R$ obtained from the concatenated covariance matrix $\mathbf{R}_{XZ}$ of measurements. This general result then leads to similar guarantees for other L-shaped co-pulsing FDAs and non-FDAs mentioned in Table~\ref{tbl:summary}. 
The following Theorem~\ref{Theo_rec} provides lower bounds on the number of antennas and pulses required by the C-CUBE radar to guarantee perfect recovery of the unknown parameter set $\bm{\gamma}_q=\{\theta_q, \varphi_q,r_q,\nu_q\}_{q=1}^{Q}$ from $\mathbf{R}_{XZ}$. 

\begin{theorem}                \label{Theo_rec}
  (C-Cube:  L-shaped co-prime array with co-prime FO and co-prime pulsing) Consider an L-shaped co-prime array with $P_s=N_s+2M_s-1$ sensors along each axis and co-prime FOs. Each sensor transmits at co-prime PRI with a total of $K=N_t+2M_t-1$ pulses in a CPI. The fundamental spatial spacing and fundamental PRI are $d$  and $T$, respectively. If \upshape \textbf{C1}-\textbf{C5} \itshape hold, then 
the unknown parameter set $\bm{\gamma}_q=\{\theta_q, \varphi_q,r_q,\nu_q\}_{q=1}^{Q}$ of $Q$ far-field targets are perfectly recovered from $\mathbf{R}_{XZ}$ with the lower bounds on the number of physical sensor elements and the number of transmit pulses as, respectively,
\begin{align}
    P_s > 2\sqrt{Q+1}-2 \textrm{\upshape{ and }} K > 2\sqrt{Q+1} -2.
\end{align}
\end{theorem}
\begin{IEEEproof}
It follows from Proposition \ref{Pro_TL} that $\mathbf{U}_{111}$ and $\mathbf{U}_{121}$ are both full column rank and, therefore, left-invertible under \textbf{C1} and \textbf{C4}. So, we can recover $\mathbf{\Phi}$ and $\mathbf{\Psi}$ uniquely based on (\ref{eq:RA7_1}) and (\ref{eq:RA7_2}) using the same permutation matrix $\mathbf{T}_L$. Then, \textbf{C1} implies that both $\angle{\Phi_{qq}} = \pi \sin\theta_q \sin\varphi_q$ and $\angle{\Psi_{qq}}= \pi \sin\theta_q \cos\varphi_q, 1\leq q\leq Q$ are unique. Consequently, the parameters $(\theta_q, \varphi_q), 1\leq q\leq Q$ can also be recovered uniquely. Similarly, from Proposition \ref{Pro_TL}, $\mathbf{V}_{B_1}$ and $\mathbf{V}_{C_1}$ are both full column rank under \textbf{C2}, \textbf{C3}, and \textbf{C4}. Thus, $\mathbf{V}_{B_1}^H$ and $\mathbf{V}_{C_1}^H$ are right-invertible. So, $\mathbf{\Gamma}$ and $\mathbf{\Omega}$ can be recovered uniquely based on (\ref{eq:RA15}) and (\ref{eq:RA16}) using the permutation matrix $\mathbf{T}_R$. In addition, under \textbf{C2} and \textbf{C3}, both $\angle{\Gamma_{qq}} = 4\pi\nu_q T/{\lambda_b}$ and $\angle{\Omega_{qq}}= 4\pi r_q \Delta f /c, 1\leq q\leq Q$ are unique thereby guaranteeing unique recovery of the parameters $(r_q, \nu_q), 1\leq q\leq Q$. Using the procedure outline in (\ref{eq:RA26})-(\ref{eq:RA32}), $(r_q, \nu_q), 1\leq q\leq Q$ are auto-paired with $(\theta_q, \phi_q), 1\leq q\leq Q$. Therefore, $\bm{\gamma}_q=\{\theta_q, \varphi_q,r_q,\nu_q\}_{q=1}^{Q}$ can be perfectly recovered from $\mathbf{R}_{XZ}$ based on \textbf{C1}-\textbf{C5}.

Next, 
the inequality of arithmetic and geometric means \cite{ZCHao1990} yields 
$2M_sN_s+2M_s =2M_s(N_s+1) \leq \frac{(2M_s+N_s+1)^2}{4}   = \frac{(P_s+2)^2}{4}$, 
which gives 
$M_sN_s+M_s-1 \leq \frac{(P_s+2)^2}{8} -1$.
Combining this with \textbf{C4} 
gives
\begin{align}
P_s>2\sqrt{Q+1}-2
\end{align}
The lower bound on physical pulses is obtained \textit{mutatis mutandis} as
\begin{align}
K>2\sqrt{Q+1}-2.
\end{align}
\end{IEEEproof}

Theorem~\ref{Theo_rec} shows that, in C-CUBE radar, both antennas and pulses need to be at least only $2\sqrt{Q+1}-1$ to guarantee perfect recovery of $Q$. In contrast, the L-shaped \textit{U}LA with \textit{u}niform pulsing (U-U) 
\cite{TKishigami2016} requires only $Q+1$ antennas in the same aperture for perfect recovery while transmitting $Q+1$ pulses in the same CPI. The L-shaped \textit{c}o-prime array with \textit{u}niform pulsing (C-U) \cite{AMElbir2020} requires the same number of pulses but only $2\sqrt{Q+1}-1$ antennas. Similarly, U-C and C-U (non-FDA) arrays may also be hypothesised (see Table~\ref{tbl:summary}). As for L-shaped FDA that employs uniform array, FO, and PRI, i.e. U-Cube, the following Corollary~\ref{Coro_UUU} is a direct consequence of Theorem~\ref{Theo_rec}.

\begin{corollary}                \label{Coro_UUU}
  (U-Cube: L-shaped ULA with uniform FO and uniform pulsing) Consider an L-shaped ULA with $P_s=N_s+2M_s-1$ sensors along each axis and uniform linear FOs. Each sensor transmits at uniform PRI with a total of $K=N_t+2M_t-1$ pulses in a CPI. The fundamental spatial spacing and fundamental PRI are $d$  and $T$, respectively. If \upshape\textbf{C1}-\textbf{C3} \itshape hold and
\begin{description} 
\item[C6] $P_s \triangleq N_s+2M_s-1 >  Q$,\vspace{4pt}
\item[C7] $K \triangleq N_t +2M_t-1 > Q$,
\end{description}
the unknown parameter set $\bm{\gamma}_q=\{\theta_q, \varphi_q,r_q,\nu_q\}_{q=1}^{Q}$ of $Q$ far-field targets are perfectly recovered from $\mathbf{R}_{XZ}$ with the lower bounds on the number of physical sensor elements and the number of transmit pulses as, respectively,
\begin{align}
    P_s > Q \textrm{\upshape{ and }} K > Q.
\end{align}
\end{corollary}
\begin{IEEEproof}
  Replace the virtual manifold matrices $\widetilde{\mathbf{A}}_x$, $\widetilde{\mathbf{A}}_z$, $\widetilde{\mathbf{C}}$, and $\widetilde{\mathbf{B}}$ with the corresponding physical manifold matrices $\mathbf{A}_x$, $\mathbf{A}_z$, $\mathbf{C}$ and $\mathbf{B}$, respectively. Then, \textit{ceteris paribus}, the result follows from repeating the steps of the proof in Theorem \ref{Theo_rec}. 
\end{IEEEproof}

As mentioned in Section~\ref{Section_SM}, we are concerned with the scenario that all parameters of targets are distinct. It follows from Theorem~\ref{Theo_rec} and Corollary~\ref{Coro_UUU} that, if the number of antennas along $x(z)$-axis and pulses is fixed as $P_s$ and $K$, respectively, the DOFs of C-Cube may reach $\mathcal{O}(\min\{(P_s+2)^2/4-1,(K+2)^2/4-1\})$ while the DoFs of uniform counterparts could be at most only $\mathcal{O}(\min\{(P_s,K\})$. If the number of targets is fixed, for large $Q$, C-Cube clearly outperforms U-Cube in the required number of sensors and pulses. If the number of potential targets is no more than $Q$, our proposed C-Cube radar transmits only $2\sqrt{Q+1}-1$ pulses using $2\sqrt{Q+1}-1$ antennas compared to the U-Cube that emits $Q+1$ pulses within the same CPI and places $Q+1$ antennas in the same aperture. Obviously, our system has lower power and pulse transmission rate, which reduces its interception probability when compared to U-Cube.  For similar results for UUC, CCU, UCU, CUC, CUU, and UCC L-shaped FDAs, we refer the reader to Corollary \ref{Coro_otherFDAs} in the Appendix~\ref{app:perf_other_arrays}. Further, the guarantees of non-FDA L-shaped arrays such as U-U, U-C, C-U, and C-C also follow from Theorem~\ref{Theo_rec} as in Corollary \ref{Coro_UU} of the Appendix~\ref{app:perf_other_arrays}. 

\section{Relationship with Other Configurations}
\label{sec:adv}
Several sparse alternatives are possible for the L-shaped FDA, FOs, and pulsing. Although different sparse configurations could be adopted for the arrays, FOs and pulses, the number of array elements must be the same as that of FOs. Here, we compare a few related configurations with the basic co-prime structure of Fig.~\ref{L-shaped}. Some of these sparse structures have been reported for 1-D, usually non-FDA, arrays. But their suitability for L-shaped FDAs remains unexamined. 

\subsection{Alternative L-shaped FDA coarrays}
\label{subsec:alt_array}
\begin{figure}[p]
\includegraphics[width=1.0\columnwidth]{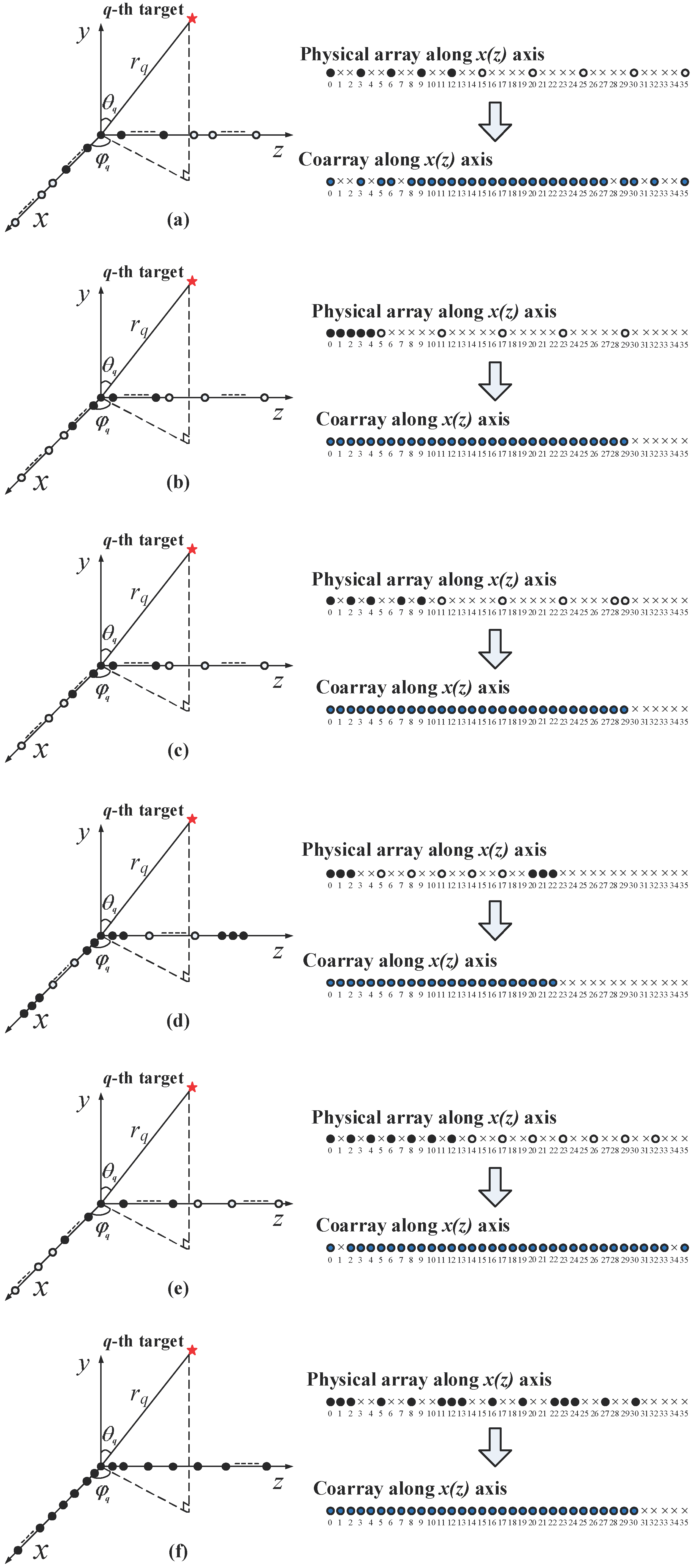}
\caption{Alternative sparse L-shaped FDA coarrays with the common physical aperture $35d$. Bullets denote sensors and crosses indicate empty locations. Black and grey bullets denote two uniform physical sparse subarrays. The virtual sensors of coarray are marked with blue bullets. (a) CADiS structure with $M=3$, $N=5$, $L=3$ (b) nested array with $N_1=N_2=N=5$ (c) super-nested array with $N_1=N_2=N=5$ (d) CNA structure with $N_1=N_2=N=3$ (e) GNA structure with $N_1=N_2=5$, $\alpha=2$, $\beta=3$; this reduces to a nested array when $\alpha=1$, $\beta=N$ (f) multi-coset array with coset pattern $\mathcal{P}=\{0,1,2,5 \}$ with block length of 7. }
\label{Other_L_FDA}
\end{figure}

A straightforward sparse array is achieved by randomly removing the elements from the ULA FDA \cite{sedighi2019optimum}. This may also be applied to an L-shaped FDA. It is capable of reducing the hardware costs and retaining a reasonable estimation performance. However, random arrays also lead to increased sidelobes \cite{lo1964mathematical}. Here, the DoFs are also restricted because a coarray is not exploited. To increase the spatial DoFs, it is pertinent to consider sparse coarray structures. 

The mutual coupling matrix of $2Q+1$-element filled L-shaped array, including FDA, is a $(2Q+1)\times (2Q+1)$ symmetric matrix 
$\mathbf{H} =
\left(
\begin{matrix}
\mathbf{H}_x & \mathbf{H}_{xz} \\
\mathbf{H}_{zx} & \mathbf{H}_{z} 
\end{matrix} \right)$,
where 
\begin{align}           
\mathbf{H}_x = \mathbf{H}_z=
\left(
\begin{matrix}
1 & h_1 & \cdots & h_{Q} \\
h_1 & 1 & \ddots  &   \vdots     \\
\vdots &  & \ddots  & h_1    \\
h_{Q} & \cdots & h_1 & 1
\end{matrix} \right)   \in \mathbb{C}^{(Q+1) \times (Q+1)},
\end{align}
is the self-coupling (symmetric and Toeplitz) matrix along each axis \cite{TSvantesson1999} and 
\begin{align}           
\mathbf{H}_{xz} = \mathbf{H}_{zx}=
\left(
\begin{matrix}
h_{z_1 x_Q} & h_{z_1 x_{Q-1}} & \cdots & h_{z_1 x_1} \\
h_{z_2 x_Q} & h_{z_2 x_{Q-1}} & \cdots  & h_{z_2 x_1}     \\
\vdots & \vdots & \ddots  & \vdots    \\
h_{z_Q x_Q} & h_{z_Q x_{Q-1}} & \cdots & h_{z_Q x_1}
\end{matrix} \right)   \in \mathbb{C}^{Q \times Q},
\end{align}
$\mathbf{H}_{xz}$ and $\mathbf{H}_{zx}$ are the cross-coupling matrices between the $x$- and $z$-axes. Recall the following definition. 
\begin{definition}[Coupling Leakage \cite{CLiu2016}]\label{eq:Def_CL} The coupling leakage of an L-shaped array with a mutual coupling matrix $H$ is defined as
    $\mathcal{L}_{\mathbf{H}} \triangleq \frac{\Vert \mathbf{H} -\mathrm{diag}(\mathbf{H}) \Vert_{\mathcal{F}}}{\Vert \mathbf{H}\Vert_{\mathcal{F}}}$.
\end{definition}
We set $h_1 = 0.3 e^{j\pi /3}$ and $h_i = \frac{h_1}{i}e^{-j(i-1)\pi/8}$. Then, the coupling leakage is computed using Definition~\ref{eq:Def_CL} for any given L-shaped array structure (see Table \ref{tbl:coupling}). 

Our adopted co-prime array is a special case of compressed inter-element spacing (CACIS) co-prime array \cite{SQin2015}, which doubles the number of sensors in a constituting subarray. To mitigate the mutual coupling among array elements, the coprime array with displaced subarrays (CADiS) L-shaped FDA structure is proposed \cite{SQin2015} (Fig.~\ref{Other_L_FDA}a), allowing the minimum inter-element spacing to be much larger than the typical half-wavelength requirement while it occupies larger space for the same number of physical elements.

The DoFs are also enhanced in a nested array \cite{pal2010nested}, which is a concatenation of
two ULAs: the inner and outer with $N_1$ and $N_2$ elements spaced at $d_1$ and $d_2$, respectively, such that $d_2=(N_1+1)d_1$. Fig.~\ref{Other_L_FDA}b shows nested array for an L-shaped structure. Note that setting $M=1$ and $L=N+1$ in the CADiS L-shaped FDA converts it to nested CADiS structure \cite{SQin2015}. The super-nested L-shaped array (Fig.~\ref{Other_L_FDA}c) retains the benefits of the nested array and also achieves reduced mutual coupling by redistributing the elements of the dense ULA portion of the nested array. Fig.~\ref{Other_L_FDA}d and e show the L-shaped concatenated nested array (CNA) \cite{RRobin2017} and generalized nested array (GNA) \cite{JShi2018}, respectively. While the former is a nested array concatenated with its mirror image, the latter enlarges the inter-element spacing of two concatenated uniform linear subarrays with two flexible co-prime factors. Both CNA and GNA enjoy the merit of higher DoFs using fewer elements. However, some of the physical elements are closely located in CNA leading to severe mutual coupling. The GNA alleviates this problem at the expense of space. This is also the case with the multi-coset sparse array \cite{KJames2014} (Fig.~\ref{Other_L_FDA}f), which is constructed by a collection of interleaved sparse uniform subarrays such that the elements are laid out in a periodic nonuniform pattern over the aperture.  

Our classical co-prime L-shaped structure is attractive because it offers a trade-off between high DoFs and mild mutual coupling arising from smaller aperture (see Table~\ref{tbl:coupling}). It may be viewed as the most fundamental coarray from which many arrays in Fig.~\ref{Other_L_FDA} are obtained. The receive processing steps of these alternative arrays are obtained by replacing $\mathbf{A}_x$ ($\mathbf{A}_z$) in (\ref{eq:SM8}) ((\ref{eq:SM24})) with the manifold matrix of the corresponding sparse array along each axis. 
\begin{table}[t]
    \caption{Coupling leakage level of different L-shaped FDAs 
    }
    \label{tbl:coupling}       
    \centering
    \begin{threeparttable}
    \begin{tabular}{p{4.5cm}P{3.0cm}}
    \hline\noalign{\smallskip}
    Array  &  $\mathcal{L}_{\mathbf{H}}$\tnote{a}
    \\
    \noalign{\smallskip}
    \hline
    \noalign{\smallskip}
    L-shaped uniform FDA (U-Cube) & 0.76     \\
    L-shaped CADiS FDA & 0.4302     \\
    L-shaped  nested FDA  & 0.6303  \\
    L-shaped super-nested FDA & 0.5342      \\
    L-shaped CNA FDA & 0.5859      \\
    L-shaped GNA FDA & 0.5969         \\
    L-shaped multi-coset FDA   &   0.5985    \\
    L-shaped Co-prime FDA (C-Cube) & 0.5340 \\
    \noalign{\smallskip}\hline\noalign{\smallskip}
    \end{tabular}
    \begin{tablenotes}[para]
     \item[a] Computed for common physical aperture of $35d$ along each axis.
    \end{tablenotes} 
    \end{threeparttable}
    \end{table}

\subsection{Alternatives to co-prime FOs}
\label{subsec:alt_FO}
Co-prime FO is a classical ``spectrum saving'' co-FO sequence. In Table.~\ref{tbl:summary}, when the number of FOs is same across all L-shaped structures, the co-FO has larger cumulative transmit bandwidth than the linear FO. Denote their common frequency occupancy by $B_a$. From another perspective, considering the operating bandwidth of a radar is usually fixed in practice, assume that the cumulative transmit bandwidth of L-shaped array is the same and given by $B_c=2(L_f-1) \Delta f$. 
Define the spectrum occupancy rate as 
    $\eta=\frac{B_a}{B_c}$. 
Smaller values of $\eta$ imply that more spectral resources are free to be used for other purposes \cite{na2018tendsur}.

The literature suggests alternatives such as logarithmic FOs  \cite{khan2014frequency,YLiao2019}, time-dependent FOs (TDFOs) \cite{WKhan2014} and random FOs \cite{liu2016random}. In each case, the co-prime FO has lower frequency occupancy under the same cumulative transmit bandwidth. For example,
the frequency fed to the $m$-th element in logarithmic FO configuration is 
    $f_m= f_b+\Delta f_m$, 
where the FO of the $m$-th element with reference to the carrier frequency $f_b$ is
\begin{align}
    \Delta f_m= 
    \begin{cases}
\log(m+1) \Delta f,~& m\ge 0, \\
-\log(-m+1) \Delta f,~& m<0.
\end{cases}
\end{align}
It follows that the inter-element FO decreases with the increase in the absolute value of $m$. If $\Delta f_{m+1}-\Delta f_m \ge B$ holds for all $m$, then the spectrum occupancy rate of logarithmic FOs is the same as that of U-Cube, namely the number of FOs is $2(L_f-1)$ and the spectrum occupancy rate is 
\begin{align}          \label{eq:ROC_3}
    \eta_{\mathrm{logFOs}} = \frac{2(L_f-1) B}{2(L_f-1)\Delta f}= \frac{B}{\Delta f}.
\end{align}

In case of TDFOs, the transmit frequency of the $m$-th element is 
    $f_m=f_b+m \Delta f(t)$. 
Here, the transmit frequency is flexible but the FOs must be controlled accurately in real-time leading to operational difficulties. The number of FOs is at least $\frac{2(L_f-1)\Delta f}{\max_t(\Delta f(t))}$ with the spectrum occupancy rate 
\begin{align}
    \eta_{\mathrm{TDFO}} = \frac{2(L_f-1)\Delta f B}{\max_t(\Delta f(t))2(L_f-1)\Delta f}= \frac{B}{\max_t(\Delta f(t))}.
\end{align}

With random sparse FOs \cite{liu2016random}, the carrier frequency of the $m$-th element is 
    $f_m=f_b + \xi_m \Delta f$, 
where $\xi_m$ is randomly distributed over the interval $[-(L_f-1),L_f-1]$. While random FOs may utilize the same spectrum, its recovery procedure is more complicated.

While not reported in the existing literature, analogous to some of the coarrays mentioned in Section~\ref{subsec:alt_array}, concatenated nested FOs (CNFOs) and generalized nested FOs (GNFOs) are also possible to gain more DoFs and lower frequency occupancy for the same cumulative bandwidth. For CNFOs, the frequency transmitted from the $m$-th element is 
    $f_m=f_b+ \Delta f_m$, 
where
\begin{align}
        \Delta f_m= 
    \begin{cases}
[m+2N(1-N)] \Delta f, ~& -(4N-2) \leq m \leq -(3N-1),          \\
[mN+(N-1)^2] \Delta f                     ,~ & -(3N-2) \leq m \leq -N,   \\
m\Delta f,~& -(N-1)\leq m \leq N-1, \\
[mN-(N-1)^2]\Delta f,~& N\leq m \leq 3N-2,  \\
[m+2N(N-1)]\Delta f      ,~ & 3N-1 \leq m \leq 4N-2,
\end{cases}
\end{align}
and $N=\lfloor\frac{-1+\sqrt{1+2(1+L_f)}}{2} \rfloor$. So, the number of FOs is at least $8N-4$ resulting in the occupancy rate
\begin{align}            \label{eq:ROC_6}
    \eta_{\mathrm{CNFO}} = \frac{(8N-4)B}{2(L_f-1)\Delta f} =\frac{(4N-2)B}{(L_f-1)\Delta f}.
\end{align}
Similarly, for GNFOs, we have the $m$-th element transmit frequency 
    $f_m=f_b+ \Delta f_m$, 
where
\begin{align}
        \Delta f_m= 
    \begin{cases}
-[N\alpha -(m+N)\beta] \Delta f,~ & -(2N-1) \leq m \leq -N,     \\
m\alpha \Delta f, ~& -(N-1) \leq m \leq N-1,          \\
[N\alpha + (m-N)\beta] \Delta f                     ,~ & N \leq m \leq 2N-1.   \\
\end{cases}
\end{align}
Here, $\alpha$ and $\beta$ are co-prime integers and $N=\lfloor \frac{L_f+\beta -1}{\alpha+\beta} \rfloor$. Then, the number of FOs is $4N-2$ and the spectrum occupancy rate is 
\begin{align}                 \label{eq:ROC_8}
    \eta_{\mathrm{GNFO}} = \frac{(4N-2)B}{2(L_f-1)\Delta f}= \frac{(2N-1)B}{(L_f-1)\Delta f}.
\end{align}

Fig.~\ref{FO_comp} compares the occupancy rates of co-prime FOs, logarithmic FOs, CNFOs, and GNFOs. The co-prime FO has the lowest frequency occupancy than other FOs. This excludes the probabilistic random FOs and TDFOs, which also depend on other radar parameters.
\begin{figure}[t!]
\centerline{\includegraphics[width=0.8\columnwidth]{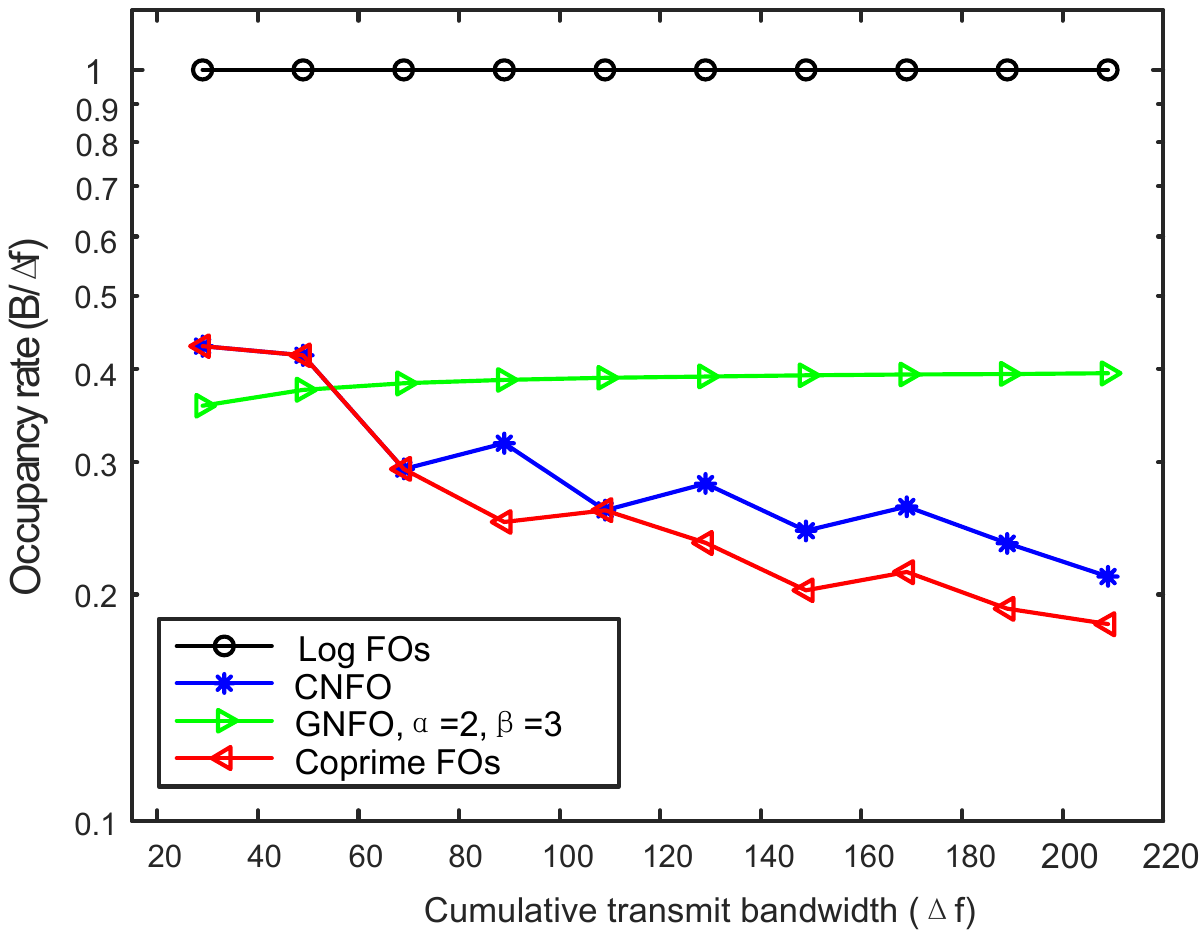}}
\caption{The occupancy rate of different FOs for the same cumulative transmit bandwidth in an L-shaped array. 
For each cumulative transmit bandwidth, the number of elements of L-shaped array is the same as the number of FOs computed using (\ref{eq:ROC_3}), (\ref{eq:ROC_6}) and (\ref{eq:ROC_8}).
}
\label{FO_comp}
\end{figure}

\subsection{Alternatives to co-prime pulsing}
Some other non-uniform PRF sequences have been developed in \cite{GQuan2016,xu2021difference,XWang2018}. The random sparse pulsing eliminates the Doppler ambiguity and enhances electronic counter-countermeasure (ECCM) capability. However, it leads to high sidelobes and restricted DoFs in the Doppler domain \cite{GQuan2016}. Co-pulsing mitigates false Doppler peaks and saves useful dwell time. The nested pulsing suggested in \cite{xu2021difference} may be extended to L-shaped arrays to obtain high Doppler resolution under the difference co-pulse concept. Here, a CPI comprises two sparse uniform pulse trains that have $N_1$ pulses with PRI $T$ and $N_2$ pulses with PRI $(N_1+1)T$, respectively. The CPI of L-shaped array remains the same, i.e., $T_c=(L_p -1)T$. Denote the dwell time by $T_a$ and define the dwell time occupancy rate as $\kappa= \frac{T_a}{T_c}$. From ECCM perspective, lesser the dwell time occupancy rate, better the pulsing scheme. Following a similar analysis as in the previous subsection for co-FOs, the co-pulsing achieves the same Doppler resolution but with lower dwell time than the uniform PRI. 

Note that co-prime pulsing is the most fundamental co-pulsing scheme and other sequences could be obtained from co-prime pulsing. 
For example, if the first and second uniform sparse pulse trains of co-prime pulsing have $N_1=N$ pulses with PRI $T_1=MT$ and $N_2=2M-1$ with PRI $T_2=NT$, respectively, then, \textit{ceteris paribus} nested pulsing is obtained by setting $T_1=T$. Nested pulsing yields super-nested sequence with a rearrangement of the positions of pulses as explained for spatial domain in \cite{CLiu2016}. The concatenated nested pulsing is derived from nested pulsing and its mirror image placed in succession. From co-prime pulsing, CADiS pulsing results when $L>0$. Finally, reducing the number of pulses in the second pulse train of co-prime pulsing yields CACIS pulsing.

\section{Numerical Experiments}          
\label{Section_simulation}
We validated our co-pulsing FDA model and methods through extensive numerical experiments. Unless otherwise noted, we set co-prime integers to $M_s=2$ and $N_s=3$ for co-prime arrays and co-prime FOs. Thus, the number of sensors along either $x$-axis or $z$-axis is $P_s= N_s+2M_s-1=6$. The total number of sensors in an L-shaped co-prime array is $2P_s-1=11$. For co-prime PRI, the co-prime numbers are set to $M_t=2$ and $N_t=3$. Hence, a total of $K=N_t+2M_t-1=6$ pulses are transmitted during the CPI with the fundamental PRI and pulse duration of $T=0.05$ ms and $T_p=0.5$ $\mu$s, respectively. Since the range periodicity is $R_u=\frac{c}{2\Delta f}$, the frequency increment $\Delta f$ needs to satisfy the boundary condition: $R_u \ge R_{\max}$, i.e $\Delta f \leq c/(2R_{\max})=20$ kHz \cite{XLi2018}. So we set the base carrier frequency $f_b=1$ GHz and 
$\Delta f=20$ kHz. Thus, the maximum unambiguous range becomes 7.5km and the maximum unambiguous velocity is $\nu_{\mathrm{max}}=\frac{\lambda_b}{2T}=\frac{c}{2f_bT}=3000$ m/s. The parameters of $Q$ far-field targets are assumed to be in the scope of $\theta_q \in (0,90^{\circ})$, $\varphi_q \in (-70^{\circ},70^{\circ})$, $r_q \in (100,5000)$ m and $\nu_q \in (10,400)$ m/s. The receive signal-to-noise ratio (SNR) was computed as 
\begin{align}
    \mathrm{SNR}=10\log_{10} \left( \frac{\Vert\mathbf{x}\Vert^2+\Vert\mathbf{z}\Vert^2}{2KP_s (2P_s-1) \sigma_n^2 } \right).
\end{align}
where $\sigma_n^2$ is the additive noise variance. Unless otherwise stated, the reflectivities $\sigma_s^2$ of targets are assumed to be equal. Throughout all experiments, we used our CCing algorithm for parameter recovery.

In the following, we present the target recovery during various experiments in the elevation-azimuth, elevation-range, elevation-Doppler, azimuth-range, azimuth-Doppler, and range-Doppler planes simultaneously. Here, a successful detection (blue cross) occurs when the estimated target is within one range cell, one azimuth bin and one Doppler bin of the ground truth (red circle); otherwise, the estimated target is labeled as a false alarm (circle with dark fill). Note that, for the purposes of clear illustration, the markers have been magnified; the exact location of the markers should be taken as their geometric centers. 

\noindent\textbf{DoFs enhancement}: 
We first examine the ability of our proposed C-Cube radar to enhance the DoFs in angle-range-Doppler domains by exploiting the difference coarray, frequency difference equivalence, and PRI difference equivalence. The proposed method is compared with the U-Cube. In this comparison, we assume that all the parameters of $Q$ targets are distinct, namely $\theta_q \neq \theta_p, \varphi_q \neq \varphi_p, r_q \neq r_q, \nu_q \neq \nu_p, 1\leq p \neq q \leq Q$. Without loss of generality, we set $M_s=M_t$ and $N_s=N_t$. From \textbf{C4}-\textbf{C7}, U-Cube and C-Cube can detect a maximum of $\mathcal{O}(\min\{P_s,K\})=6$ and $\mathcal{O}(\min\{2L_s,2L_t\})=14$ targets, respectively. When $Q=3$, Figs.~\ref{Detection_filled_3targets} and  \ref{Detection_coprime_3targets}, respectively, show that both C-Cube and U-Cube perfectly recover all target parameters. However, when the targets are nearly doubled with $Q=7$, U-Cube (Fig.~\ref{Detection_filled_7targets}) is unable to retrieve all the parameters while C-Cube (Fig.~\ref{Detection_coprime_7targets}) successfully estimates the entire set of target parameters, clearly demonstrating the superior performance of the latter FDA. Additionally, Fig. \ref{Detection_coprime_7targets_same} shows our proposed method also works when one or more target parameters are identical (as long as the conditions \textbf{C1}-\textbf{C5} of Theorem~\ref{Theo_rec} are satisfied). 
\begin{figure}[t!]
\centerline{\includegraphics[width=0.85\columnwidth]{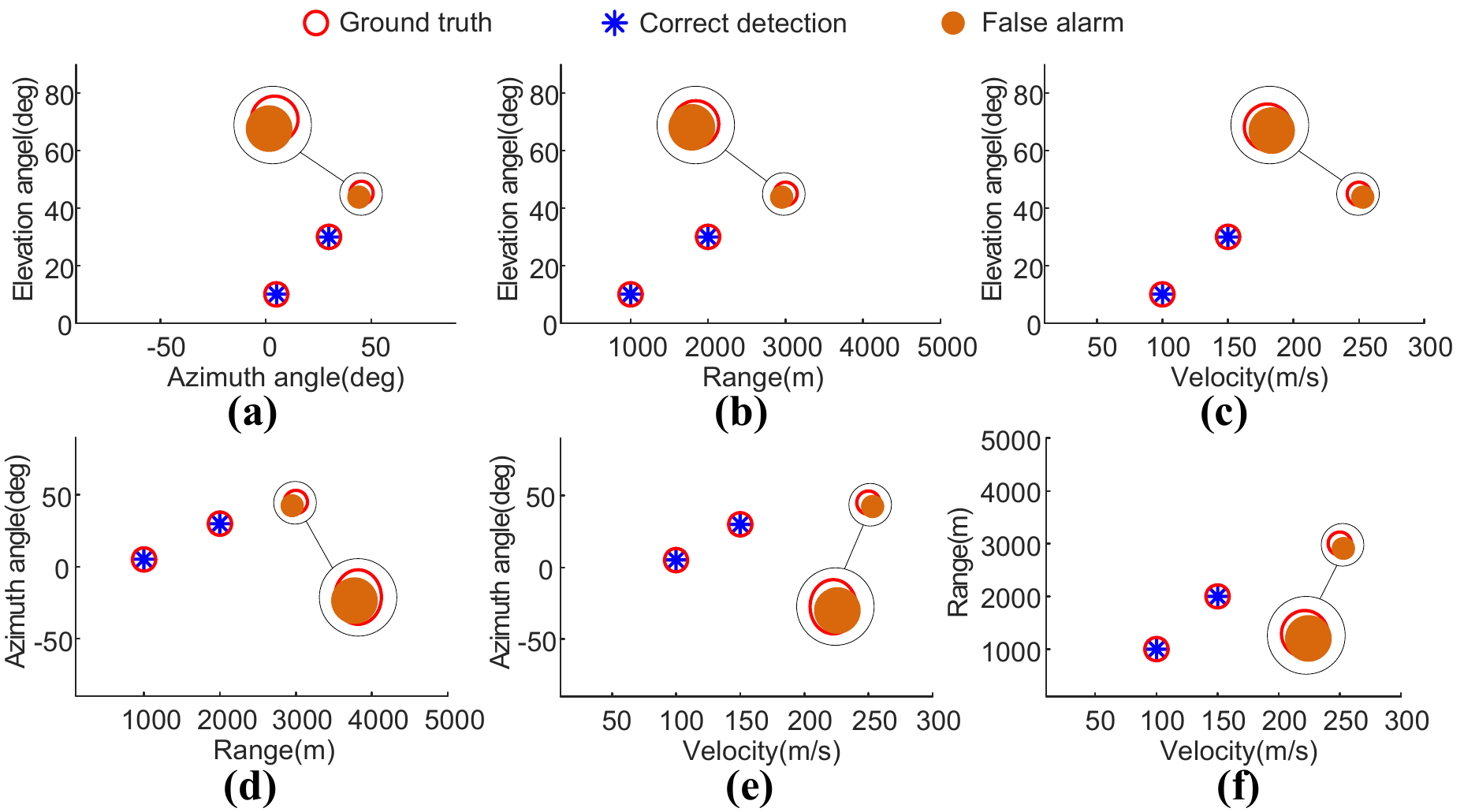}}
\caption{Target detection by U-Cube radar with $Q=3$ and SNR$=10$ dB in (a) elevation-azimuth, (b) elevation-range, (c) elevation-Doppler, (d) azimuth-range, (e) azimuth-Doppler, and (f) range-Doppler planes. The red circles (blue crosses) indicate ground truth (detected targets).}
\label{Detection_filled_3targets}
\end{figure}
\begin{figure}[t!]
\centerline{\includegraphics[width=0.85\columnwidth]{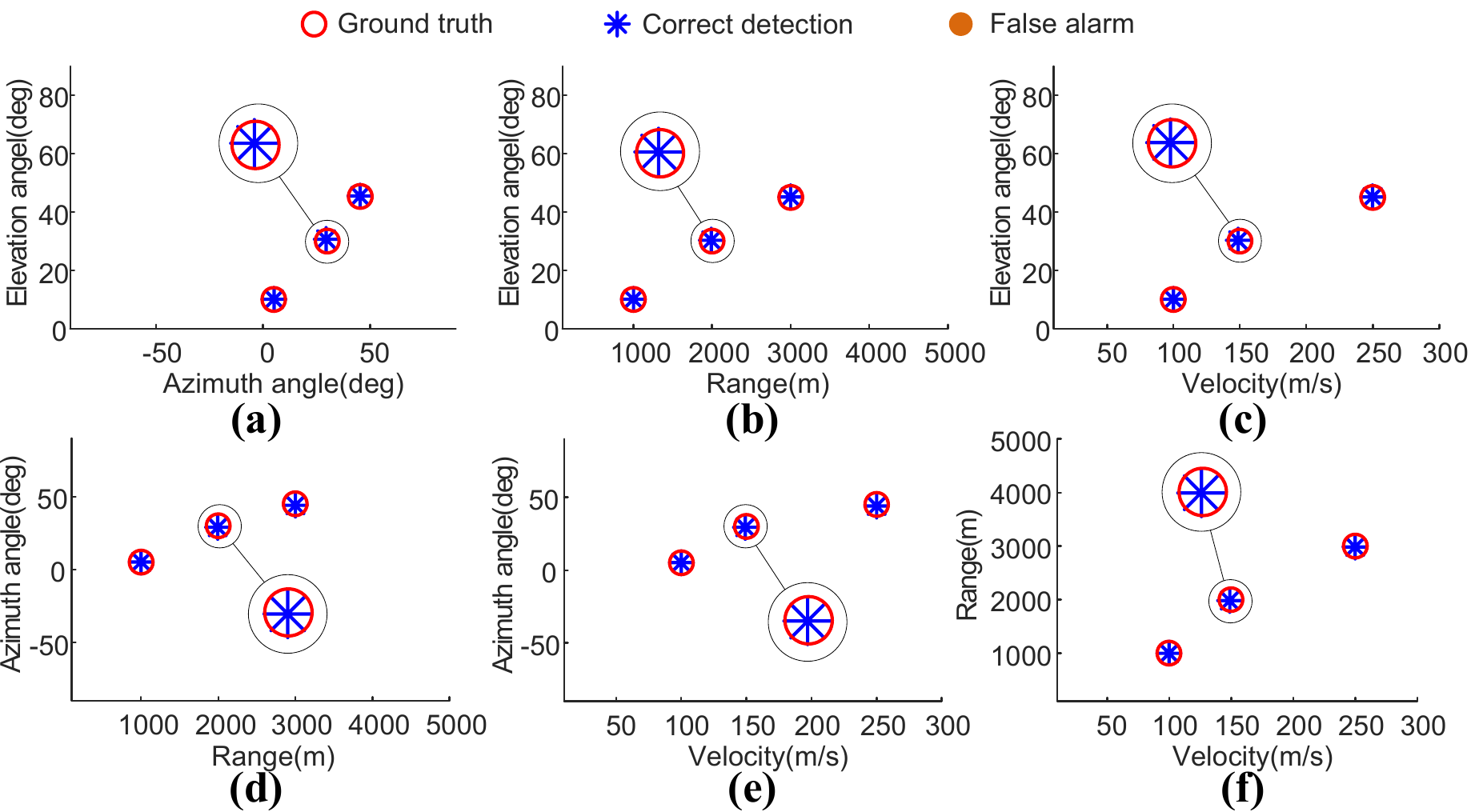}}
\caption{As in Fig.~\ref{Detection_filled_3targets} but for C-Cube radar.}
\label{Detection_coprime_3targets}
\end{figure}
\begin{figure}[t!]
\centerline{\includegraphics[width=0.85\columnwidth]{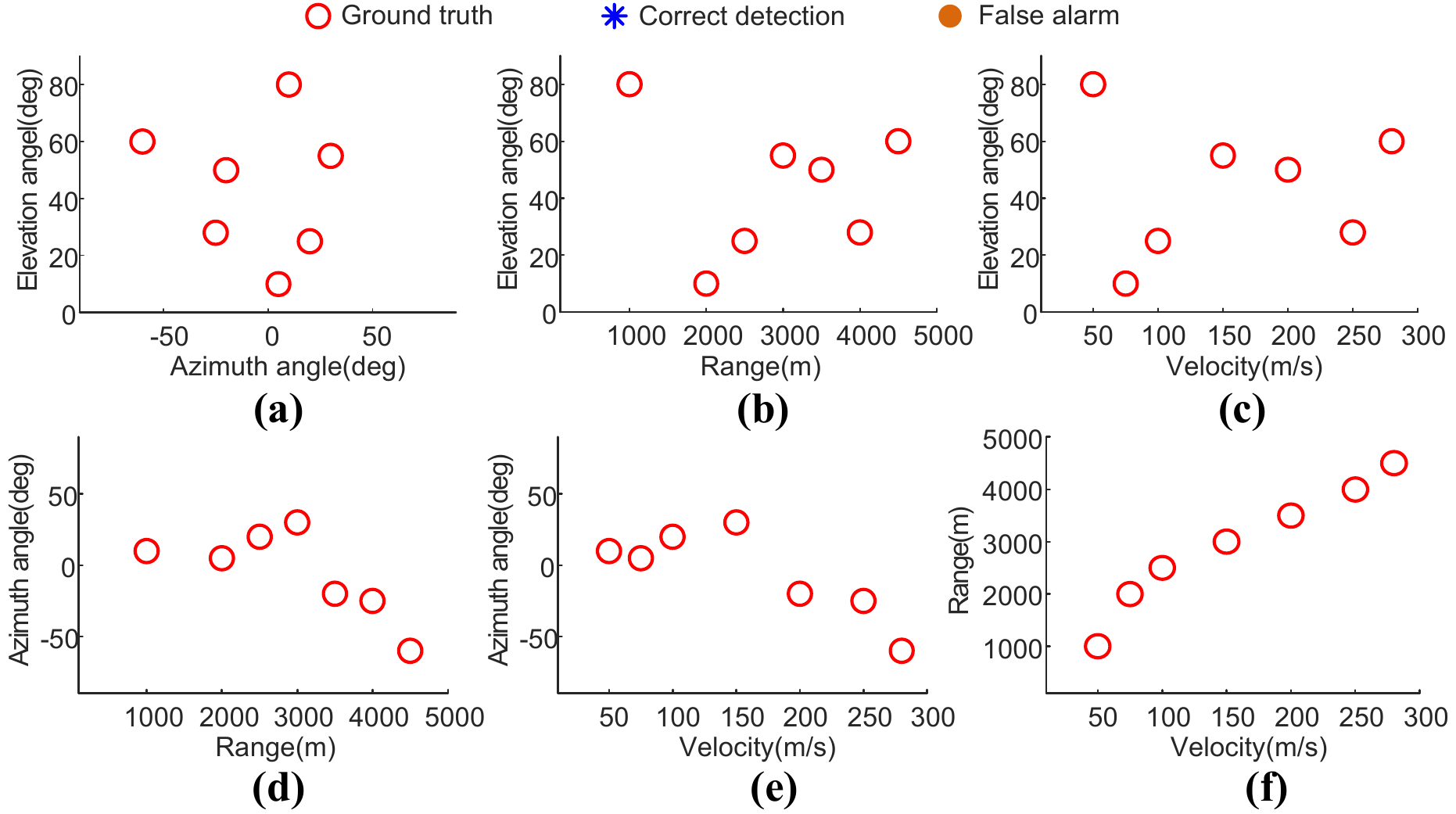}}
\caption{As in Fig.~\ref{Detection_filled_3targets} but for $Q=7$, all of which could not be detected.
\label{Detection_filled_7targets}}
\end{figure}

Define
$\textrm{Hit Rate}= \frac{1}{Q}\vert \{q \in [1,Q]: \vert \theta_l - \hat{\theta}_q \vert \leq \epsilon_{\theta}, \vert \varphi_l - \hat{\phi}_q \vert \leq \epsilon_{\varphi},  
     \vert r_l - \hat{r}_q \vert \leq \epsilon_{r} , \vert \nu_l - \hat{\nu}_q \vert \leq \epsilon_{\nu}, l\in [1,K]\}  \vert$,
where $\epsilon_{\theta},\epsilon_{\varphi},\epsilon_{r}$ and $\epsilon_{\nu}$ are single-bin tolerances in elevation, azimuth, range and velocity direction, respectively. We averaged hit rates over all Monte Carlo realizations at SNR $=10$ dB;  Fig.~\ref{Detection_filled_3targets} and Fig.~\ref{Detection_coprime_3targets} show a single such realization for U-Cube and C-Cube. The hit rate for C-Cube was $15\%$ higher than U-Cube. 

\noindent\textbf{Statistical performance}: 
We analyse the statistical performance of our proposed C-Cube FDA with UUU, UUC, UCU, UCC, CUU, CUC, and CCU. Since the number of uncorrelated targets that L-shaped ULA-based FDA can detect is restricted by the number of physical sensors, for a fair comparison we set $Q=2$ with $\gamma_1=\{10^{\circ}, 5^{\circ}, 1000\textrm{ m}, 100\textrm{ m/s}\}$ and $\gamma_2=\{45^{\circ}, 45^{\circ}, 3000\textrm{ m}, 250\textrm{ m/s}\}$, respectively. The number of samples in each PRI were $L_r=\lfloor T/T_p \rfloor=100$. We benchmark the performance of various arrays using the root mean square error (RMSE) 
$=\frac{1}{J}\sum_{j=1}^J \sqrt{\frac{1}{Q}\sum_{q=1}^Q(\eta'_{q,j}-\eta_q)^2}$, 
where $J$ (set to $200$) is the number of Monte Carlo simulations, $\eta'_{q,j}$ is the estimated parameter (elevation angle, azimuth angle, range or Doppler velocity) of the $q$-th target in the $j$-th Monte Carlo simulation, $\eta_q$ is the true value of the same parameter. 
 
Figs.~\ref{RMSEvsSNR_AllParams}a and b show the RMSE of, respectively, elevation and azimuth angles versus the received SNR that was varied from $-15$ to $15$ dB in the steps of $3$ dB. For comparison, we also plot the corresponding root of CRB (RCRB) values. The DoA estimation performance of all methods improves with SNR. The C-Cube structure benefits from the larger aperture in the coarray domain than, say, U-Cube. Similarly, other structures such as UUC, UCU, UCC, CUU, 
and CCU achieve improved estimation of either DoA (Figs.~\ref{RMSEvsSNR_AllParams}a-b), range (Fig.~\ref{RMSEvsSNR_AllParams}c) or Doppler (Fig.~\ref{RMSEvsSNR_AllParams}d) over U-Cube because of the co-prime structure in their sensor positions, FOs, and/or PRIs. Clearly, the C-Cube with a simultaneous co-prime structure in its sensor positions, FOs and PRIs outperforms all other arrays. 
While both of our proposed C-Cube and CUC methods are more robust to noise than other configurations, the C-Cube also offers savings in the spectrum usage when compared with CUC. The gap between the RCRB and RMSEs in Fig.~\ref{RMSEvsSNR_AllParams} arises from the fact that the RCRB is based on the whole space-time-frequency coarray of Fig.~\ref{Space_time_frequency_coarray}. But the CCing algorithm adopts a more practical approach of the contiguous space-time-frequency coarray. In order to reduce the gap, approaches similar to \cite{ZMao2022} may be employed to fill the holes in the coarray of Fig.~\ref{Space_time_frequency_coarray}.

\begin{figure}[t!]
\centerline{\includegraphics[width=0.85\columnwidth]{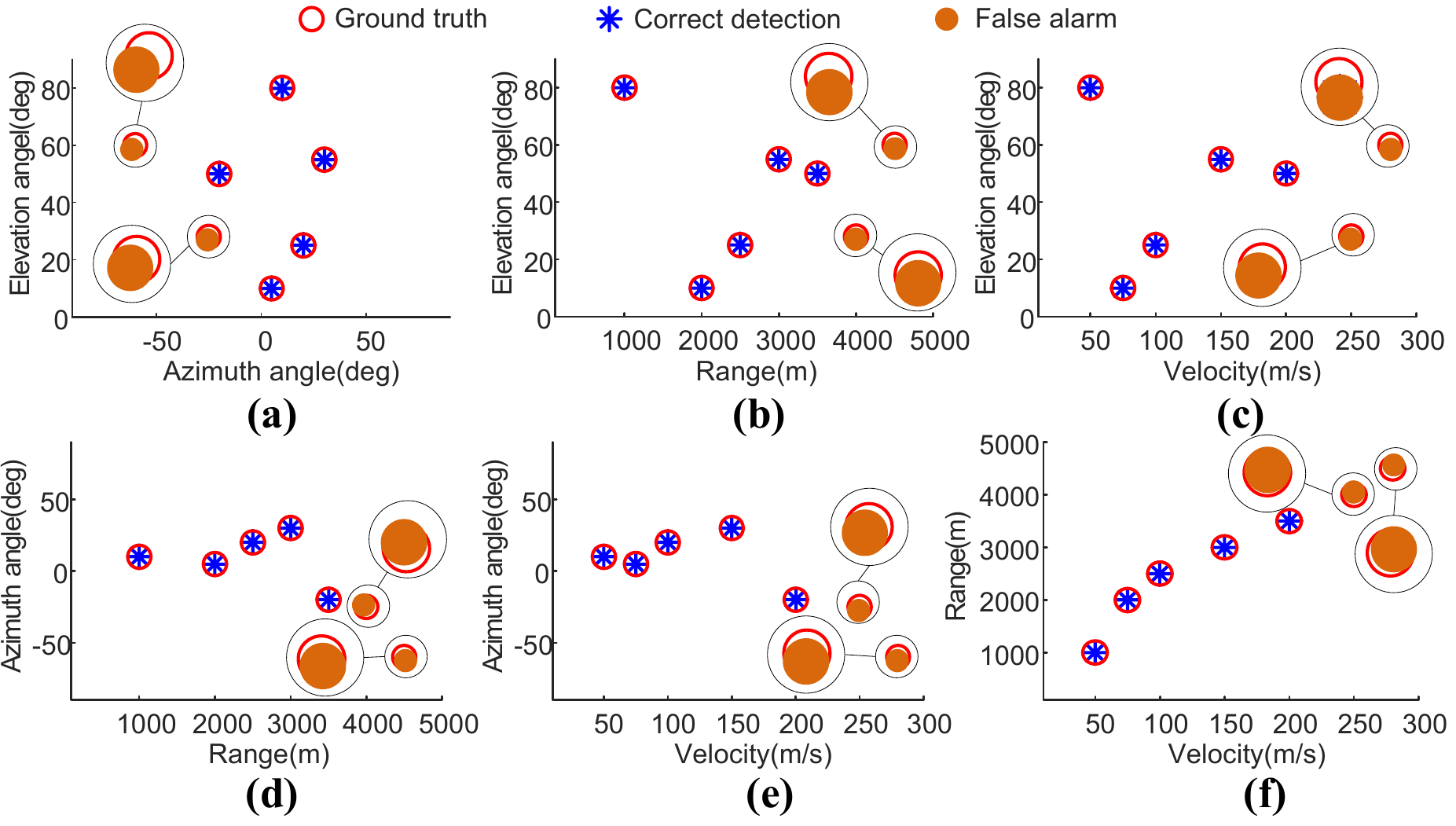}}
\caption{As in Fig.~\ref{Detection_filled_7targets} but for C-Cube radar.
\label{Detection_coprime_7targets}}
\end{figure}
\begin{figure}[t!]
\centerline{\includegraphics[width=0.85\columnwidth]{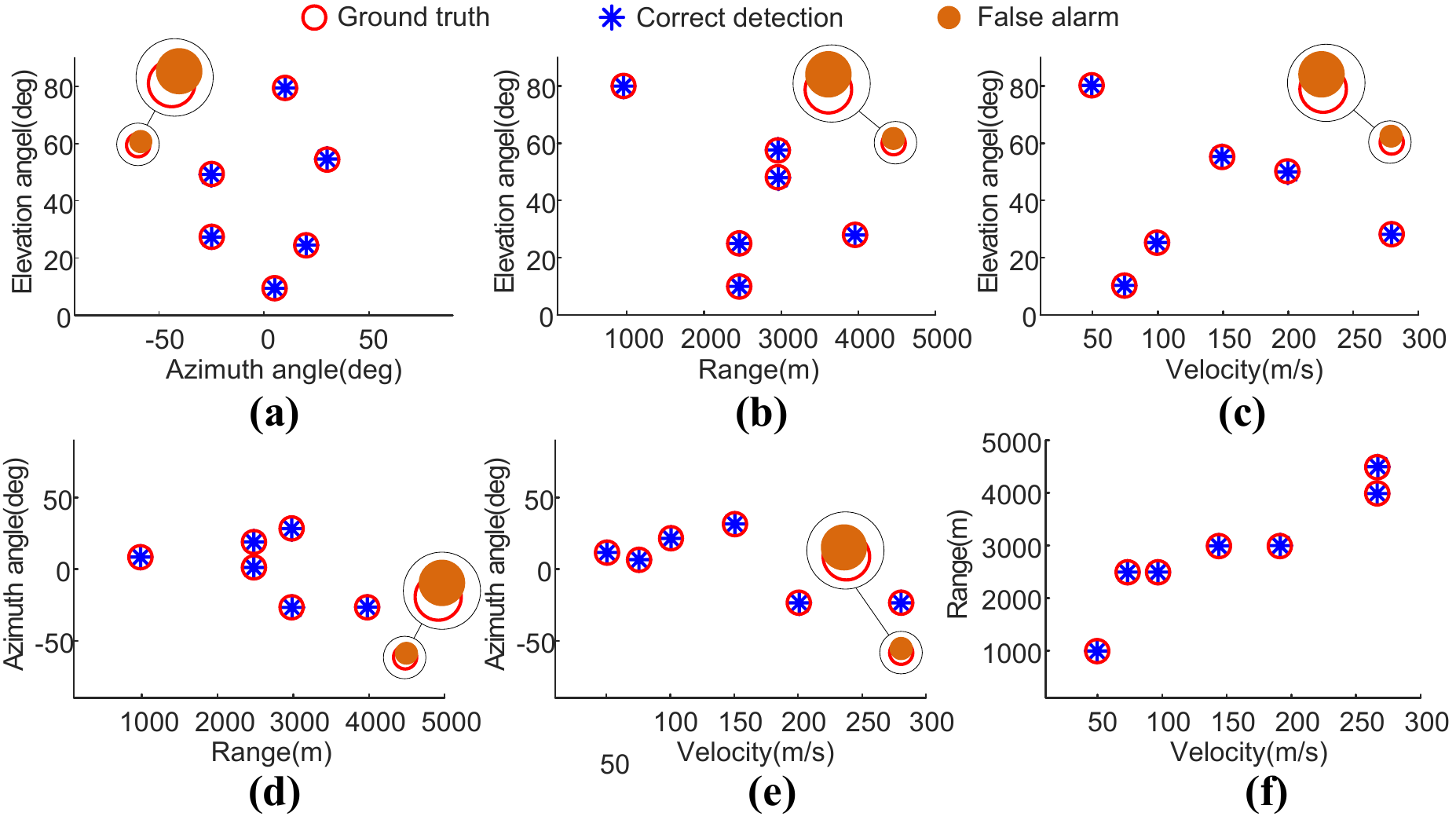}}
\caption{As in Fig.~\ref{Detection_coprime_7targets} but two targets have identical ranges and Doppler velocities.
\label{Detection_coprime_7targets_same}}
\end{figure}

\noindent\textbf{Higher co-prime integers and unequal reflectivities}: 
Keeping other parameters the same as in Fig.~\ref{RMSEvsSNR_AllParams}, assume that there are $Q=5$ targets with $\gamma_1=\{10^{\circ}, 5^{\circ}, 2000\textrm{ m}, 75\textrm{ m/s}\}$, $\gamma_2=\{25^{\circ}, 20^{\circ}, 2500\textrm{ m}, 100\textrm{ m/s}\}$, $\gamma_3=\{55^{\circ}, 30^{\circ}, 3000\textrm{ m}, 150\textrm{ m/s}\}$, $\gamma_4=\{50^{\circ}, -20^{\circ}, 3500\textrm{ m}, 200\textrm{ m/s}\}$ and $\gamma_5=\{60^{\circ}, -60^{\circ}, 4500\textrm{ m}, 280\textrm{ m/s}\}$, respectively. We choose the co-prime integers as $M_s=3$ and $N_s=5$. We analysed the statistical performance of our proposed C-Cube FDA for unequal target reflectivities $\{\sigma_q^2 \}_{q=1}^{Q}$. The mean value of target power sequence was fixed to $\sigma_s^2$, namely $\frac{1}{Q}\sum_{q=1}^{Q}\sigma_q^2=\sigma_s^2$. The difference among target powers is reflected by the standard deviation (s.d.) of target power sequence, i.e.,
\begin{align}
\textrm{s.d.} = \sqrt{\sum_{q=1}^Q (\sigma_q^2-\sigma_s^2)^2/Q}.
\end{align}
It follows from Fig. \ref{RMSE_Power_Differ} that the parameter estimation performance of different s.d. values improves with increase in SNR. In addition, a larger s.d. results in performance degradation, especially in low SNR regime because some weak targets get buried in noise or occluded by strong targets.

\noindent\textbf{Performance/complexity comparison}: 
For U-Cube, we compared the performance  of our proposed CCing algorithm with the joint ESPRIT \cite{SSIoushua2017} and sparse-reconstruction-based PARAFAC (SR-PARAFAC) tensor decomposition 
\cite{LXu2018} algorithms. 
The target parameters and FDA radar settings were as in the first experiment 
but with $Q=3$ and fixed SNR of $20$ dB. The number of sensors along each axis is the same as the number of pulses; we varied them simultaneously 
in increments of $2$. The CCing had superior parameter estimation than the other two algorithms (Fig.~\ref{RMSEvsSensorNum_AllParams}). 
 Further, the CCing algorithm is an order faster than joint ESPRIT and SR-PARAFAC algorithms (Fig.~\ref{runtime}); note that, in addition, the latter are not applicable to C-Cube. The joint ESPRIT and SR-PARAFAC estimate all parameters simultaneously and are, therefore, computationally more complex. The CCing algorithm reduces the run-time by splitting the process into two separate 2-D DoA and range-Doppler estimation steps followed by auto-pairing. 

\section{Summary}           
\label{Section_conclusion}
We investigated several new L-shaped co-pulsing FDA radar systems with the goal of reducing resources such as sensors, spectrum, and dwell time. Drawing on the advances in coarrays, our proposed C-Cube configuration adopts co-prime structure across all three design parameters -- array geometry, FOs, and PRIs -- and hence, achieves a high number of DoFs compared to L-shaped ULA, L-shaped co-prime array, and L-shaped FDAs. These complex systems invariably present challenges in the joint estimation of target parameters. To this end, our CCing retrieval algorithm allows for automatic pairing of 2D-DoA with range and Doppler velocity while also exhibiting faster execution times. Both analytical and numerical results validate the effectiveness and superiority of our C-CUBE system over various other L-shaped co-pulsing FDAs.
\begin{figure*}
\centerline{\includegraphics[width=0.95\textwidth]{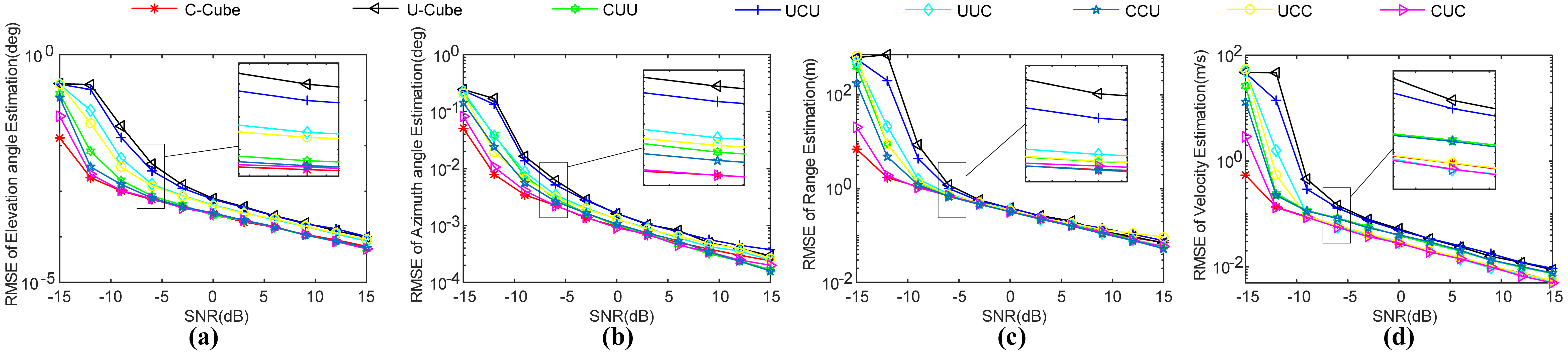}}
\caption{RMSE of (a) elevation, (b) azimuth, (c) range, and (d) Doppler velocity versus SNR for various L-shaped FDAs with $M_s=2$, $N_s=3$, and $Q=2$.
}
\label{RMSEvsSNR_AllParams}
\end{figure*}
\begin{figure*}
\centerline{\includegraphics[width=0.95\textwidth]{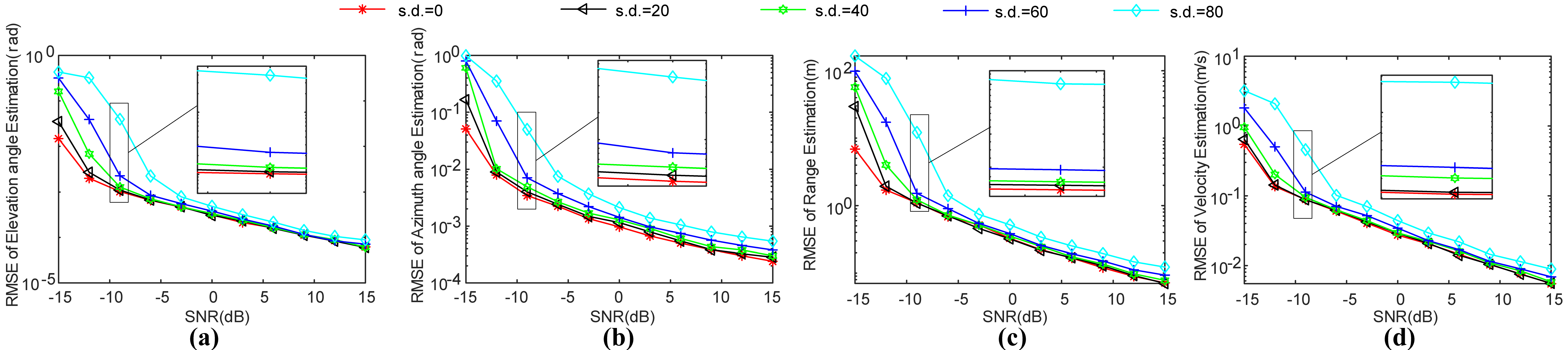}}
\caption{RMSE of (a) elevation, (b) azimuth, (c) range, and (d) Doppler velocity versus SNR for various s.d. values in case of C-Cube configuration.
}
\label{RMSE_Power_Differ}
\end{figure*}

\begin{figure*}
\centerline{\includegraphics[width=0.95\textwidth]{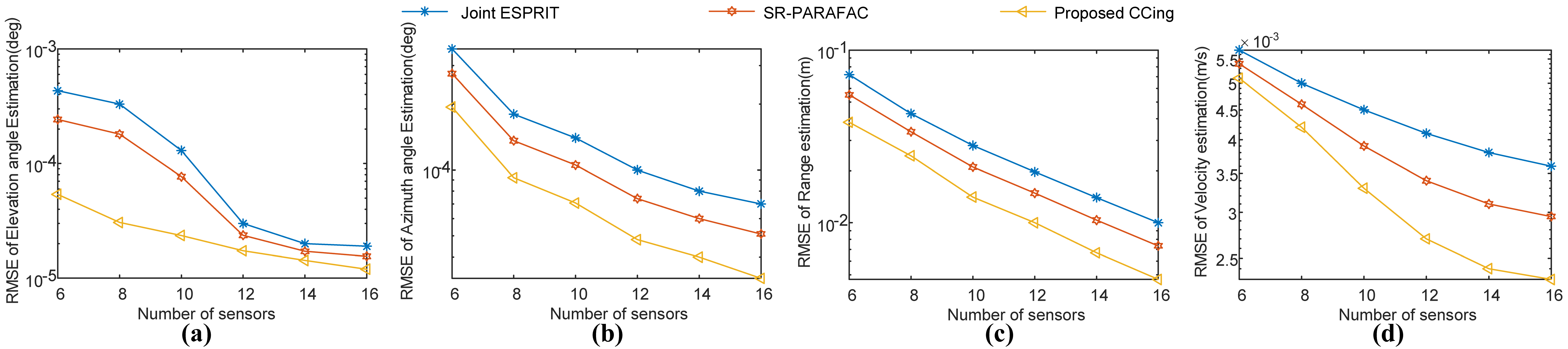}}
\caption{RMSE of (a) elevation, (b) azimuth, (c) range, and (d) Doppler velocity for various retrieval algorithms at SNR $=20$ dB in U-Cube radar.
}
\label{RMSEvsSensorNum_AllParams}
\end{figure*}
\begin{figure}[t]
\centerline{\includegraphics[width=0.8\columnwidth]{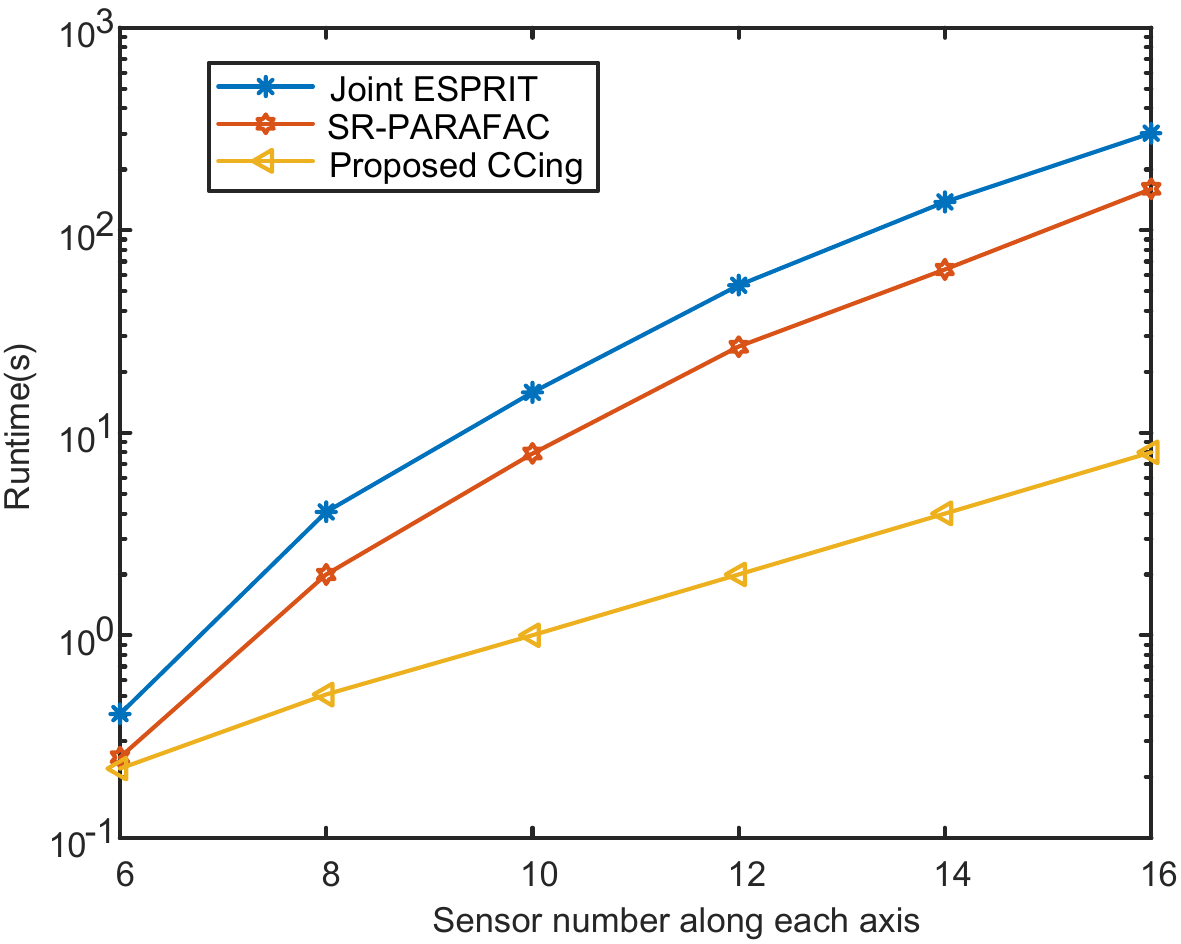}}
\caption{The run-time of our proposed CCing compared with the joint ESPRIT and SR-PARAFAC algorithms with respect to the number of physical array elements for U-Cube radar at fixed SNR of $20$ dB. 
}
\label{runtime}
\end{figure}

While more complex 2-D coarrays exist as in circular and elliptical geometries \cite{kozick1992coarray}, the L-shaped arrays are simpler to implement. There is a rich heritage of sparse arrays where parameter recovery algorithms does not require sparse reconstruction, such as zero redundancy arrays (ZRAs) \cite{arsac1955nouveau}, minimum redundancy arrays (MRAs), and low redundancy arrays (LRAs) \cite{camps2001synthesis}, including array patterns that approach Leech’s bounds \cite{leech1956representation}. In particular, the MRA \cite{moffet1968minimum} removes array elements from ULA such that the resulting sparse array has all possible inter-element spacings of the full array. The ZRA, MRA, and LRA minimize aliasing from the grating lobes while maintaining a reasonably constant integrated sidelobe level. However, they are not feasible for sparsifying large arrays. In this context, our approach may be viewed as obtaining low complexity for arbitrarily large sparse arrays. As a future research direction, 
this work may be extended to develop L-shaped processing for higher DoF structures such as mirrored arrays \cite{chen2012one},  
which comprise an ordinary linear array and a reflector.

\appendices
\section{Proof of Proposition~\ref{Pro_TL}}
\label{app:proof_prop2}
\subsection{Preliminaries to the Proof}
\label{app:defn}
In order to prove the Proposition, recall the definition of Kruskal rank.
\begin{definition}\cite{JBKruskal1977}
The Kruskal rank $\mathrm{Krank}(\mathbf{A} )$ of the matrix $\mathbf{A}$ is defined as the largest integer $\kappa$ such that every $\kappa$ columns of $\mathbf{A}$ are linearly independent.
\end{definition}
\subsection{Proof of the Proposition}
\label{app:proof}
We transform the L-shaped co-prime array with $P_s=N_s+2M_s-1$ physical sensors along each axis into a virtual filled L-shaped array in the difference coarray domain using (\ref{eq:SM20}), (\ref{eq:SM21}), and (\ref{eq:SM22}), where each axis has $2M_sN_s+2M_s-1$ consecutive lags with the virtual sensor positions from $-(M_sN_s+M_s-1)d$ to $(M_sN_s+M_s-1)d$. Both $\widetilde{\mathbf{A}}_x$ and $\widetilde{\mathbf{A}}_z$ are Vandermonde matrices with the size $(2L_s+1)\times Q$. Using (\ref{eq:SM3}) and \textbf{C1}, these matrices have distinct columns. Hence, if \textbf{C4} is satisfied, both $\widetilde{\mathbf{A}}_x$ and $\widetilde{\mathbf{A}}_z$ have full column rank equal to $Q$. Ergo, the stacked version $[\widetilde{\mathbf{A}}_x^T ~ \widetilde{\mathbf{A}}_z^T]^T$ also has the full column rank $Q$, namely
\begin{align}            \label{eq:Pro1}
    \mathrm{rank}\left( \left[\begin{matrix}
 \widetilde{\mathbf{A}}_x  \\
 \widetilde{\mathbf{A}}_z
\end{matrix} \right] \right) = Q.
\end{align}

Similarly, if \textbf{C2} and \textbf{C3} hold, $\widetilde{\mathbf{C}}$ and $\widetilde{\mathbf{B}}$ are Vandermonde matrices with the size $(2L_s+1)\times Q$ and $(2L_t+1)\times Q$, respectively. Then, when \textbf{C4} and \textbf{C5} are satisfied, both $\widetilde{\mathbf{C}}$ and $\widetilde{\mathbf{B}}$ have full column rank, namely 
    $\mathrm{rank}(\widetilde{\mathbf{C}}) = \mathrm{rank}(\widetilde{\mathbf{B}})=Q$. 
Using this result and following the definition of Kruskal, we have
\begin{align}            \label{eq:Pro4}
\mathrm{Krank}(\widetilde{\mathbf{C}}) = \mathrm{Krank}(\widetilde{\mathbf{B}})=Q.
\end{align}
From \cite{NDSidiropoulos2000}, it follows that
\begin{align}            \label{eq:Pro5}
\mathrm{Krank}(\widetilde{\mathbf{C}} \odot \widetilde{\mathbf{B}}) \geq \min (Q,\mathrm{Krank}(\widetilde{\mathbf{C}})+\mathrm{Krank}(\widetilde{\mathbf{B}})-1 ).
\end{align}
Substituting (\ref{eq:Pro4}) into (\ref{eq:Pro5}) yields $\mathrm{Krank}(\widetilde{\mathbf{C}} \odot \widetilde{\mathbf{B}}) \geq Q$. The number of columns of $\widetilde{\mathbf{C}} \odot \widetilde{\mathbf{B}}$ is $Q$. So
\begin{align}             \label{eq:Pro7}
\mathrm{Krank}(\widetilde{\mathbf{C}} \odot \widetilde{\mathbf{B}})=\mathrm{rank}(\widetilde{\mathbf{C}} \odot \widetilde{\mathbf{B}})=Q.
\end{align}
In addition, $\mathbf{R}_{\rho}$ in (\ref{eq:RA2}) is an $Q\times Q$ diagonal matrix which is invertible, namely
\begin{align}              \label{eq:Pro8}
\mathrm{rank}(\mathbf{R}_{\rho}) =Q.
\end{align}
Combining (\ref{eq:Pro1}), (\ref{eq:Pro7}) and (\ref{eq:Pro8}), it follows that
\begin{align}             \label{eq:Pro9}
    \mathrm{rank}((\widetilde{\mathbf{C}}\odot \widetilde{\mathbf{B}})^*\mathbf{R}_{\bm{\rho}}) =Q,~~ \mathrm{rank}\left( \left[\begin{matrix}
 \widetilde{\mathbf{A}}_x  \\
 \widetilde{\mathbf{A}}_z
\end{matrix} \right]\mathbf{R}_{\bm{\rho}}  \right) =Q.
\end{align}

After SVD decomposition in (\ref{eq:RA3}), we obtain
\begin{align}
\mathbf{R}_{XZ}^H \mathbf{U}_2 = (\widetilde{\mathbf{C}}\odot \widetilde{\mathbf{B}})^*\mathbf{R}_{\bm{\rho}} \left[\begin{matrix}
 \widetilde{\mathbf{A}}_x  \\
 \widetilde{\mathbf{A}}_z
\end{matrix} \right]^H \mathbf{U}_2 =0.
\end{align}
Since $(\widetilde{\mathbf{C}}\odot \widetilde{\mathbf{B}})^*\mathbf{R}_{\bm{\rho}}$ is full column rank, it follows that
\begin{align}
\left[\begin{matrix}
 \widetilde{\mathbf{A}}_x  \\
 \widetilde{\mathbf{A}}_z
\end{matrix} \right]^H \mathbf{U}_2 =0.
\end{align}
Because $[\widetilde{\mathbf{A}}_x^H~\widetilde{\mathbf{A}}_z^H]$ has the same size with $\mathbf{U}_1^H$, $\mathbf{U}_2$ is the null space of $[\widetilde{\mathbf{A}}_x^H~\widetilde{\mathbf{A}}_z^H]$ and $\mathbf{U}_1^H$. This implies
that
\begin{align}
\mathrm{range}\left(\begin{matrix}
 \widetilde{\mathbf{A}}_x  \\
 \widetilde{\mathbf{A}}_z
\end{matrix}\right)=\mathrm{range}(\mathbf{U}_1).
\end{align}
Ergo, there exists an invertible $Q\times Q$ matrix $\mathbf{T}_L$ such that (\ref{eq:RA5}) holds.

Similarly, following the conjugate transpose version of (\ref{eq:RA3})
\begin{eqnarray}
\mathbf{R}_{XZ}^H &=& [\mathbf{V}_1~  \mathbf{V}_2]
\left[\begin{matrix}
 \mathbf{\Lambda} & \mathbf{0}  \\
 \mathbf{0}  &  \mathbf{0}
\end{matrix} \right]^H  [\mathbf{U}_1~ \mathbf{U}_2]^H   = \mathbf{V}_1 \mathbf{\Lambda} \mathbf{U}_1^H     \notag \\
&=& (\widetilde{\mathbf{C}}\odot \widetilde{\mathbf{B}})^*
 \mathbf{R}_{\bm{\rho}}[\widetilde{\mathbf{A}}_x^H ~ \widetilde{\mathbf{A}}_z^H]^H,
\end{eqnarray}
we have
\begin{align}
    \mathbf{R}_{XZ} \mathbf{V}_2=\left[\begin{matrix}
 \widetilde{\mathbf{A}}_x  \\
 \widetilde{\mathbf{A}}_z
\end{matrix} \right] \mathbf{R}_{\bm{\rho}} (\widetilde{\mathbf{C}}\odot \widetilde{\mathbf{B}})^T \mathbf{V}_2=0.
\end{align}
Invoking (\ref{eq:Pro9}) gives 
    $(\widetilde{\mathbf{C}}\odot \widetilde{\mathbf{B}})^T \mathbf{V}_2=0$. 
This implies that $\mathbf{V}_2$ is the null space of $(\widetilde{\mathbf{C}}\odot \widetilde{\mathbf{B}})^T$ and $\mathbf{V}_1^H$ so that
\begin{align}
    \mathrm{range}((\widetilde{\mathbf{C}}\odot \widetilde{\mathbf{B}})^*)= \mathrm{range}(\mathbf{V}_1).
\end{align}
Thus, there exists an invertible $Q\times Q$ matrix $\mathbf{T}_R$ such that (\ref{eq:RA5_1}) holds.

\section{Proof of Theorem~\ref{Theo_CRB}}
\label{app:proof_theo_CRB}
From (\ref{eq:PA_CRB4}), we have
\begin{align}
\frac{\mathbf{J}(\bm{\gamma})}{L_r} &= [\mathbf{V}_{\bm{\theta}}~ \mathbf{V}_{\bm{\varphi}}~ \mathbf{V}_{\bm{r}} ~ \mathbf{V}_{\bm{\nu}} ]^H \mathbf{W} [\mathbf{V}_{\bm{\theta}}~ \mathbf{V}_{\bm{\varphi}}~ \mathbf{V}_{\bm{r}} ~ \mathbf{V}_{\bm{\nu}} ]        \notag \\
&= [\mathbf{D}_L ~ \mathbf{D}_R]^H \mathbf{W} [\mathbf{D}_L ~ \mathbf{D}_R]                   \notag \\
&= \left[  \begin{matrix}
            \mathbf{D}_L^H \mathbf{W} \mathbf{D}_L       &    \mathbf{D}_L^H \mathbf{W} \mathbf{D}_R   \\
            \mathbf{D}_R^H \mathbf{W} \mathbf{D}_L       &    \mathbf{D}_R^H \mathbf{W} \mathbf{D}_R
             \end{matrix} 
    \right],
\end{align}
Using the matrix inversion lemma \cite{HornRA1985}, we have
\begin{align}
\mathrm{CRB}(\bm{\theta},\bm{\varphi}) &= \frac{1}{L_r} [\mathbf{D}_L^H \mathbf{W} \mathbf{D}_L                  \notag \\
&~~ -\mathbf{D}_L^H \mathbf{W} \mathbf{D}_R (\mathbf{D}_R^H \mathbf{W} \mathbf{D}_R)^{-1} \mathbf{D}_R^H \mathbf{W} \mathbf{D}_L]   \notag \\
&= \frac{1}{L_r}\{[\mathbf{V}_{\bm{\theta}}~ \mathbf{V}_{\bm{\varphi}}]^H\mathbf{W}^{1/2} \mathbf{\Pi}^{\bot}_{\mathbf{W}^{1/2}\mathbf{D}_R} \mathbf{W}^{1/2} [\mathbf{V}_{\bm{\theta}}~ \mathbf{V}_{\bm{\varphi}}]\}^{-1}               \notag \\
&= \frac{1}{L_r}  \left[   
 \begin{matrix}
            \mathbf{M}_{11}       &    \mathbf{M}_{12}   \\
            \mathbf{M}_{21}      &    \mathbf{M}_{22}
             \end{matrix} 
\right]^{-1},
\end{align}
where 
\begin{subequations}
\begin{align}
\mathbf{M}_{11} &= \mathbf{V}_{\bm{\theta}}^H \mathbf{W}^{1/2} \mathbf{\Pi}^{\bot}_{\mathbf{W}^{1/2}\mathbf{D}_R} \mathbf{W}^{1/2} \mathbf{V}_{\bm{\theta}},  \\
\mathbf{M}_{12} &= \mathbf{V}_{\bm{\theta}}^H \mathbf{W}^{1/2} \mathbf{\Pi}^{\bot}_{\mathbf{W}^{1/2}\mathbf{D}_R} \mathbf{W}^{1/2} \mathbf{V}_{\bm{\varphi}},  \\
\mathbf{M}_{21} &= \mathbf{V}_{\bm{\varphi}}^H \mathbf{W}^{1/2} \mathbf{\Pi}^{\bot}_{\mathbf{W}^{1/2}\mathbf{D}_R} \mathbf{W}^{1/2} \mathbf{V}_{\bm{\theta}},   
\end{align}
and
\begin{align}
\mathbf{M}_{22} &= \mathbf{V}_{\bm{\varphi}}^H \mathbf{W}^{1/2} \mathbf{\Pi}^{\bot}_{\mathbf{W}^{1/2}\mathbf{D}_R} \mathbf{W}^{1/2} \mathbf{V}_{\bm{\varphi}}.
\end{align}
\end{subequations}
The CRBs of $\bm{\theta}$ and $\bm{\varphi}$ are, respectively,
\begin{align}
\mathrm{CRB}(\bm{\theta}) &= \frac{1}{L_r} (\mathbf{M}_{11}-\mathbf{M}_{12}\mathbf{M}_{22}^{-1}\mathbf{M}_{21})^{-1}                 \notag \\
&=\frac{1}{L_r} \left\{ \mathbf{V}_{\bm{\theta}}^H \mathbf{W}^{1/2} \mathbf{\Pi}_{\mathbf{W}^{1/2}\mathbf{D}_R}^{\bot}  \left(\mathbf{\Pi}_{\mathbf{\Pi}_{\mathbf{W}^{1/2}\mathbf{D}_R}^{\bot}  \mathbf{W}^{1/2}\mathbf{V}_{\bm{\varphi}}}^{\bot}\right) \right. \notag  \\
&~~ \times \left. \mathbf{\Pi}_{\mathbf{W}^{1/2}\mathbf{D}_R}^{\bot}  \mathbf{W}^{1/2}\mathbf{V}_{\bm{\theta}} \right\}^{-1},
\end{align}
and
\begin{align}
\mathrm{CRB}(\bm{\varphi}) &=   \frac{1}{L_r} (\mathbf{M}_{22}-\mathbf{M}_{21}\mathbf{M}_{11}^{-1}\mathbf{M}_{12})^{-1}              \notag \\
&= \frac{1}{L_r} \left\{ \mathbf{V}_{\bm{\varphi}}^H \mathbf{W}^{1/2} \mathbf{\Pi}_{\mathbf{W}^{1/2}\mathbf{D}_R}^{\bot}  \left(\mathbf{\Pi}_{\mathbf{\Pi}_{\mathbf{W}^{1/2}\mathbf{D}_R}^{\bot}  \mathbf{W}^{1/2}\mathbf{V}_{\bm{\theta}}}^{\bot}\right) \right. \notag  \\
&~~ \times \left. \mathbf{\Pi}_{\mathbf{W}^{1/2}\mathbf{D}_R}^{\bot}  \mathbf{W}^{1/2}\mathbf{V}_{\bm{\varphi}} \right\}^{-1}.
\end{align}

Similarly,
\begin{align}
\mathrm{CRB}(\bm{r},\bm{\nu}) &= \frac{1}{L_r} [\mathbf{D}_R^H \mathbf{W} \mathbf{D}_R                  \notag \\
&~~ -\mathbf{D}_R^H \mathbf{W} \mathbf{D}_L (\mathbf{D}_L^H \mathbf{W} \mathbf{D}_L)^{-1} \mathbf{D}_L^H \mathbf{W} \mathbf{D}_R]   \notag \\
&=  \frac{1}{L_r}\{[\mathbf{V}_{\bm{r}}~ \mathbf{V}_{\bm{\nu}}]^H\mathbf{W}^{1/2} \mathbf{\Pi}^{\bot}_{\mathbf{W}^{1/2}\mathbf{D}_L} \mathbf{W}^{1/2} [\mathbf{V}_{\bm{r}}~ \mathbf{V}_{\bm{\nu}}]\}^{-1}               \notag \\
&= \frac{1}{L_r}  \left[   
 \begin{matrix}
            \mathbf{N}_{11}       &    \mathbf{N}_{12}   \\
            \mathbf{N}_{21}      &    \mathbf{N}_{22}
             \end{matrix} 
\right]^{-1},
\end{align}
where 
\begin{subequations}
\begin{align}
\mathbf{N}_{11}& =\mathbf{V}_{\bm{r}}^H \mathbf{W}^{1/2} \mathbf{\Pi}^{\bot}_{\mathbf{W}^{1/2}\mathbf{D}_L} \mathbf{W}^{1/2} \mathbf{V}_{\bm{r}},   \\
\mathbf{N}_{12}& =\mathbf{V}_{\bm{r}}^H \mathbf{W}^{1/2} \mathbf{\Pi}^{\bot}_{\mathbf{W}^{1/2}\mathbf{D}_L} \mathbf{W}^{1/2} \mathbf{V}_{\bm{\nu}},  \\
\mathbf{N}_{21}& =\mathbf{V}_{\bm{\nu}}^H \mathbf{W}^{1/2} \mathbf{\Pi}^{\bot}_{\mathbf{W}^{1/2}\mathbf{D}_L} \mathbf{W}^{1/2} \mathbf{V}_{\bm{r}}, 
\end{align}
and
\begin{align}
\mathbf{N}_{22}& =\mathbf{V}_{\bm{\nu}}^H \mathbf{W}^{1/2} \mathbf{\Pi}^{\bot}_{\mathbf{W}^{1/2}\mathbf{D}_L} \mathbf{W}^{1/2} \mathbf{V}_{\bm{\nu}}.
\end{align}
\end{subequations}
Then, the CRBs of $\bm{r}$ and $\bm{\nu}$ are, respectively,
\begin{align}
\mathrm{CRB}(\bm{r}) &= \frac{1}{L_r} (\mathbf{N}_{11}-\mathbf{N}_{12}\mathbf{N}_{22}^{-1}\mathbf{N}_{21})^{-1}                 \notag \\
&=\frac{1}{L_r} \left\{ \mathbf{V}_{\bm{r}}^H \mathbf{W}^{1/2} \mathbf{\Pi}_{\mathbf{W}^{1/2}\mathbf{D}_L}^{\bot}  \left(\mathbf{\Pi}_{\mathbf{\Pi}_{\mathbf{W}^{1/2}\mathbf{D}_L}^{\bot}  \mathbf{W}^{1/2}\mathbf{V}_{\bm{\nu}}}^{\bot}\right) \right. \notag  \\
&~~ \times \left. \mathbf{\Pi}_{\mathbf{W}^{1/2}\mathbf{D}_L}^{\bot}  \mathbf{W}^{1/2}\mathbf{V}_{\bm{r}} \right\}^{-1},
\end{align}
and
\begin{align}
\mathrm{CRB}(\bm{\nu}) &= \frac{1}{L_r} (\mathbf{N}_{22}-\mathbf{N}_{21}\mathbf{N}_{11}^{-1}\mathbf{N}_{12})^{-1}                 \notag \\
&=\frac{1}{L_r} \left\{ \mathbf{V}_{\bm{\nu}}^H \mathbf{W}^{1/2} \mathbf{\Pi}_{\mathbf{W}^{1/2}\mathbf{D}_L}^{\bot}  \left(\mathbf{\Pi}_{\mathbf{\Pi}_{\mathbf{W}^{1/2}\mathbf{D}_L}^{\bot}  \mathbf{W}^{1/2}\mathbf{V}_{\bm{r}}}^{\bot}\right) \right. \notag  \\
&~~ \times \left. \mathbf{\Pi}_{\mathbf{W}^{1/2}\mathbf{D}_L}^{\bot}  \mathbf{W}^{1/2}\mathbf{V}_{\bm{\nu}} \right\}^{-1}.
\end{align}

The joint CRBs exist if and only if the FIM $\mathbf{J}(\bm{\gamma})$ is nonsingular.  From (\ref{eq:PA_CRB2}), the (positive definite) matrix $\mathbf{R}(\bm{\gamma}) \succ 0$, leading to $\mathbf{R}^{-1}(\bm{\gamma}) \succ 0  $. Then, we have $\mathbf{W}= \mathbf{R}^{-T}(\bm{\gamma}) \otimes \mathbf{R}^{-1}(\bm{\gamma}) \succ 0$, implying that $\mathbf{W}$ is nonsingular. 

For the necessary condition, i.e. if the CRBs exist, namely the $\mathbf{J}(\bm{\gamma})$ is nonsingular, we have
\begin{align}
\mathrm{rank}([\mathbf{D}_L~ \mathbf{D}_R])  &= \mathrm{rank}([\mathbf{D}_L ~ \mathbf{D}_R]^H \mathbf{W} [\mathbf{D}_L ~ \mathbf{D}_R])         \notag \\
&=\mathrm{rank}(\mathbf{J}(\bm{\gamma})/L_r)                                 \notag \\
&=4Q  .  
\end{align}

For the sufficient condition, i.e. if $\mathrm{rank}([\mathbf{D}_L~ \mathbf{D}_R]) = 4Q$, then
\begin{align}
\mathrm{rank}(\mathbf{J}(\bm{\gamma})/L_r) &= \mathrm{rank}([\mathbf{D}_L ~ \mathbf{D}_R]^H \mathbf{W} [\mathbf{D}_L ~ \mathbf{D}_R])        \notag \\
&=  \mathrm{rank}([\mathbf{D}_L ~ \mathbf{D}_R]                      \notag \\
&= 4Q,
\end{align}
which implies that $\mathbf{J}(\bm{\gamma})$ is nonsingular. This concludes the proof.

\section{Recovery Guarantees for Other L-shaped Arrays}
\label{app:perf_other_arrays}

\begin{corollary}                \label{Coro_otherFDAs}
  (UUC, UCU, UCC, CUU, CCU, and CUC FDAs) Consider an L-shaped FDA with $P_s=N_s+2M_s-1$ sensors along each axis. Each sensor transmits a total of $K=N_t+2M_t-1$ pulses in a CPI. The fundamental spatial spacing and fundamental PRI are $d$ and $T$, respectively. If\\
  for UUC: \upshape\textbf{C1}-\textbf{C3} \itshape and \upshape\textbf{C5}-\textbf{C6} \itshape hold,\\ 
  for UCU: \upshape\textbf{C1}-\textbf{C3} \itshape and \upshape\textbf{C6}-\textbf{C7} \itshape hold,\\
  for UCC: \upshape\textbf{C1}-\textbf{C3} \itshape and \upshape\textbf{C5}-\textbf{C6} \itshape hold,\\
  for CUU: \upshape\textbf{C1}-\textbf{C3} \itshape and \upshape\textbf{C6}-\textbf{C7} \itshape hold,\\ 
  for CCU: \upshape\textbf{C1}-\textbf{C4} \itshape and \upshape\textbf{C6} \itshape hold,\\
  for CUC: \upshape\textbf{C1}-\textbf{C3} \itshape and \upshape\textbf{C5}-\textbf{C6} \itshape hold,\\
  then the unknown parameter set $\bm{\gamma}_q=\{\theta_q, \varphi_q,r_q,\nu_q\}_{q=1}^{Q}$ of $Q$ far-field targets are perfectly recovered from $\mathbf{R}_{XZ}$ with the lower bounds on the number of physical sensor elements and the number of transmit pulses as, respectively,\\
    $\textrm{for UUC: } P_s > Q \textrm{\upshape{ and }} K > 2\sqrt{Q+1}-2$,\\
    $\textrm{for UCU: } P_s > Q \textrm{\upshape{ and }} K > Q$,\\
    $\textrm{for UCC: } P_s > Q \textrm{\upshape{ and }} K > 2\sqrt{Q+1}-2$,\\
    $\textrm{for CUU: } P_s > Q \textrm{\upshape{ and }} K > Q$,\\
    $\textrm{for CCU: } P_s > 2\sqrt{Q+1}-2 \textrm{\upshape{ and }} K > Q$,\\
    $\textrm{for CUC: } P_s > Q \textrm{\upshape{ and }} K > 2\sqrt{Q+1}-2$.
\end{corollary}
\begin{IEEEproof}
    For UUC, replace the virtual manifold matrices $\widetilde{\mathbf{A}}_x$, $\widetilde{\mathbf{A}}_z$, and $\widetilde{\mathbf{C}}$ with the corresponding physical manifold matrices $\mathbf{A}_x$, $\mathbf{A}_z$, and $\mathbf{C}$, respectively. For UCU, replace the virtual manifold matrices $\widetilde{\mathbf{A}}_x$, $\widetilde{\mathbf{A}}_z$, and $\widetilde{\mathbf{B}}$ with the corresponding physical manifold matrices $\mathbf{A}_x$, $\mathbf{A}_z$, and $\mathbf{B}$, respectively. For UCC, replace the virtual manifold matrices $\widetilde{\mathbf{A}}_x$ and $\widetilde{\mathbf{A}}_z$ with the corresponding physical manifold matrices $\mathbf{A}_x$ and $\mathbf{A}_z$, respectively.
    
    For CUU, replace the virtual manifold matrices $\widetilde{\mathbf{C}}$ and $\widetilde{\mathbf{B}}$ with the corresponding physical manifold matrices $\mathbf{C}$ and $\mathbf{B}$, respectively. For CCU, replace the virtual manifold matrix $\widetilde{\mathbf{B}}$ with the corresponding physical manifold matrix $\mathbf{B}$. For CUC, replace the virtual manifold matrix $\widetilde{\mathbf{C}}$ with the corresponding physical manifold matrix $\mathbf{C}$. 
    
    Then, \textit{ceteris paribus}, the result follows from repeating the steps of the proof in Theorem \ref{Theo_rec}. 
\end{IEEEproof}

\begin{corollary}                \label{Coro_UU}
  (U-U, U-C, C-U, and C-C) Consider an L-shaped array  with $P_s=N_s+2M_s-1$ sensors along each axis. Each sensor transmits a total of $K=N_t+2M_t-1$ pulses in a CPI. The fundamental spatial spacing and fundamental PRI are $d$ and $T$, respectively. If\\
  for U-U: \upshape\textbf{C1}-\textbf{C2} \itshape and \upshape\textbf{C6}-\textbf{C7} \itshape hold,\\ 
  for U-C: \upshape\textbf{C1}-\textbf{C2} \itshape and \upshape\textbf{C5}-\textbf{C6} \itshape hold,\\ 
  for C-U: \upshape\textbf{C1}-\textbf{C2} \itshape and \upshape\textbf{C4}-\textbf{C7} \itshape hold,\\
  for C-C: \upshape\textbf{C1}-\textbf{C2} \itshape and \upshape\textbf{C4}-\textbf{C5} \itshape hold,\\
  then the unknown parameter set $\bm{\gamma}_q=\{\theta_q, \varphi_q,\nu_q\}_{q=1}^{Q}$ of $Q$ far-field targets are perfectly recovered from $\mathbf{R}_{XZ}$ with the lower bounds on the number of physical sensor elements and the number of transmit pulses as, respectively,\\
    $\textrm{for U-U: } P_s > Q \textrm{\upshape{ and }} K > Q$,\\
    $\textrm{for U-C: } P_s > Q \textrm{\upshape{ and }} K > 2\sqrt{Q+1}-2$,\\
    $\textrm{for C-U: } P_s > 2\sqrt{Q+1}-2 \textrm{\upshape{ and }} K > Q$,\\
    $\textrm{for C-C: } P_s > 2\sqrt{Q+1}-2 \textrm{\upshape{ and }} K > 2\sqrt{Q+1}-2$.
\end{corollary}
\begin{IEEEproof}
    Since there are no FOs in the L-shaped arrays U-U, U-C, C-U and C-C, the matrix $\mathbf{B}$ does not exist in these cases. Hence, \textit{ceteris paribus}, the result follows from repeating the steps of the proofs of Theorem~\ref{Theo_rec}, Corollary~\ref{Coro_UUU}, and Corollary~\ref{Coro_otherFDAs}, without considering any matrix $\mathbf{B}$ in the computations. The results of U-U, U-C, C-U, and C-C then follow from that of U-Cube, UUC, CCU, and C-Cube, respectively. 
\end{IEEEproof}

\bibliographystyle{IEEEtran}
\bibliography{main}

\begin{thebibliography}{10}
\providecommand{\url}[1]{#1}
\csname url@samestyle\endcsname
\providecommand{\newblock}{\relax}
\providecommand{\bibinfo}[2]{#2}
\providecommand{\BIBentrySTDinterwordspacing}{\spaceskip=0pt\relax}
\providecommand{\BIBentryALTinterwordstretchfactor}{4}
\providecommand{\BIBentryALTinterwordspacing}{\spaceskip=\fontdimen2\font plus
\BIBentryALTinterwordstretchfactor\fontdimen3\font minus
  \fontdimen4\font\relax}
\providecommand{\BIBforeignlanguage}[2]{{%
\expandafter\ifx\csname l@#1\endcsname\relax
\typeout{** WARNING: IEEEtran.bst: No hyphenation pattern has been}%
\typeout{** loaded for the language `#1'. Using the pattern for}%
\typeout{** the default language instead.}%
\else
\language=\csname l@#1\endcsname
\fi
#2}}
\providecommand{\BIBdecl}{\relax}
\BIBdecl

\bibitem{RJMailloux2018}
R.~J. Mailloux, \emph{Phased array antenna handbook}, 3rd~ed.\hskip 1em plus
  0.5em minus 0.4em\relax Artech House, 2018.

\bibitem{galati1994advanced}
R.~P. Shenoy, ``Phased array antennas,'' in \emph{Advanced radar techniques and
  systems}, G.~Galati, Ed.\hskip 1em plus 0.5em minus 0.4em\relax Peter
  Peregrinus, 1993.

\bibitem{frank2008advanced}
J.~Frank and J.~D. Richards, ``Phased array radar antennas,'' in \emph{Radar
  handbook}, 3rd~ed., M.~I. Skolnik, Ed.\hskip 1em plus 0.5em minus 0.4em\relax
  McGraw-Hill Education, 2008.

\bibitem{NIGiannoccaro2012}
N.~I. Giannoccaro and L.~Spedicato, ``A new strategy for spatial reconstruction
  of orthogonal planes using a rotating array of ultrasonic sensors,''
  \emph{IEEE Sensors Journal}, vol.~12, no.~5, pp. 1307--1316, 2012.

\bibitem{SYang2012}
S.~Yong and J.~T. Bernhard, ``A pattern reconfigurable null scanning antenna,''
  \emph{IEEE Transactions on Antennas and Propagation}, vol.~60, no.~10, pp.
  4538--4544, 2012.

\bibitem{MElmer2012}
M.~Elmer, B.~D. Jeffs, K.~F. Warnick, J.~R. Fisher, and R.~D. Norrod,
  ``Beamformer design methods for radio astronomical phased array feeds,''
  \emph{IEEE Transactions on Antennas and Propagation}, vol.~60, no.~2, pp.
  903--914, February 2012.

\bibitem{PAntonik2006_1}
P.~Antonik, M.~C. Wicks, H.~D. Griffiths, and C.~J. Baker, ``Frequency diverse
  array radars,'' in \emph{IEEE Radar Conference}, 2006, pp. 215--217.

\bibitem{PAntonik2006_2}
------, ``Multi-mission multi-mode waveform diversity,'' in \emph{IEEE Radar
  Conference}, 2006, pp. 580--582.

\bibitem{PFSammartino2013}
P.~F. Sammartino, C.~J. Baker, and H.~D. Griffiths, ``Frequency diverse {MIMO}
  techniques for radar,'' \emph{IEEE Transactions on Aerospace and Electronic
  Systems}, vol.~49, no.~1, pp. 201--222, 2013.

\bibitem{WQWang2014}
W.~Q. Wang and H.~C. So, ``Transmit subaperturing for range and angle
  estimation in frequency diverse array radar,'' \emph{IEEE Transactions on
  Signal Processing}, vol.~62, no.~8, pp. 2000--2011, 2014.

\bibitem{WQWang2015}
W.~Q. Wang, ``Frequency diverse array antenna: {N}ew opportunities,''
  \emph{IEEE Antennas Propagation Magazine}, vol.~57, no.~2, pp. 145--152,
  2015.

\bibitem{MSecmen2007}
M.~Secmen, S.~Demir, A.~Hizal, and T.~Eker, ``Frequency diverse array antenna
  with periodic time modulated pattern in range and angle,'' in \emph{IEEE
  Radar Conference}, 2007, pp. 427--430.

\bibitem{JJHuang2009}
J.~J. Huang, K.-F. Tong, and C.~J. Baker, ``Frequency diverse array:
  {S}imulation and design,'' in \emph{IEEE Radar Conference}, 2009, pp. 1--4.

\bibitem{THiggins2009}
T.~Higgins and S.~Blunt, ``Analysis of range-angle coupled beamforming with
  frequency diverse chirps,'' in \emph{IEEE International Waveform Diversity
  and Design Conference}, 2009, pp. 140--144.

\bibitem{WQWang2014_2}
W.~Q. Wang and H.~Shao, ``Range-angle localization of targets by a double pulse
  frequency diverse array radar,'' \emph{IEEE Journal of Selected Topics in
  Signal Processing}, vol.~8, no.~1, pp. 106--114, 2014.

\bibitem{JXu2015}
J.~Xu, G.~Liao, S.~Zhu, L.~Huang, and H.~C. So, ``Joint range and angle
  estimation using {MIMO} radar with frequency diverse array,'' \emph{IEEE
  Transactions on Signal Processing}, vol.~63, no.~13, pp. 3396--3410, 2015.

\bibitem{RGui2018}
R.~Gui, X.~Wang, Y.~Pan, and J.~Xu, ``Cognitive target tracking via
  angle-range-{D}oppler estimation with transmit subaperturing {FDA} radar,''
  \emph{IEEE Journal of Selected Topics in Signal Processing}, vol.~12, no.~1,
  pp. 76--89, 2018.

\bibitem{FLiu2021}
F.~Liu, X.~Wang, M.~Huang, and L.~Wan, ``Joint angle and range estimation for
  bistatic {FDA-MIMO} radar via real-valued subspace decomposition,''
  \emph{Signal Processing}, vol. 185, p. 108065, 2021.

\bibitem{agrawal1972mutual}
V.~D. Agrawal and Y.~T. Lo, ``Mutual coupling in phased arrays of randomly
  spaced antennas,'' \emph{IEEE Transactions on Antennas and Propagation},
  vol.~20, no.~3, pp. 288--295, 1972.

\bibitem{lo1964mathematical}
Y.~T. Lo, ``A mathematical theory of antenna arrays with randomly spaced
  elements,'' \emph{IEEE Transactions on Antennas and Propagation}, vol.~12,
  no.~3, pp. 257--268, 1964.

\bibitem{PPVaidyanathan2011}
P.~P. Vaidyanathan and P.~Pal, ``Sparse sensing with co-prime samplers and
  arrays,'' \emph{IEEE Transactions on Signal Processing}, vol.~59, no.~2, pp.
  573--586, 2011.

\bibitem{SQin2015}
S.~Qin, Y.~D. Zhang, and M.~G. Amin, ``Generalized coprime array configurations
  for direction-of-arrival estimation,'' \emph{IEEE Transactions on Signal
  Processing}, vol.~63, no.~6, pp. 1377--1390, 2015.

\bibitem{WZheng2021}
W.~Zheng, X.~Zhang, J.~Li, and J.~Shi, ``Extensions of co-prime array for
  improved {DOA} estimation with hole filling strategy,'' \emph{IEEE Sensors
  Journal}, vol.~21, no.~5, pp. 6724--6732, 2021.

\bibitem{khan2014frequency}
W.~Khan, I.~M. Qureshi, and S.~Saeed, ``Frequency diverse array radar with
  logarithmically increasing frequency offset,'' \emph{IEEE antennas and
  wireless propagation letters}, vol.~14, pp. 499--502, 2014.

\bibitem{basit2017beam}
A.~Basit, I.~M. Qureshi, W.~Khan, S.~ur~Rehman, and M.~M. Khan, ``Beam pattern
  synthesis for an {FDA} radar with hamming window-based nonuniform frequency
  offset,'' \emph{IEEE Antennas and Wireless Propagation Letters}, vol.~16, pp.
  2283--2286, 2017.

\bibitem{liu2016random}
Y.~Liu, H.~Ruan, L.~Wang, and A.~Nehorai, ``The random frequency diverse array:
  {A} new antenna structure for uncoupled direction-range indication in active
  sensing,'' \emph{IEEE Journal of Selected Topics in Signal Processing},
  vol.~11, no.~2, pp. 295--308, 2016.

\bibitem{SQin2017}
S.~Qin, Y.~D. Zhang, Y.~Pan, M.~G. Amin, and F.~Gini, ``Frequency diverse
  coprime arrays with coprime frequency offsets for multitarget localization,''
  \emph{IEEE Journal of Selected Topics in Signal Processing}, vol.~11, no.~2,
  pp. 321--335, 2017.

\bibitem{RCao2021}
R.~Cao, S.~Liu, Z.~Mao, and Y.~Huang, ``Doubly-{T}oeplitz-based interpolation
  for joint {DoA}-range estimation using coprime {FDA},'' in \emph{IEEE Radar
  Conference}, 2021, pp. 1--6.

\bibitem{LLiu2022}
L.~Liu, S.~Liu, Y.~Huang, and M.~G. Amin, ``Joint doa-range estimation using
  moving time-modulated frequency diverse coprime array,'' in \emph{IEEE Radar
  Conference}, 2022, pp. 1--5.

\bibitem{JWang2019}
J.~Wang and S.~Jiang, ``Angle-polarization-range estimation using sparse
  polarization sensitive {FDA-MIMO} radar with co-prime frequency offsets,''
  \emph{IEEE Access}, vol.~7, pp. 46\,759--14\,677, 2019.

\bibitem{CWang2020}
C.~Wang, Z.~Li, and X.~Zhang, ``{FDA-MIMO} for joint angle and range
  estimation: {U}nfolded coprime framework and parameter estimation
  algorithm,'' \emph{IET Radar, Sonar \& Navigation}, vol.~14, no.~6, pp.
  917--926, 2020.

\bibitem{sedighi2019optimum}
S.~Sedighi, B.~Shankar, K.~V. Mishra, and B.~Ottersten, ``Optimum design for
  sparse {FDA-MIMO} automotive radar,'' in \emph{Asilomar Conference on
  Signals, Systems, and Computers}, 2019, pp. 913--918.

\bibitem{XLi2018}
X.~Li, D.~Wang, W.~Wang, W.~Liu, and X.~Ma, ``Range-angle localization of
  targets with planar frequency diverse subaperturing {MIMO} radar,''
  \emph{IEEE Access}, vol.~6, pp. 12\,505--12\,517, 2018.

\bibitem{CWang2021}
C.~Wang, X.~Zhang, and J.~Li, ``{FDA-MIMO} radar for {3D} localization:
  {V}irtual coprime planar array with unfolded coprime frequency offset
  framework and {TRD-MUSIC} algorithm,'' \emph{Digital Signal Processing}, vol.
  113, p. 103017, 2021.

\bibitem{TKishigami2016}
T.~Kishigami, H.~Yomo, A.~Matsuoka, and J.~Satou, ``Millimeter-wave {MIMO}
  radar system using {L}-shaped {Tx} and {Rx} arrays,'' in \emph{European Radar
  Conference}, 2016, pp. 29--32.

\bibitem{AMElbir2020}
A.~M. Elbir, ``{L}-shaped coprime array structures for {DOA} estimation,''
  \emph{Multidimensional Systems and Signal Processing}, vol.~31, pp. 205--219,
  2020.

\bibitem{na2018tendsur}
S.~Na, K.~V. Mishra, Y.~Liu, Y.~C. Eldar, and X.~Wang, ``{TenDSuR}:
  {T}ensor-based {4D} sub-{N}yquist radar,'' \emph{IEEE Signal Processing
  Letters}, vol.~26, no.~2, pp. 237--241, 2018.

\bibitem{Vijay_Random_Sparse_Step_Frequency_2020}
K.~V. Mishra, S.~Mulleti, and Y.~C. Eldar, ``Ra{SS}te{R}: Random sparse
  step-frequency radar,'' \emph{arXiv preprint arXiv:2004.05720}, 2020.

\bibitem{Maier_Non_Uniform_PRI_1993}
M.~W. Maier, ``Non-uniform {PRI} pulse-{D}oppler radar,'' in \emph{Southeastern
  Symposium on System Theory}, 1993, pp. 164--168.

\bibitem{xu2021difference}
L.~Xu, S.~Sun, and K.~V. Mishra, ``Difference co-chirps-based non-uniform {PRF}
  automotive {FMCW} radar,'' in \emph{IEEE International Conference on
  Autonomous Systems}, 2021, pp. 1--5.

\bibitem{xu2016adaptive}
J.~Xu, G.~Liao, Y.~Zhang, H.~Ji, and L.~Huang, ``An adaptive
  range-angle-doppler processing approach for {FDA-MIMO} radar using
  three-dimensional localization,'' \emph{IEEE Journal of Selected Topics in
  Signal Processing}, vol.~11, no.~2, pp. 309--320, 2016.

\bibitem{skolnik2008radar}
M.~Skolnik, \emph{Radar handbook}, 3rd~ed.\hskip 1em plus 0.5em minus
  0.4em\relax McGraw-Hill, 2008.

\bibitem{JXu2015_STAP}
J.~Xu, S.~Zhu, and G.~Liao, ``Range ambiguous clutter suppression for airborne
  {FDA-STAP} radar,'' \emph{IEEE Journal of Selected Topics in Signal
  Processing}, vol.~9, no.~8, pp. 1620--1631, 2015.

\bibitem{EBouDaher2015}
E.~BouDaher, Y.~Jia, F.~Ahmad, and M.~G. Amin, ``Multi-frequency co-prime
  arrays for high-resolution direction-of-arrival estimation,'' \emph{IEEE
  Transactions on Signal Processing}, vol.~63, no.~14, pp. 3797--3808, 2015.

\bibitem{ZMao2022}
Z.~Mao, S.~Liu, Y.~D. Zhang, L.~Han, and Y.~Huang, ``Joint {DoA}-range
  estimation using space-frequency virtual difference coarray,'' \emph{IEEE
  Transactions on Signal Processing}, vol.~70, pp. 2576--2592, 2022.

\bibitem{LLiu2018}
L.~Liu and P.~Wei, ``Joint {DOA} and frequency estimation with sub-{N}yquist
  sampling in the sparse array system,'' \emph{IEEE Signal Processing Letters},
  vol.~25, no.~9, pp. 1285--1289, 2018.

\bibitem{FWang2018}
F.~Wang, J.~Fang, H.~Duan, and H.~Li, ``Phased-array-based sub-{N}yquist
  sampling for joint wideband spectrum sensing and direction-of-arrival
  estimation,'' \emph{IEEE Transactions on Signal Processing}, vol.~66, no.~23,
  pp. 6110--6123, 2018.

\bibitem{ZZhang2021}
Z.~Zhang, P.~Wei, H.~Zhang, and L.~Deng, ``Joint spectrum sensing and {DOA}
  estimation with sub-{N}yquist sampling,'' \emph{Signal Processing}, vol.
  2021, no. 189, pp. 1--10, 2021.

\bibitem{ZPeng2020}
Z.~Peng, W.~Si, and F.~Zeng, ``{2-D} {DOA} estimation for {L}-shaped sparse
  array via joint use of spatial and temporal information,'' \emph{IEEE
  Communications Letters}, vol.~24, no.~9, pp. 1934--1938, 2020.

\bibitem{ZLi2021}
Z.~Li, X.~Zhang, and J.~Shen, ``{2D-DOA} estimation of strictly noncircular
  sources utilizing connection-matrix for {L}-shaped array,'' \emph{IEEE
  Wireless Communications Letters}, vol.~10, no.~2, pp. 296--300, 2021.

\bibitem{CLi2021}
C.~Li, Q.~Wang, H.~Chen, and L.~Teng, ``1-bit {DOA} estimation using
  expectation-maximization generalized approximate message passing with two
  {L}-shaped arrays,'' \emph{IEEE Communications Letters}, vol.~25, no.~8, pp.
  2554--2558, 2021.

\bibitem{SSIoushua2017}
S.~S. Ioushua, O.~Yair, D.~Cohen, and Y.~Eldar, ``{CaSCADE}: {C}ompressed
  carrier and {DOA} estimation,'' \emph{IEEE Transactions on Signal
  Processing}, vol.~65, no.~10, pp. 2645--2658, 2017.

\bibitem{MWang2017}
M.~Wang and A.~Nehorai, ``Coarrays, {MUSIC}, and the {C}ram\'{e}r–{R}ao
  bound,'' \emph{IEEE Transactions on Signal Processing}, vol.~65, no.~4, pp.
  933--946, 2017.

\bibitem{ZCHao1990}
H.~Z. Chuan, ``Note on the inequality of the arithmetic and geometric means,''
  \emph{Pacific Journal of Mathematics}, vol. 143, no.~1, pp. 43--46, 1990.

\bibitem{TSvantesson1999}
T.~Svantesson, ``Modeling and estimation of mutual coupling in a uniform linear
  array of dipoles,'' in \emph{IEEE International Conference on Acoustics,
  Speech, and Signal Processing}, vol.~5, 1999, pp. 2961--2964.

\bibitem{CLiu2016}
C.~Liu and P.~P. Vaidyanathan, ``Super nested arrays: {L}inear sparse arrays
  with reduced mutual coupling — {P}art {I}: {F}undamentals,'' \emph{IEEE
  Transactions on Signal Processing}, vol.~64, no.~15, pp. 3997--4012, 2016.

\bibitem{pal2010nested}
P.~Pal and P.~P. Vaidyanathan, ``Nested arrays: {A} novel approach to array
  processing with enhanced degrees of freedom,'' \emph{IEEE Transactions on
  Signal Processing}, vol.~58, no.~8, pp. 4167--4181, 2010.

\bibitem{RRobin2017}
R.~Rajam\"{a}ki and V.~Koivunen, ``Sparse linear nested array for active
  sensing,'' in \emph{European Signal Processing Conference}, 2017, pp.
  1976--1980.

\bibitem{JShi2018}
J.~Shi, G.~Hu, X.~Zhang, and H.~Zhou, ``Generalized nested array:
  {O}ptimization for degrees of freedom and mutual coupling,'' \emph{IEEE
  Communications Letters}, vol.~22, no.~6, pp. 1208--1211, 2018.

\bibitem{KJames2014}
J.~D. Krieger, Y.~Kochman, and G.~W. Wornell, ``Multi-coset sparse imaging
  arrays,'' \emph{IEEE Transactions on Antennas and Propagation}, vol.~62,
  no.~4, pp. 1701--1715, 2014.

\bibitem{YLiao2019}
Y.~Liao, W.~Q. Wang, and Z.~Zheng, ``Frequency diverse array beampattern
  synthesis using symmetrical logarithmic frequency offsets for target
  indication,'' \emph{IEEE Transactions on Antennas and Propagation}, vol.~67,
  no.~5, pp. 3505--3509, 2019.

\bibitem{WKhan2014}
W.~Khan and I.~Qureshi, ``Frequency diverse array radar with time-dependent
  frequency offset,'' \emph{IEEE Antennas and Wireless Propagation Letters},
  vol.~13, pp. 758--761, 2014.

\bibitem{GQuan2016}
G.~Quan, Z.~Yang, J.~Huang, and J.~Huang, ``Sparsity-based space-time adaptive
  processing in random pulse repetition frequency and random arrays radar,'' in
  \emph{IEEE International Conference on Signal Processing}, 2016, pp.
  1642--1646.

\bibitem{XWang2018}
X.~Wang, Z.~Yang, and J.~Huang, ``Sparsity-based space-time adaptive processing
  for airborne radar with coprime array and coprime pulse repetition
  interval,'' in \emph{IEEE International Conference on Acoustics, Speech and
  Signal Processing}, 2018, pp. 3310--3314.

\bibitem{LXu2018}
L.~Xu, R.~Wu, X.~Zhang, and Z.~Shi, ``Joint two-dimensional {DOA} and frequency
  estimation for {L}-shaped array via compressed sensing {PARAFAC} method,''
  \emph{IEEE Access}, vol.~6, pp. 37\,204--37\,213, 2018.

\bibitem{kozick1992coarray}
R.~Kozick and S.~Kassam, ``Coarray synthesis with circular and elliptical
  boundary arrays,'' \emph{IEEE Transactions on Image Processing}, vol.~1,
  no.~3, pp. 391--405, 1992.

\bibitem{arsac1955nouveau}
J.~Arsac and A.~Danjon, ``Nouveau r\'{o}seau pour l'observation radio
  astronomique de labrillance sur le soleil \`{a} 9359 {MC},'' \emph{Comptes
  Rendus Hebdomadaires Des Seances De L Academie Des Sciences}, vol. 240,
  no.~9, pp. 942--945, 1955, in French.

\bibitem{camps2001synthesis}
A.~Camps, A.~Cardama, and D.~Infantes, ``Synthesis of large low-redundancy
  linear arrays,'' \emph{IEEE Transactions on Antennas and Propagation},
  vol.~49, no.~12, pp. 1881--1883, 2001.

\bibitem{leech1956representation}
J.~Leech, ``On the representation of $1$, $2$, $\ldots$, $n$ by differences,''
  \emph{Journal of the London Mathematical Society}, vol.~1, no.~2, pp.
  160--169, 1956.

\bibitem{moffet1968minimum}
A.~Moffet, ``Minimum-redundancy linear arrays,'' \emph{IEEE Transactions on
  Antennas and Propagation}, vol.~16, pp. 172--175, 1968.

\bibitem{chen2012one}
L.~Chen, Q.~Li, G.~Yi, and Y.~Zhu, ``One-dimensional mirrored interferometric
  aperture synthesis: {P}erformances, simulation, and experiments,'' \emph{IEEE
  Transactions on Geoscience and Remote Sensing}, vol.~51, no.~5, pp.
  2960--2968, 2012.

\bibitem{JBKruskal1977}
J.~B. Kruskal, ``Three-way arrays: {R}ank and uniqueness of trilinear
  decompositions, with application to arithmetic complexity and statistics,''
  \emph{Linear Algebra and its Applications}, vol.~18, no.~2, pp. 95--138,
  1977.

\bibitem{NDSidiropoulos2000}
N.~D. Sidiropoulos and R.~Bro, ``On the uniqueness of multilinear decomposition
  of {N}-way arrays,'' \emph{Journal of Chemometrics}, vol.~14, no.~3, pp.
  229--239, 2000.

\bibitem{HornRA1985}
R.~A. Horn and C.~R. Johnson, \emph{Matrix Analysis}.\hskip 1em plus 0.5em
  minus 0.4em\relax Cambridge Univ. Press,, 1985.

\end{thebibliography}

\end{document}